\documentclass[11pt,a4paper,reqno]{amsart}%
\usepackage{amsthm,amsmath,amsfonts,amssymb,amsxtra,appendix,bookmark,dsfont,latexsym}
\usepackage{amsfonts,mathrsfs}

\makeatletter
\usepackage{hyperref}
\usepackage{delarray,a4,color}
\usepackage{euscript}
\usepackage{bbm}
\usepackage{amsfonts}
\usepackage[latin1]{inputenc}

\theoremstyle{plain}
\newtheorem{theorem}{Theorem}[section]
\newtheorem{lemma}[theorem]{Lemma}

\newtheorem{proposition}[theorem]{Proposition}

\newtheorem{assumption}{Assumption}[section]
\newtheorem{definition}[theorem]{Definition}

\theoremstyle{remark}
\newtheorem{remark}[theorem]{Remark}

\numberwithin{equation}{section}

\DeclareMathOperator{\Tr}{Tr}
\DeclareMathOperator{\tr}{Tr}

\def\geqslant{\ge}
\def\leqslant{\le}
\def\bq{\begin{eqnarray}}
\def\eq{\end{eqnarray}}
\def\bqq{\begin{eqnarray*}}
\def\eqq{\end{eqnarray*}}
\def\nn{\nonumber}

\def\eps{\varepsilon}

\newcommand{\norm}[1]{\left\lVert #1 \right\rVert}

\newcommand\1{{\ensuremath {\mathds 1} }}

\newcommand{\brar}{\right|}
\newcommand{\bral}{\left\langle}
\newcommand{\ketr}{\right\rangle}
\newcommand{\ketl}{\left|}

\newcommand{\im}{\mathrm{i}}
\newcommand{\Emp}{\mathrm{Emp}}

\renewcommand{\epsilon}{\varepsilon}

\def\R {\mathbb{R}}

\def\cP {\mathcal{P}}
\def\E {\mathcal{E}}

\def\F {\mathcal{F}}

\def\R {\mathbb{R}}

\def\E {\mathcal{E}}

\def\d{{\rm d}}

\def\b{|}

\renewcommand{\leq}{\leqslant}
\renewcommand{\geq}{\geqslant}

\newcommand{\mut}{\widetilde{\mu}}

\newcommand{\bA}{\mathbf{A}}
\newcommand{\bB}{\mathbf{B}}

\newcommand{\curl}{\mathrm{curl}}

\newcommand{\gammat}{\widetilde{\gamma}}
\newcommand{\rhot}{\widetilde{\rho}}
\newcommand{\Ave}{\mathrm{Ave}}

\title[Semiclassical limit for almost fermionic anyons]{Semiclassical limit for almost fermionic anyons}

\author[T.Girardot]{Th\'{e}otime Girardot}
\address{Universit\'e Grenoble Alpes \& CNRS, LPMMC (UMR 5493)}
\email{theotime.girardot@lpmmc.cnrs.fr}
\author[N.Rougerie]{Nicolas Rougerie}
\address{Ecole Normale Sup\'erieure de Lyon \& CNRS, UMPA (UMR 5669)}
\email{nicolas.rougerie@ens-lyon.fr}

\date{January, 2021}

\begin{document}

\maketitle

\begin{abstract}
In two-dimensional space there are possibilities for quantum statistics continuously interpolating between the bosonic and the fermionic one. Quasi-particles obeying such statistics can be described as ordinary bosons and fermions with magnetic interactions. We study a limit situation where the statistics/magnetic interaction is seen as a ``perturbation from the fermionic end''. We vindicate a mean-field approximation,  proving that the ground state of a gas of anyons is described to leading order by a semi-classical, Vlasov-like, energy functional. The ground state of the latter displays anyonic behavior in its momentum distribution. Our proof is based on coherent states, Husimi functions, the Diaconis-Freedman theorem and a quantitative version of a semi-classical Pauli pinciple. 
\end{abstract}

\setcounter{tocdepth}{2}

 
\makeatletter
\def\@tocline#1#2#3#4#5#6#7{\relax
  \ifnum #1>\c@tocdepth 
  \else
    \par \addpenalty\@secpenalty\addvspace{#2}%
    \begingroup \hyphenpenalty\@M
    \@ifempty{#4}{%
      \@tempdima\csname r@tocindent\number#1\endcsname\relax
    }{%
      \@tempdima#4\relax
    }%
    \parindent\z@ \leftskip#3\relax \advance\leftskip\@tempdima\relax
    \rightskip\@pnumwidth plus4em \parfillskip-\@pnumwidth
    #5\leavevmode\hskip-\@tempdima
      \ifcase #1
       \or\or \hskip 1em \or \hskip 2em \else \hskip 3em \fi%
      #6\nobreak\relax
      \dotfill
      \hbox to\@pnumwidth{\@tocpagenum{#7}}
    \par
    \nobreak
    \endgroup
  \fi}
\makeatother
\tableofcontents


\section{Model and main results}

\subsection{Introduction}

In two dimensional systems there are possibilities for quantum statistics other than the bosonic one and the fermionic one, so called intermediate or fractional statistics. Particles following such statistics, termed anyons (as in \emph{any}thing in between bosons and fermions), can arise as effective quasi-particles in correlated many-body quantum systems confined to two dimensions. They have been conjectured~\cite{AroSchWil-84,LunRou-16,Rougerie-hdr} to be relevant 
for the fractional quantum Hall effect~(see~\cite{Goerbig-09,Laughlin-99,Jain-07} for reviews and~\cite{BarEtalFev-20,NakEtalMan-20} for recent experimental developments). They can also be simulated in certain cold atoms systems with synthetic gauge fields~\cite{EdmEtalOhb-13,ClaEtalChi-18,ValWesOhb-20,YakEtal-19,YakLem-18,CorDubLunRou-19,Bar-20,nak-20}.

The statistics of anyonic particles is labeled by a single parameter $\alpha\in[0,2[$ where the cases $\alpha = 0,1$ correspond to ordinary bosons and fermions, respectively. Since the many-anyon problem is not exactly soluble (even in the absence of interactions other than the statistical ones) except for the latter special values, it makes sense to study limiting parameter regimes where reduced models are obtained as effective descriptions. 

The ``almost bosonic'' limit $\alpha \to 0$ is studied in \cite{Girardot-19,LunRou-15}, in the joint limit of the particle number $N\to \infty,\alpha \sim N^{-1}$, which turns out to be the relevant one for the non-trivial statistics to have an effect at leading order~\cite{CorLunRou-16,CorDubLunRou-19}. Here we study an ``almost fermionic limit'' $\alpha \to 1$ in dependence with the particle number. The relevant scaling in dependence of $N$ leads to a problem resembling the usual mean-field/semi-classical limit for interacting fermions~\cite{LieSim-77b,FouLewSol-15,BenPorSch-14,BarGolGotMau-03,ElgErdSchYau-04,FroKno-11}.


Our starting point is a gas of anyonic particles evolving in $\R^{2}$ subjected to a magnetic field 
$$\bB_{e}=\curl\bA_{e}$$
and a trapping external potential 
$$ 
V(x) \underset{|x|\to \infty}{\to} \infty.
$$
To model anyonic behavior we introduce the statistical gauge vector potential felt by  particle $j$ due to the influence of all the other particles
\begin{equation}
\bA(x_{j})=\sum_{j\neq k}\frac{(x_{j}-x_{k})^{\perp}}{|x_{j}-x_{k}|^{2}}.
\label{gauge_vector}
\end{equation}
The associated magnetic field is
\begin{equation}
\curl \bA(x_{j})=2\pi\sum_{k\neq j}\delta(x_{j}-x_{k}).
\label{flux}
\end{equation}
The statistics parameter is denoted by $\alpha $, 
corresponding to a statistical phase $e^{i\pi\left (1+\alpha\right )}$ under a continuous
simple interchange of two particles. In this so-called ``magnetic gauge picture'', 2D anyons are thus described 
as fermions, each of them carrying an Aharonov-Bohm magnetic flux of strength $\alpha $. Instead of looking at the ground state of an non-interacting Hamiltonian acting on anyonic wave-functions we thus study the following Hamiltonian
\begin{equation}
 H_{N}=\sum_{j=1}^{N}\left((-\im\hbar\nabla_{j}+\bA_{e}(x_{j})+\alpha\bA(x_{j}))^{2}+V(x_{j})\right)\;\;\text{acting on}\;\; L^{2}_{\mathrm{asym}}(\mathbb{R}^{2N}).
 \label{Ham2}
\end{equation} 
Note that in the above we view the anyons statistics parameter $\alpha$ counting from the fermionic representation of anyons. The limit $\alpha \to 0$ for~\eqref{Ham2} is then formally equivalent to $\alpha\to 1$ in the bosonic representation alluded to above~\cite{Girardot-19,LunRou-15}. 

We are interested in a joint limit $N\to\infty$, $\alpha\to 0$ in which the problem would simplify but still display signatures of anyonic behavior, i.e. a non-trivial dependence on $\alpha$. Since we perturb around fermions, the Pauli principle implies that the kinetic energy grows faster than $N$. This behavior is encoded in the Lieb-Thirring inequality \cite{LieThi-76,LieThi-75} which, for any $N$-particle normalized antisymmetric wave function $\Psi_N$ with support in a bounded domain $\Omega^{N}$, gives
\begin{equation}
\sum_{j=1}^{N}\int_{\Omega^{N}}\b\nabla_{j}\Psi\b^{2}\geqslant C\b\Omega\b^{-\frac{2}{d}}N^{1+\frac{2}{d}}
\label{LT1}
\end{equation}
where $d$ is the dimension of the physical space. Because of this rapid growth of the kinetic energy, the natural mean-field limit~\cite{BenPorSch-14,BenPorSch-15,FouLewSol-15} for fermions involves a scaling $\hbar \sim N^{-1/d}$. Here with $d=2$ we get a kinetic energy proportional to $N^{2}$ and thus, in expectation value $-\im \nabla_j\sim N^{1/2}$. Investigating the $N$-dependence of each term of $\eqref{Ham2}$ we thus expect for any $\Psi_{N}\in L^{2}\left (\R^{2N}\right )$
\begin{equation}
\bral \Psi_{N},H_{N}\Psi_{N}\ketr\approx C_1 N \left(\hbar N^{1/2}+\alpha N\right)^2 + C_2 N\nn
\end{equation}
where the terms in parenthesis accounts for kinetic energy, including the influence of~\eqref{gauge_vector}, and the second term is the energy in the external potential $V$. To obtain a non trivial limit when $\alpha \to 0 $, a natural choice is to take
\begin{equation}
\hbar= \frac{1}{\sqrt{N}} \;\;\text{and}\;\; \alpha = \frac{\beta}{N}\label{regime}
\end{equation} 
$\beta \in \R$ a constant allowing us to keep a track of the anyonic behavior. We shall consider the limit  
$$\lim_{N\to \infty} \frac{E\left (N\right )}{N}$$
with 
\begin{equation}
E(N):=\inf \left\{ \bral\Psi_N, H_{N}\Psi_N \ketr , \Psi_N \in L^{2}_{\mathrm{asym}}(\mathbb{R}^{2N}), \int_{\R^{2N}} |\Psi_N|^2 = 1 \right\}.
\label{en}
\end{equation}
Such an approach combines two limits (see also Remark~\ref{rem:scaling} below), the semi-classical $\hbar\to 0$ and the quasi-fermionic $\alpha\to 0$.

\subsection{Extended anyons: energy convergence}
The Hamiltonian \eqref{Ham2} is actually too singular. Consequently we introduce a length $R$ over which the magnetic flux attached to our particles is smeared. In our approach, $R\to0$ will make us recover the point-like anyons point of view. This ``extended anyons'' model is discussed in \cite{Mashkevich-96, Trugenberger-92b, ChoLeeLee-92}. Anyons   arising as quasi-particles in condensed-matter systems are not point-like objects. The radius $R$ then corresponds to their size of the quasi-particle, i.e. a length-scale below which the quasi-particle description breaks down. 

We introduce the 2D Coulomb potential generated by a unit charge smeared over the disc of radius $R$
\begin{equation}
 w_{R}(x)=\left(\log\b \;.\;\b *\chi_{R}\right)(x),\:\:\text{with the convention}\:\:w_{0}=\log\b \;.\;\b
 \label{wr}
\end{equation}
and $\chi_{R}(x)$ a positive smooth function of unit mass
\begin{equation}
\chi_{R}=\frac{1}{R^{2}}\chi\left (\frac{.}{R}\right)\; \text{and}\; \int_{\mathbb{R}^{2}}\chi=1,\;\;\;
 \chi( x)= 
   \begin{cases}
1/\pi^{2}\;\; \b x\b \le 1  \\
0\;\;\b x\b \ge 2.
   \end{cases}
   \label{khiR}
\end{equation}
We take $\chi$ to be smooth, positive and decreasing between $1$ and $2$.
We will recover the magnetic field $\eqref{flux}$ (in a distributional sense) in the limit $R\to 0$ by defining
\begin{equation}
 \bA^{R}(x_{j})=\sum_{k\neq j}\nabla^{\perp}w_{R}(x_{j}-x_{k}).
 \label{AJR}
\end{equation}
We henceforth work with the regularized Hamiltonian of $N$ anyons of radius $R$
\begin{equation}
 \boxed{H_{N}^{R}:=\sum_{j=1}^{N}\left(\left(p^{\bA}_{j}+\alpha \bA^{R}(x_{j})\right)^{2}+V(x_{j})\right)}
 \label{HRN}
\end{equation}
denoting
\begin{equation}\label{eq:momentum}
p^{\bA}_{j}=-\im\hbar\nabla_{j}+\bA_{e}(x_{j}). 
\end{equation}
It is essentially self-adjoint on $L^{2}_{\mathrm{asym}}(\mathbb{R}^{2N})$ (see \cite[Theorem X.17]{ReeSim2} and \cite{AvrHerSim-78}). The bottom of its spectrum then exists for any fixed $R>0$, we denote it 
\begin{equation}\label{eq:GSE}
E^{R}(N)=\inf \sigma\left(H^{R}_{N} \right). 
\end{equation}
Under our assumptions $H^{R}_{N}$ will actually have compact resolvent, so the above is an eigenvalue. However, this operator does not have a unique limit as $R\to 0$ and the Hamiltonian at $R=0$ is \emph{not} essentially self-adjoint, see for instance \cite{LunSol-14,CorOdd-18,DabSto-98,AdaTet-98,BouSor-92}. Nevertheless, in the joint limit $R\to 0$ and $N\to \infty $ we recover a unique well-defined (non-linear) model.
It is convenient to expand $\eqref{HRN}$ and treat summands separately
\begin{align}
 H_{N}^{R}&=\sum_{j=1}^{N}\left((p^{\bA}_{j})^{2}+V(x_{j})\right) 
\;\text{``Kinetic and potential terms"}\nn\\
 &+\alpha\sum_{j\neq k}\left(p^{\bA}_{j}\cdot\nabla^{\perp}w_{R}(x_{j}-x_{k})+\nabla^{\perp}w_{R}(x_{j}-x_{k})\cdot p^{\bA}_{j}\right)\;\text{``Mixed two-body term"}\nonumber\\
 &+\alpha^{2}\sum_{j\neq k\neq l}\nabla^{\perp}w_{R}(x_{j}-x_{k})\cdot \nabla^{\perp}w_{R}(x_{j}-x_{l})\;\text{``Three-body term"}\nonumber\\
 &+\alpha^{2}\sum_{j\neq k}\left|\nabla^{\perp}w_{R}(x_{j}-x_{k})\right|^{2}\;\text{``Singular two-body term"}.
 \label{expanded_H}
 \end{align}
The fourth term of the above, being $N$ times smaller than the others in the regime~\eqref{regime}, will easily be estimated. Further heuristics will only involve the mixed two-body term and the three-body term.
 These being of the same order, we expect (as in other types of mean-field limits, see e.g.~\cite{Rougerie-EMS} and references therein) that, for large $N$, particles behave independently, and thus
$$\sum_{j\neq k}\nabla^{\perp}w_{R}(x_{j}-x_{k})\approx \int_{\mathbb{R}^{2}}\nabla^{\perp}w_{R}(x_{j}-y)\rho(y)\mathrm{d}y=\nabla^{\perp}w_{R}*\rho(x_{j})$$
 where $\rho$ is the density of particles. Therefore, it is convenient to define, given a one-body density $\rho $ normalized in $L^{1}(\mathbb{R}^{2})$
 \begin{equation}
 \bA[\rho]=\nabla^{\perp}w_{0}*\rho\;\;\text{and}\;\;\bA^{R}[\rho]=\nabla^{\perp}w_{R}*\rho.\nn
\end{equation}
Natural mean-field functionals associated with~\eqref{expanded_H} are then
\begin{align}
\E^{\mathrm{af}}_{R}[\gamma]&=\frac{1}{N}\Tr\left [\left (p_{1}^{\bA}+\alpha\bA^{R}[\rho_{\gamma }]\right )^{2}\gamma\right ]+\frac{1}{N}\int_{\R^{2}}V(x_{1})\rho_{\gamma}(x_{1})\d x_{1}
\label{EAF}\\
\E^{\mathrm{af}}[\gamma]&=\frac{1}{N}\Tr\left [\left (p_{1}^{\bA}+\alpha\bA[\rho_{\gamma }]\right )^{2}\gamma\right ]+\frac{1}{N}\int_{\R^{2}}V(x_{1})\rho_{\gamma}(x_{1})\d x_{1}
\label{EAF0}
\end{align}
where the argument is a trace-class operator
\begin{equation}
\gamma\in \mathfrak{S}^{1}\left (L^{2}\left (\R^{2}\right )\right ), 0\leqslant \gamma \leqslant 1, \tr\gamma =1. 
\end{equation}
Similar Hartree-like functionals were obtained in the almost-bosonic limit \cite{Girardot-19,LunRou-15}. Here we simplify things one step further by considering a semi-classical analogue 
\begin{equation}
\E_{\mathrm{Vla}}\left [m\right ]=\frac{1}{\left (2\pi\right )^{2}}\int_{\R^{2}\times\R^{2}}\left| p+\bA_{e}(x)+\beta\bA[\rho](x)\right |^{2}m(x,p)\d x\d p +\int_{\R^{2}}V(x)\rho_{m}(x)\d x
\label{EVLA}
\end{equation} 
with $m(x,p)$ a positive measure on the phase-space $\R^{4}$ of positions/momenta and 
\begin{equation}
\rho_{m}(x)=\frac{1}{\left (2\pi\right )^{2}}\int_{\R^{2}}m(x,p)\d p.
\label{rhom}
\end{equation}
Our convention is that $m$ satisfies the mass constraint
\begin{equation}\label{eq:mass m}
\iint_{\R^2 \times \R^2} m(x,p) \d x\d p = \left (2\pi\right )^{2}
\end{equation}
and the semi-classical Pauli principle
\begin{equation}\label{eq:Pauli semi}
\boxed{0\leqslant m(x,p)\leqslant 1\;\; \mathrm{a.e.}}
\end{equation}
The latter is imposed at the classical level by the requirement that a classical state $(x,p)$ cannot be occupied by more than one fermion. 
By the Bathtub principle \cite[Theorem 1.14]{LieLos-01} we can perform the minimization in $p$ explicitly (this parallels the considerations in~\cite{FouLewSol-15}). We find minimizers of the form
\begin{equation}
m_{\rho}(x,p)=\1\Big(\left \b p+\bA_{e}(x)+\beta\bA[\rho](x)\right \b^{2}\leqslant 4\pi\rho(x)\Big)\nn
\end{equation}
where $\rho$ minimizes the Thomas-Fermi energy
\begin{equation}
\E_{\mathrm{TF}}[\rho]=\E_{\mathrm{Vla}}[m_{\rho}]=2\pi \int_{\R^{2}}\rho^{2}(x)\d x+\int_{\R^{2}}V(x)\rho(x)\d x.
\label{ETF}
\end{equation}
We define the minimum of the Thomas-Fermi functional
\begin{equation}
e_{\mathrm{TF}}=\inf\Big\{ \E_{\mathrm{TF}}[\rho]: 0\leqslant \rho \in L^{1}(\R^{2})\cap L^{2}(\R^{2}), \int_{\R^{2}}\rho =1 \Big\}.
\label{etf}
\end{equation}
Our first purpose is to obtain the above as the limit of the true many-body energy. We shall prove this under the following assumptions

\begin{assumption}[\textbf{External Potentials}]\mbox{}\label{HYP}\\
The external potentials entering~\eqref{HRN} satisfy
\begin{itemize}
\item $\bA_{e}\in L^{2}(\R^{2})\cap W^{2,\infty}(\R^{2})$
\item $\left\b \bA_{e}\right \b^{2}\in W^{1,\infty}(\R^{2})\cap L^{2}(\R^{2})$
\end{itemize}
and, for some $s>1$ and $c,C,c',C',c'',C''>0$
\begin{itemize}
 \item $V(x) \geq c \b x\b^{s} - C$ 
 \item $\left| \nabla V (x) \right| \leq c' |x|^{s-1}  + C'$
 \item $\left| \Delta V (x) \right| \leq c'' |x|^{s-2}  + C''$
\end{itemize}
%
%
\end{assumption}

\bigskip


\begin{theorem}[\textbf{Convergence of the ground state energy}\label{convergencedelenergy}]\mbox{}\\
\label{th1}
We consider $N$ extended anyons of radius $R = N^{-\eta}$ in an external potential $V$.  Under Assumption $\eqref{HYP}$ and setting
\begin{equation}
0 < \eta <\frac{1}{4}\nn
\end{equation}
we have, in the parameter regime~\eqref{regime}
\begin{equation}
\boxed{\lim_{N\to\infty} \frac{E^{R}(N)}{N}=e_{\mathrm{TF}}}\nn
\end{equation}
where the many-anyon ground state energy $E^R (N)$ and the Thomas-Fermi energy are defined in~\eqref{eq:GSE} and~\eqref{etf}, respectively.
\end{theorem}

The study of the regime~\eqref{regime} resembles the usual Thomas-Fermi limit for a large fermionic system~\cite{LieSim-77b,FouLewSol-15,Thirring-81,Lieb-81b}, but with important new aspects:
\begin{itemize}
\item The effective interaction comprises a three-body term, and a two-body term (see Equation \eqref{expanded_H}) which mixes position and momentum variables.
\item The limit problem $\eqref{EVLA}$ comprises an effective self-consistent magnetic field $\bA\left [\rho\right ]$.
\item One should deal with the limit $R\to 0$, $\alpha\to 0$ and $\hbar\to 0$ at the same time as $N\to\infty $.
\end{itemize}

Regarding the last point, we make the 


\begin{remark}[Scaling of parameters]\label{rem:scaling}\mbox{}\\
Extracting a multiplicative factor $N$ from the energy, the parameter regime~\eqref{regime} we consider is equivalent to 
\begin{equation}\label{eq:regime bis}
\hbar = 1\;\;\text{and}\;\; \alpha = \frac{\beta}{\sqrt{N}}
\end{equation} 
provided we replace the external potential $V(x)$ by $NV(x)$ and a similar replacement for the external vector potential $\bA_e$. The scaling~\eqref{eq:regime bis} better highlights the quasi-fermionic character of the limit we take. We however find it preferable technically to work with~\eqref{regime}, where the semi-classical aspect is more apparent.

Regarding the scaling of $R$, it would be highly desirable (and challenging) to be able to deal with a dilute regime $\eta > 1/2$ where the radius of the magnetic flux is much smaller than the typical inter-particle distance. We could probably relax our constraint $\eta <1/4 $ to $\eta <1/2$ if we assume a Lieb-Thirring inequality for extended anyons of the form
 \begin{equation}
\bral \Psi_{N}, \sum_{j=1}^{N}\left( p^{\bA}_{j}+\alpha \bA^{R}(x_{j})\right )^{2}\Psi_{N}\ketr\geqslant C_{\alpha , R}\int_{\R^{2}}\rho_{\Psi_{N}}(x)^{2}\d x\nn
\end{equation}
for any fermionic wave-function $\Psi_N$ with one-particle density $\rho_{\Psi_{N}}$. In the above we would need a $C_{\alpha ,R}$ uniformly bounded from below when $N\to \infty, \alpha \to 0$ and $R\gg N^{-1/2}$. Such an inequality is made very plausible by the results of~\cite{LarLun-16} on the homogeneous extended anyon gas and the known Lieb-Thirring inequalities for point-like anyons~\cite{LunSol-13a,LunSol-13b,LunSol-14,LunSei-18}. We plan to return to this matter in the future.\hfill $\diamond$
\end{remark}

The limit ground state energy $e_{\mathrm{TF}}$ does not reveal any anyonic behavior, i.e. it does not depend on $\alpha$. This is however deceptive, for the behavior of the system's ground state \emph{does} depend on $\alpha$ and shows the influence of the non-trivial statistics. This will become apparent when we state the second part of our result, the convergence of ground states. It is expressed in terms of Husimi functions, constructed using coherent states. We discuss this next.

\subsection{Squeezed coherent states}\label{sec:coherent}

Let\footnote{We make the standard choice that $f$ is a gaussian but any radial $L^{2}$ function could be used instead.}
\begin{equation}
f(x):=\frac{1}{\sqrt{\pi}}e^{-\frac{\b x\b^{2}}{2}}.
\label{f}
\end{equation}
For every $(x,p)$ in the phase space $\R^{2}\times\R^{2}$ we define the squeezed coherent state
\begin{equation}
F_{x,p}(y):=\frac{1}{\sqrt{\hbar_{x}}}f\left(\frac{y-x}{\sqrt{\hbar_{x}}} \right)e^{\im\frac{p\cdot y}{\hbar}}. 
\label{CS}
\end{equation}
It has the property of being localized on a scale $\sqrt{\hbar_{x}}$ in space and $\sqrt{\hbar_{p}}$ in momentum, these two scales being related by
\begin{equation}
\hbar=\sqrt{\hbar_{x}}\sqrt{\hbar_{p}}.
\label{hbar}
\end{equation} 
The usual coherent state~\cite[Section 2.1]{FouLewSol-15} is the particular case $\sqrt{\hbar_{x}}=\sqrt{\hbar_{p}}=\sqrt{\hbar}$. Any such state saturates Heisenberg's uncertainty principle, and is thus as classical as a quantum state can be. Indeed, for the two observables $x_{1}$ and $p_{1}=-\im\partial_{x_{1}}$ evaluated in this state we get
$$\Delta_{x_{1}}\Delta_{p_{1}}=\frac{\hbar}{2}
\;\;\text{where}\;\; \Delta_{a}=\bral F_{x,p}, a^{2}F_{x,p}\ketr-\bral F_{x,p}, aF_{x,p}\ketr^2 . $$ 
In our proofs we will take $\hbar_{x}\ll \hbar_{p}\to 0$, which is convenient because our Hamiltonian is singular in $x$ but not in $p$. 

We define the $\hbar$-Fourier tranform
\begin{equation}
\F_{\hbar}[\psi](p)=\frac{1}{2\pi\hbar}\int_{\R^{2}}\psi(x)e^{-\im\frac{p.x}{\hbar}}\d x
\label{TF}
\end{equation}
and denote 
\begin{equation}
F_{\hbar_{x}}(x)=F_{0,0}(x)=\frac{1}{\sqrt{\hbar_{x}}\sqrt{\pi}}e^{-\frac{\b x\b^{2}}{2\hbar_{x}}}.
\label{F}
\end{equation}
Its Fourier transform (see the calculation in Appendix~\ref{app:misc})
\begin{equation}
 G_{\hbar_{p}}(p)=\F_{\hbar}[F_{\hbar_{x}}](p)=\frac{1}{\sqrt{\hbar_{p}}\sqrt{\pi}}e^{-\frac{\b p\b^{2}}{2\hbar_{p}}}
 \label{G}
\end{equation}
makes apparent the localization in momentum space we just claimed. 

To $F_{x,p}$ we associate the orthogonal projector $P_{x,p}$
\begin{equation}
P_{x,p}:=\b  F_{x,p}\rangle \langle F_{x,p}\b.
\label{PP}
\end{equation}
Then we have the well-known

\begin{lemma}[\textbf{Resolution of the identity on $L^{2}\left (\R^{2}\right )$}]\mbox{}\\
\label{Resolution_identity}
We have
\begin{equation}
\frac{1}{(2\pi\hbar)^{2}}\int_{\R^{2}}\int_{\R^{2}}P_{x,p}\;\d x\d p=\1 
\label{id}_{L^2 (\R^2)}.
\end{equation}
\end{lemma}

The proof is recalled in Appendix~\ref{app:misc}, see also \cite[Theorem 12.8]{LieLos-01}.

\subsection{Convergence of Husimi functions}

The Husimi function $m_{\Psi_{N}}^{(k)}(x_{1},p_{1},\ldots ,x_{k},p_{k})$\cite{Husimi-40,ComRob-12,Takahashi-86} describes how many particles are distributed in the $k$ semi-classical boxes of size $\sqrt{\hbar_{x}}\times \sqrt{\hbar_{p}}$ centered at its arguments $(x_{1},p_{1}),\ldots ,(x_{k},p_{k})$. A wave function $\Psi_{N}$ being given we define 
\begin{equation}
m_{\Psi_{N}}^{(k)}(x_{1},p_{1},\ldots ,x_{k},p_{k})=\frac{N!}{(N-k)!}\bral
\Psi_{N},\bigotimes_{j=1}^{k}P_{x_{j},p_{j}}\otimes\1_{N-k}\Psi_{N}\ketr .
\label{Hus}
\end{equation}
Alternatively, 
\begin{equation}
m_{\Psi_{N}}^{(k)}(x_{1},p_{1},\ldots ,x_{k},p_{k})=\frac{N!}{(N-k)!}\int_{\R^{2(N-k)}}\left \b\Big < F_{x_{1},p_{1}}\otimes\ldots \otimes F_{x_{k},p_{k}},\Psi_{N}(\cdot ,z)
\Big >_{L^{2}(\R^{2k})}\right \b^{2}\d z\nn
\end{equation}
or, in terms of fermionic annihilation and creation operators (see~\eqref{eq:adagger} below)
\begin{equation}
m_{\Psi_{N}}^{(k)}(x_{1},p_{1},\ldots ,x_{k},p_{k})=\bral \Psi_{N},a^{*}\left (F_{x_{1},p_{1}}\right )\ldots a^{*}\left (F_{x_{k},p_{k}}\right )a\left (F_{x_{k},p_{k}}\right )\ldots a\left (F_{x_{1},p_{1}}\right )\Psi_{N}\ketr ,\nn 
\end{equation}
where $a$ and $a^{*}$ satisfy the canonical anticommutation relations~\cite{Solovej-notes,GusSig-06}
\begin{align*}
a^{*}(f)a(g)+a(g)a^{*}(f)=\bral g,f\ketr \\
a^{*}(f)a^{*}(g)+a^{*}(g)a^{*}(f)=0.
\end{align*}
We will also use the $k$-particle density matrices of $\Psi_{N}$, i.e. the operator $\gamma^{(k)}_{\Psi_{N}}$ with integral kernel
\begin{equation}
\gamma^{(k)}_{\Psi_{N}}(x_{1},.,x_{k};x'_{1},.,x'_{k})=\begin{pmatrix}
   N \\
   k
\end{pmatrix}
\int_{\R^{2(N-k)}}\Psi_{N}(x_{1},.,x_{N})\overline{\Psi_{N}(x'_{1},.,x'_{k},x_{k+1},.,x_{N})}\d x_{k+1}\ldots \d x_{N}.
\label{gamma}
\end{equation}
We can also express the Husimi function in terms of the $k$-particle density matrices of above
\begin{equation}
m_{\Psi_{N}}^{(k)}(x_{1},p_{1},\ldots ,x_{k},p_{k})=k!\,\Tr\left [\bigotimes_{j=1}^{k}P_{x_{j},p_{j}}\gamma^{(k)}_{\Psi_{N}}\right ]
\label{mgam}.
\end{equation}
These objects are furthermore related to the density marginals in space and in momentum, respectively defined as
\begin{align}
\rho^{(k)}_{\Psi_{N}}(x_{1},\ldots ,x_{k})&=\begin{pmatrix}
   N \\
   k
\end{pmatrix}\int_{\R^{2(N-k)}}\left| \Psi_{N}(x_{1},\ldots ,x_{N})\right|^{2}\d x_{k+1}\ldots \d x_{N}
\label{rhok}\\
t^{(k)}_{\Psi_{N}}(p_{1},\ldots ,p_{k})&=\begin{pmatrix}
   N \\
   k
\end{pmatrix}\int_{\R^{2(N-k)}}\left| \F_{\hbar } [\Psi_{N}](p_{1},\ldots ,p_{N})\right|^{2}\d p_{k+1}\ldots \d p_{N}.
\label{tk}
\end{align}
We can now express the convergence of states corresponding to Theorem~\eqref{convergencedelenergy} using the above objects. The minimizer of the Vlasov energy~\eqref{EVLA} is
\begin{equation}\label{minimizerinp}
m_{\rho^{\mathrm{TF}}}(p,x)=\1\left (\b p+\bA_{e}(x)+\beta\bA [\rho^{\mathrm{TF}} ](x)\b^{2}\leqslant 4\pi\rho^{\mathrm{TF}}(x)\right ).
\end{equation} 
Its marginals in space and momentum are respectively given by
\begin{equation}\label{eq:rhoTF}
\rho^{\mathrm{TF}}(x)= \frac{1}{4\pi}\left( \lambda^{\rm TF} - V (x)\right)_+
\end{equation}
where $\lambda^{\rm TF}$ is a Lagrange multiplier ensuring normalization, and
\begin{equation}\label{eq:tTF}
t^{\mathrm{TF}}(p)=\int_{\R^2} \1_{\left\{\left| p + \bA_e (x) + \beta\bA\left[\rho^{\mathrm{TF}}\right] (x) \right| ^{2} \leq 4\pi \rho^{\mathrm{TF}} (x) \right\}} \d x.
\end{equation}
We will prove the

\begin{theorem}[\textbf{Convergence of states}]\mbox{}\\
\label{th2}
Let $\{\Psi_{N} \}\subset L^{2}_{asym}(\R^{2N})=\bigotimes_{\mathrm{asym}}^{N}L^{2}(\R^{2})$ be any $L^2$-normalized sequence such that
\begin{equation}
\bral\Psi_{N},H_{N}\Psi_{N}\ketr = E^R (N) +o(N)\nn
\end{equation}
in the limit $\hbar=1/\sqrt{N}$, $\alpha =\beta /N$, $N\to \infty$. For any choice of $\hbar_x,\hbar_p \to 0$, related by~\eqref{hbar}, and under Assumption~\ref{HYP}, the Husimi functions of $\Psi_N$ converge to those of the Vlasov minimizer
\begin{equation}
 m_{\Psi_{N}}^{(k)}(x_{1},p_{1},\ldots ,x_{k},p_{k})\to \prod_{j=1}^{k}m_{\rho^{\mathrm{TF}}}(p_j,x_j)\nn
\end{equation}
weakly as measure for all $k\geqslant 1$, i.e. 
\begin{equation}
\int_{\R^{4k}}m^{(k)}_{\Psi_{N}}\phi_k \to\int_{\R^{4k}}(m_{\rho^{\mathrm{TF}}})^{\otimes k}\phi_k \nn
\end{equation} 
for any bounded continous function $\phi_k\in C_{b}\left (\R^{4k}\right )$.

Consequently we also have the convergence of the $k$-particles marginals in position and momentum
\begin{align}
\begin{pmatrix}
   N \\
   k
\end{pmatrix}^{-1}\rho^{(k)}_{\Psi_{N}}(x_{1},\ldots ,x_{k}) &\to \prod_{j=1}^{k}\rho^{\mathrm{TF}} (x_{j})\nn\\
\begin{pmatrix}
   N \\
   k
\end{pmatrix}^{-1}t^{(k)}_{\Psi_{N}}(p_{1},\ldots ,p_{k}) &\to \prod_{j=1}^{k} t^{\mathrm{TF}}(x_{j})\nn
\end{align}
weakly as measures.
\end{theorem}

The convergence of Wigner functions follows from the Husimi functions convergence, see \cite[Theorem 2.1]{FouLewSol-15}. The result above does not depend on the choice of function $f$ used to construct the coherent states~\cite[Lemma~ 2.8]{FouLewSol-15}. Likewise, the result is the same with any choice of scaling of $\hbar_x,\hbar_p\to 0$ as long as~\eqref{hbar} is satisfied. In our proof of energy convergence it will however be convenient to make a particular choice, and in particular use squeezed coherent states, $\hbar_x \ll \hbar_p$.

Recall that the semi-classical Thomas-Fermi energy does not depend on $\alpha $, but the Vlasov minimizer does. Anyonic features are retained e.g. in the limiting object's momentum distribution~\eqref{eq:tTF}.

\subsection{Outline and Sketch of the proof}

Our general strategy is inspired by works on mean-field limits for interacting fermions~\cite{LieSim-77b,LieSolYng-94b,LieSolYng-95,FouMad-19,LewMadTri-19,Thirring-81}, in particular by the method of~\cite{FouLewSol-15}. Several improvements are required to handle the singularity of the anyonic Hamiltonian that emerges in the limit $R\to 0$. In particular 
\begin{itemize}
 \item we replace the use of the Hewitt-Savage theorem by that of its' quantitative version, the Diaconis-Freedman theorem.
 \item we implement the Pauli principle quantitatively at the semi-classical level.
\end{itemize}
To obtain an energy upper bound (Section~\ref{sec:upper}) we use a Slater determinant as trial state, leading by Wick's theorem to a Hartree-Fock-like energy. We then discard exchange terms to get to the simpler Hartree energy~\eqref{EAF}. Finally, a suitable choice of the trial state's one-body density matrix $\gamma$ allows to take the semi-classical limit and establish an upper bound in terms of the Vlasov energy\eqref{EVLA}.

For the correponding lower bound we first express (Section~\ref{sec:lower semi}) the energy semi-classically, in terms of the Husimi functions. The remainder terms are shown to be negligible when $N\to\infty$. The use of squeezed coherent states with $\hbar_x \ll \hbar_p$ takes into account the singularities of the Hamiltonian, more severe in position than in momentum space. 

Section~\ref{sec:lower MF} contains the main novelties compared to~\cite{FouLewSol-15}. There we tackle the mean-field limit of the semi-classical functional expressed in terms of Husimi functions. We use the Diaconis-Freedman theorem~\cite{DiaFre-80} to  express the latter as statistical superpositions of factorized measures, with a quantitatively controled error. We then estimate quantitatively the probability for the Diaconis-Freedman measure to violate the Pauli-principle~\eqref{eq:Pauli semi}. This leads to the sought-after energy lower bound and completes the proof of Theorem \ref{th1}. The convergence of states Theorem \ref{th2} follows as a corollary.

\bigskip

\noindent\textbf{Acknowledgments.} 
Funding from the European Research Council (ERC) under the European Union's Horizon 2020 Research and Innovation Programme (Grant agreement CORFRONMAT No 758620) is gratefully acknowledged.

\section{Energy upper bound}\label{sec:upper}

In this section we derive an upper bound to the ground state energy~\eqref{eq:GSE}. We first introduce a Hartree-Fock trial state and apply Wick's theorem to calculate its energy. We next discard lower contributions (exchange terms) to obtain a regularized ($R>0$) Hartree energy in Lemma~\ref{Ha} and use Lemma~\ref{CRE} to send the regularization to $0$. We conclude by using a particular sequence of semi-classical states, which will make the Hartree energy converge to the Thomas-Fermi one. All in all this will prove the

\begin{proposition}[\textbf{Energy upper bound}]\label{pro:upper}\mbox{}\\
\label{upperbound}
Under Assumption~\ref{HYP}, setting $R=N^{-\eta}$ with $\eta <1/2$, we have
\begin{equation}
\limsup_{N\to \infty}\frac{E^R(N)}{N}\leqslant e_{\mathrm{TF}}.
\label{lowerboundeq}
\end{equation}
\end{proposition}

\subsection{Hartree-Fock trial state}
We recall the definition of creation and annihlation operators of a one-body state $f\in L^2 (\R^2)$ (for more details about second quantization see \cite{Solovej-notes,GusSig-06})
\begin{align}\label{eq:adagger}
a(f)\sum_{\sigma\in \Sigma_{N}}\left (-1\right )^{\mathrm{sgn(\sigma)}}f_{\sigma(1)}\otimes\ldots\otimes f_{\sigma (N)}&=\sqrt{N}\sum_{\sigma\in \Sigma_{N}} \left (-1\right )^{\mathrm{sgn(\sigma)}}\bral f,f_{\sigma(1)}\ketr f_{\sigma (2)} \otimes\ldots\otimes f_{\sigma (N)}\nonumber\\
a^{\dagger}(f)\sum_{\sigma\in \Sigma_{N}} \left (-1\right )^{\mathrm{sgn(\sigma)}}f_{\sigma(1)}\otimes\ldots\otimes f_{\sigma (N)}&=\frac{1}{\sqrt{N+1}}\sum_{\sigma\in \Sigma_{N+1}}\left (-1\right )^{\mathrm{sgn(\sigma)}}f_{\sigma(1)}\otimes\ldots\otimes f_{\sigma (N+1)}.
\end{align}
Sums are over the permutation group and $\mathrm{sgn}(\sigma )$ is the signature of a permutation. Let now $\{\psi_{j}\}_{j=1\ldots N}\in L^2 (\R^2)$ be an orthonormal family. We use it to construct a Slater determinant $\Psi_{N}^{\mathrm{SL}}$ (also called Hartree-Fock state)
\begin{equation}
\label{SL}
\Psi_{N}^{\mathrm{SL}}\left (x_{1},...,x_{N}\right )=\frac{1}{\sqrt{N!}}\det \left (\psi_{i}(x_{j})\right )=\sum_{\sigma\in\Sigma_{N}}(-1)^{\mathrm{sgn}(\sigma)} \prod_{j=1}^{N}\psi_{\sigma(j)}(x_{j}).
\end{equation} 
Denoting $a_{j}^{\dagger}=a^{\dagger}(\psi_{j})$ and $a_{j}=a(\psi_{j})$ we have:
 $$\bral a_{j}^{\dagger}a_{k}\ketr_{\mathrm{SL}}=\bral\Psi_{N}^{\mathrm{SL}}, a_{j}^{\dagger}a_{k}\Psi_{N}^{\mathrm{SL}}\ketr =\delta_{jk}.$$ 
Hence the associated $1$-particle reduced density matrix~\eqref{gamma} and density are
\begin{equation}
\label{gammasl}
\gamma^{(1)}_{\Psi_{N}^{\mathrm{SL}}}=N\Tr_{2\to N}\left [\big\b \Psi_{N}^{\mathrm{SL}}\big >\big < \Psi_{N}^{\mathrm{SL}}\big\b \right ]=\sum_{j}\left|\psi_{j} \right\rangle \left\langle\psi_{j}\right|, \quad \rho_\gamma = \sum_{j}\left|\psi_{j}\right|^2 
\end{equation}
with $\Tr_{2\to N}$ the partial trace.
 Replacing the last formula in \eqref{mgam} gives us the $1$-particle reduced Husimi measure 
\begin{equation}
m_{\Psi^{\mathrm{SL}}_{N}}^{(1)}(x,p)=\sum_{j=1}^{N}\Big\b \big <\psi_{j},F_{x,p}\big >\Big\b^{2}.\nn
\end{equation}
We will use Wick's theorem, see~\cite[Corollary IV.6]{Molinari-17} or~\cite{Solovej-notes}.


\begin{theorem}[\textbf{Corollary of Wick's theorem}]\mbox{}\\
\label{Wick}
Let $a^\sharp_1,\ldots, a^\sharp_{2n}$ be creation or annihilation operators. We have that
\begin{equation}
\bral a_{1}^\sharp \ldots a^\sharp _{2n} \ketr_{\mathrm{SL}}=\sum_{\sigma} (-1)^{\mathrm{sgn}(\sigma)} \prod_{j=1} ^n \left\langle a^{\sharp}_{\sigma (2j-1)} a^{\sharp}_{\sigma (2j)}\right\rangle_{\mathrm{SL}}\nn
\end{equation}
where the sum is over all pairings, i.e. permutations of the $2n$ indices such that $\sigma(2j-1) < \min \left\{ \sigma (2j), \sigma (2j+1)\right\}$ for all $j$. 
\end{theorem}

We evaluate 
$$
\big <\Psi_N^{\mathrm{SL}} ,H_{N}^{R}\Psi_N^{\mathrm{SL}}\big >=\big <H_{N}^{R}\big >_{\mathrm{SL}}
$$
term by term as in $\eqref{expanded_H}$. Until the end of this section $\bral\cdot\ketr$ means $\bral\cdot\ketr_{\mathrm{SL}}$.
We will use the following notation (recall~\eqref{eq:momentum})
\begin{align}
&W_{j}=\left (p_{j}^{\bA}\right )^{2}+V(x_{j})\label{W1}\\
&W_{jk}=p_{j}^{\bA}\cdot\nabla^{\perp}w_{R}(x_{j}-x_{k})+\nabla^{\perp}w_{R}(x_{j}-x_{k})\cdot p_{j}^{\bA}\label{W12}\\
&W_{jk\ell}=\nabla^{\perp}w_{R}(x_{j}-x_{k})\cdot\nabla^{\perp}w_{R}(x_{j}-x_{\ell})
\label{W123}.
\end{align}
Recalling the definition of the Hartree energy functional
\begin{equation*}
\E^{\mathrm{af}}_{R}[\gamma]=\frac{1}{N}\Tr\left [\left (p_{1}^{\bA}+\alpha\bA^{R}[\rho_{\gamma }]\right )^{2}\gamma\right ]+\frac{1}{N}\int_{\R^{2}}V(x)\rho_{\gamma}(x)\d x
\end{equation*}
we have the following

\begin{lemma}[\textbf{Hartree-Fock energy}]\mbox{}\label{trial}\\
With $\Psi^{\mathrm{SL}}$ as above and $\gamma$ the associated one-particle density matrix~\eqref{gammasl} we have 
\begin{align}\label{eq:HF energy}
\big <H_{N}^{R}\big >_{\mathrm{SL}}&=N\E^{\mathrm{af}}_{R}[\gamma]-\alpha\tr \left( W_{12} U_{12} \gamma \otimes \gamma \right)\nonumber\\
&+\alpha^{2}\int_{\R^{4}}\left|\nabla^{\perp}w_{R}(x_{1}-x_{2})\right|^{2}\Big (\rho_{\gamma }(x_{1})\rho_{\gamma }(x_{2})
-\b \gamma (x_{1},x_{2})\b^{2}\Big )\d x_{1} \d x_{2}\nonumber\\
&-\alpha^{2}\int_{\R^{6}}W_{123}\Big [\rho_{\gamma} (x_{1})\b\gamma(x_{2},x_{3})\b^{2}+\rho_{\gamma} (x_{2})\b\gamma(x_{1},x_{3})\b^{2}+\rho_{\gamma} (x_{3})\b\gamma(x_{1},x_{2})\b^{2}\Big ]\d x_{1}\d x_{2}\d x_{3}\nonumber\\
&+\alpha^{2}\int_{\R^{6}}W_{123}\Big [2\Re\Big (\gamma(x_{1},x_{2})\gamma(x_{2},x_{3})\gamma(x_{3},x_{1})\Big )\Big ]\d x_{1}\d x_{2}\d x_{3}.
\end{align}
In the first line we have denoted by $U_{12}$ the exchange operator 
\begin{equation}\label{eq:exchange 1}
 \left(U_{12} \psi \right)(x_1,x_2) = \psi (x_2,x_1). 
\end{equation}
\end{lemma}

\begin{proof}
By definition
$$\bral\sum_{j=1}^{N}(p_{j}^{\bA})^{2}+V(x_{j})\ketr=\Tr\left[\left ((p^{\bA})^{2}+V\right )\gamma\right].$$ 
For the two- and three-body terms we express higher density matrices of $\Psi^{\rm SL}$ using Wick's theorem. To this end, let $\left(\psi_j\right)_j\in L^2 (\R^2)$ be an orthonormal basis, and $a_j,a^\dagger_j$ the associated annihilation and creation operators. 

\medskip

\noindent\textbf{Two-body terms}. We have 
\begin{equation}
\bral a_{\alpha }^{\dagger }a_{\beta }^{\dagger }a_{\epsilon }a_{\gamma }\ketr =\delta_{\alpha\gamma}\delta_{\beta\epsilon}-\delta_{\alpha \epsilon}\delta_{\beta\gamma}\nn
\end{equation}
where $\delta_{ij}$ is the Kronecker delta. It follows that 
\begin{equation}\label{eq:SL two body}
\gamma^{(2)}_{\Psi^{\mathrm{SL}}} = \gamma \otimes \gamma - U_{12} \gamma \otimes \gamma 
\end{equation}
with the exchange operator defined as in~\eqref{eq:exchange 1}. Using~\cite[Lemma 7.12]{Solovej-notes} we have 
\begin{equation}
\sum_{j\neq k}W_{jk}=\sum_{\alpha ,\beta ,\gamma ,\epsilon }\bral \psi_{\alpha }\otimes \psi_{\beta },W_{12}\;\psi_{\gamma }\otimes \psi_{\epsilon }\ketr a_{\alpha }^{\dagger }a_{\beta }^{\dagger }a_{\epsilon }a_{\gamma }\nn
\end{equation}
and thus
\begin{align}
\bral\sum_{j\neq k}W_{j,k}\ketr =\Tr\left [\left (p ^{\bA}.\bA^{R}[\rho_{\gamma}]+\bA^{R}[\rho_{\gamma}].p ^{\bA}\right )\gamma\right ] -\tr \left( W_{12} U_{12} \gamma \otimes \gamma \right).\label{W12E}
\end{align}
Similarly we obtain
\begin{equation}
\bral\sum_{j\neq k}\left|\nabla^{\perp}w_{R}(x_{j}-x_{k})\right|^{2}\ketr
=\int_{\R^{4}}\left|\nabla^{\perp}w_{R}(x_{1}-x_{2})\right|^{2}\Big (\rho_{\gamma }(x_{1})\rho_{\gamma }(x_{2})
-\b \gamma (x_{1},x_{2})\b^{2}\Big )\d x_{1} \d x_{2}.\label{SW12}
\end{equation}

\medskip 

\noindent\textbf{Three-body term}. Using \cite[Lemma 7.12]{Solovej-notes} again we write
  $$\sum_{j\neq k\neq l}W_{jkl}=\sum_{\alpha ,\beta ,\gamma ,\delta ,\epsilon ,\zeta }\bral \psi_{\alpha }\otimes \psi_{\beta }\otimes \psi_{\gamma },W_{123}\; \psi_{\epsilon }\otimes \psi_{\zeta }\otimes \psi_{\eta }\ketr a_{\alpha }^{\dagger }a_{\beta }^{\dagger }a_{\gamma }^{\dagger }a_{\eta }a_{\zeta }a_{\epsilon }.$$
Applying Wick's  theorem~\eqref{Wick} we obtain
$$\bral a_{\alpha }^{\dagger }a_{\beta }^{\dagger }a_{\gamma }^{\dagger }a_{\eta }a_{\zeta }a_{\epsilon }\ketr =\delta_{\alpha\epsilon }\delta_{\beta\zeta }\delta_{\gamma\eta }+\delta_{\alpha\zeta }\delta_{\beta\eta }\delta_{\gamma\epsilon }+\delta_{\alpha\eta }\delta_{\beta\epsilon }\delta_{\gamma\zeta }-\delta_{\alpha\epsilon }\delta_{\beta\eta }\delta_{\gamma\zeta }-\delta_{\alpha\zeta }\delta_{\beta\epsilon }\delta_{\gamma\eta }-\delta_{\alpha\eta }\delta_{\beta\zeta }\delta_{\gamma\epsilon }$$
and we deduce the expression
\begin{equation}
\gamma_{\Psi^{\mathrm{SL}}}^{(3)}=\left( \1 +U_{12}U_{23}+U_{23}U_{12} -U_{12}-U_{13}-U_{23}\right) \gamma\otimes\gamma\otimes\gamma 
\end{equation}
for the $3$-body density matrix of the Slater determinant, where the exchange operators $U_{ij}$ are natural extensions of~\eqref{eq:exchange 1} to the three-particles space. Gathering the above expressions yields
\begin{align}
&\left\langle \sum_{j\neq k\neq l}W_{jkl}\right\rangle=\int_{\R^{6}}W_{123}\Big [\rho_{\gamma} (x_{1})\rho_{\gamma} (x_{2})\rho_{\gamma} (x_{3})\Big ]\d x_{1}\d x_{2}\d x_{3}\nonumber\\
&+\int_{\R^{6}}W_{123}\Big [2\Re\Big (\gamma(x_{1},x_{2})\gamma(x_{2},x_{3})\gamma(x_{3},x_{1})\Big )\Big ]\d x_{1}\d x_{2}\d x_{3}\label{W32}\\
&-\int_{\R^{6}}W_{123}\Big [\rho_{\gamma} (x_{1})\b\gamma(x_{2},x_{3})\b^{2}+\rho_{\gamma} (x_{2})\b\gamma(x_{1},x_{3})\b^{2}+\rho_{\gamma} (x_{3})\b\gamma(x_{1},x_{2})\b^{2}\Big ]\d x_{1}\d x_{2}\d x_{3}.\label{W31}
\end{align}
\end{proof}

We now want to discard the exchange terms and the singular two-body terms to reduce the above expression to the Hartree functional.

\begin{lemma}[\textbf{From Hartree-Fock to Hartree}]\mbox{}\label{Ha}\\
Recalling the notation~\eqref{eq:momentum}, assume that 
$$ \tr\left(|p^{\bA}|^2 \gamma \right) \leq C N.$$
Then
\begin{equation}
E^{R}(N) \leqslant N\E^{\mathrm{af}}_{R}[\gamma] + \frac{C}{R^{2}}.
\label{htoh}
\end{equation}
\end{lemma}

\begin{proof}
We recall the bound 
\begin{equation}\label{eq:w Linf}
\left\Vert \nabla^{\perp }w_{R}\right\Vert_{L^{\infty}}\leq \frac{C}{R} 
\end{equation}
from~\cite[Lemma~2.1]{Girardot-19} and the identities $\gamma =\gamma^{2}=\gamma^{*}$ with
\begin{align*}
\tr \gamma &= \int_{\R^2} \rho_\gamma (x)dx = \int_{\R^2} \gamma (x,x) \d x = N
\end{align*} 
for the density matrix~\eqref{gammasl} of a Slater determinant. It follows that 
$$ 
\left \b\int_{\R^{4}}\left|\nabla^{\perp}w_{R}(x_{1}-x_{2})\right|^{2}\Big (\rho_{\gamma }(x_{1})\rho_{\gamma }(x_{2})
-\b \gamma (x_{1},x_{2})\b^{2}\Big )\d x_{1} \d x_{2}\right \b \leq C\frac{N^{2}}{R^{2}}.
$$
Next we use Cauchy-Schwarz to obtain
$$ |W_{12} + W_{21}| \leq \eps\left( |p_1^{\bA }|^2 + |p_2^{\bA}|^2\right) + \frac{2}{\eps} |\nabla ^\perp w_R (w_1-x_2)|^2$$
and hence, since both sides commute with $U_{12}$,  
$$ 
\left|\alpha\tr \left( W_{12} U_{12} \gamma \otimes \gamma \right)\right| \leq\frac{C\eps}{N} \tr\left( |p^{\bA}|^2 \gamma \right) + \frac{C}{N\eps} \int_{\R^{4}}\left|\nabla^{\perp}w_{R}(x_{1}-x_{2})\right|^{2} \b \gamma (x_{1},x_{2})\b^{2}\d x_{1} \d x_{2}
$$
where we used that 
$$
\tr\left( |p_1^{\bA}|^2 U_{12} \gamma \otimes \gamma \right) = \tr\left(|p^{\bA}|^2 \gamma\right).
$$
Optimizing over $\eps$ gives 
$$ 
\left|\alpha\tr \left( W_{12} U_{12} \gamma \otimes \gamma \right)\right| \leq \frac{C}{R}.
$$
Next, using~\eqref{eq:w Linf} again,
$$ \left| \int_{\R^{6}} W_{123}\rho_{\gamma} (x_{1})\b\gamma(x_{2},x_{3})\b^{2}\d x_{1}\d x_{2}\d x_{3}\right| \leq \frac{C}{R^2} \left(\int_{\R^2} \rho_\gamma (x) dx \right) \left(\int_{\R^4} \left| \gamma (x,y) \right|^2 dx dy\right)\leq C\frac{N^2}{R^2}$$
and we obtain a similar bound on the last term of~\eqref{eq:HF energy} noticing that 
\begin{align*}
\left|\gamma(x_{1},x_{2})\gamma(x_{2},x_{3})\gamma(x_{3},x_{1}) \right| &\leq  \sum_{j} \left|\gamma(x_{1},x_{2}) \psi_j (x_3) \right| \left|\gamma(x_{2},x_{3}) \psi_j (x_1) \right|\\
&\leq \frac{1}{2} \sum_j \left|\gamma(x_{1},x_{2})\right|^2 \left|\psi_j (x_3) \right|^2 + \frac{1}{2} \sum_j \left|\gamma(x_{2},x_{3})\right|^2 \left|\psi_j (x_1) \right|^2\\
 &= \frac{1}{2} \left|\gamma(x_{1},x_{2})\right|^2 \rho_\gamma (x_3) + \frac{1}{2} \left|\gamma(x_{2},x_{3})\right|^2 \rho_\gamma (x_1).
\end{align*}
\end{proof}
We now want to take $R\to 0$ in $\eqref{htoh}$ and show that $\E^{\mathrm{af}}_{R}$ can be replaced by $\E^{\mathrm{af}}$.

\begin{lemma}[\textbf{Convergence of the regularized energy}]\mbox{}\\
\label{CRE}
The functional $\E^{\mathrm{af}}_{R}$ $\eqref{EAF}$ converges pointwise to  $\E^{\mathrm{af}}$ $\eqref{EAF0}$ as $R\to 0$. More precisely, for any fermionic one-particle density matrix $\gamma_{N}$ with 
$$ \tr \gamma_{N} = N$$
and associated density $\rho_{N}$ we have
\begin{equation}
\left \b \E^{\mathrm{af}}_{R}[\gamma_{N}]-\E^{\mathrm{af}}[\gamma_{N}]\right \b\leqslant CR\E^{\mathrm{af}}\left [\gamma_{N}\right ]^{3/2}.
\label{CVR}
\end{equation} 
\end{lemma}

\begin{proof}
In terms of 
$$ \gammat = N^{-1} \gamma_N \mbox{ and } \rhot = N^{-1} \rho_{\gamma_N}$$
we have the expression 
\begin{equation}\label{eq:CVR}
\E^{\mathrm{af}}_{R}[\gamma_{N}] = \tr \left( \left(p^{\bA} + \beta\bA_R [\rhot]\right)^2 \gammat \right) + \int_{\R^2} V \rhot\\
\end{equation}
and similarly for $\E^{\mathrm{af}}$, which is the case $R=0$. The Lieb-Thirring inequality~\cite[Theorem 4.3]{LieSei-09} yields 
$$
\Tr\left ( \left(p^\bA + \beta\bA_R [\rhot]\right )^{2} \gammat \right )\geq C\int_{\R^{2}}\rhot ^{2}
$$
and 
$$
\Tr\left ( \left(p^\bA \right )^{2} \gammat \right )\geq C\int_{\R^{2}}\rhot ^{2}
$$
while the weak Young inequality (see e.g. similar estimates in~\cite[Appendix~A]{Girardot-19},~\cite[Appendix~A]{LunRou-15} or Appendix~\ref{app:Vlasov} below) gives 
$$ \Tr\left ( \left|\bA_R [\rhot]\right| ^{2} \gammat \right) = \int_{\R^2}  \left|\bA_R [\rhot]\right| ^{2} \rhot \leq C \int_{\R^{2}}\rhot ^{2}.$$
Hence a use of the Cauchy-Schwarz inequality implies 
$$ 
\E^{\mathrm{af}}_{R}[\gamma_{N}] \geq C \Tr\left ( \left(p^\bA \right )^{2} \gammat \right ).
$$
We now expand the two energies of \eqref{eq:CVR} and begin by estimating the squared terms using the H\"older and Yound inequalities
\begin{align*}
\norm{ \left (\left \b\bA^{R}[\rhot]\right \b^{2}-\left \b\bA[\rhot]\right \b^{2}\right )\rhot}_{L^{1}}\leqslant\norm{\left \b\bA^{R}[\rhot]\right \b^{2}-\left \b\bA[\rhot]\right \b^{2}}_{L^{2}}^{2}\norm{\rhot}_{L^{2}}\\
\leqslant\norm{ \left (\nabla w_{R}-\nabla w_{0}\right )*\rhot}_{L^{4}}\norm{ \left (\nabla w_{R}+\nabla w_{0}\right )*\rhot}_{L^{4}}\norm{\rhot}_{L^{2}}\\
\leqslant\norm{ \nabla w_{R}-\nabla w_{0}}_{L^{4/3}}\norm{\nabla w_{R}+\nabla w_{0}}_{L^{4/3}}\norm{\rhot}_{L^{2}}^{3}\leqslant CR\E^{\mathrm{af}}\left [\gamma_{N}\right ]^{3/2}
\end{align*} 
as per the above estimates, and bounds on $\nabla w_{R}$ following from~\cite[Lemma~2.1]{Girardot-19}.
As for the mixed term
\begin{align*}
\Tr\left [\left (\bA^{R}[\rhot]-\bA[\rhot]\right ). p^{\bA}\gammat\right ]&\leqslant \norm{
 p^{\bA}\sqrt{\gammat}}_{\mathfrak{S}^{2}}\norm{ \left(\bA^{R}[\rhot]-\bA[\rhot]\right )\sqrt{\gammat}}_{\mathfrak{S}^{2}}\\
&\leqslant C\E^{\mathrm{af}}\left [\gammat\right ]^{1/2}\norm{ \big\b \bA^{R}[\rhot]-\bA[\rhot]\big\b}_{L^{4}}\norm{ \rhot}_{L^{2}}^{1/2}\\
&\leqslant C\E^{\mathrm{af}}\left [\gammat\right ]^{1/2}\norm{ \left (\nabla w_{R}-\nabla w_{0}\right )*\rhot}_{L^{4}}\norm{ \rhot}_{L^{2}}^{1/2}\\
&\leqslant C\E^{\mathrm{af}}\left [\gammat\right ]^{1/2}\norm{ \nabla w_{R}-\nabla w_{0}}_{L^{4/3}}\norm{\rhot }_{L^{2}}^{3/2}\\
&\leqslant CR\E^{\mathrm{af}}\left [\gammat\right ]^{3/2}
\end{align*}
where $\b\b W\b\b_{\mathfrak{S}^{p}}=\Tr\left [\b W\b^{p}\right ]^{\frac{1}{p}}$ is the Schatten norm \cite[Chapter 2]{Schatten-60} and where we used Young's inequality.
\end{proof}

Inserting~\eqref{CVR} in~\eqref{Ha} we get
\begin{equation}
\frac{E^{R}(N)}{N}\leqslant \E^{\mathrm{af}}[\gamma_N] + CR + \frac{C}{NR^{2}}
\label{Haaa}
\end{equation}
for a sequence $\gamma_N$ with uniformly bounded Hartree energy. The last step consists in using constructing such a sequence $\gamma_{N}$, whose energy will behave semi-classicaly.

\subsection{Semi-classical upper bound for the Hartree energy}

We use \cite[Lemma 3.2]{FouLewSol-15}, whose statement we reproduce for the convenience of the reader.

\begin{lemma}[\textbf{Semi-classical limit of the Hartree energy}]\mbox{}\\
Let $\rho\geqslant 0$ be a fixed function in $C_{c}^{\infty}(\R^{2})$ with support in the square $C_{r}=\left (-r/2, r/2\right )^{2}$ such that $\rho\geqslant 0$ and $\int_{C_{r}}\rho =1$. Let $ \bA\in L^{4}(C_{r})$ be a magnetic vector potential. If we define
\begin{equation}
\gamma_{N}=\1\Big (\left (-\im\hbar\nabla +\bA\right )^{2}_{C_{r}}-4\pi \rho(x)\leqslant 0\Big )
\label{gammatest}
\end{equation}
where $ (\left (-\im\hbar\nabla +\bA\right )^{2}_{C_{r}}$ is the magnetic Dirichlet Laplacian in the cube and $ \hbar =1/\sqrt{N} $ we have
\begin{equation}
\lim_{N\to\infty}N^{-1}\Tr\left(\left (-\im\hbar\nabla +\bA\right )^{2}\gamma_{N}\right) =2\pi\int_{\R^{2}}\rho(x)^{2}\d x
\label{semmilim}
\end{equation}
and
\begin{equation}\label{eq:CV dens semi}
\lim_{N\to\infty}N^{-1}\Tr\gamma_{N}=\int_{\R^{2}}\rho(x)\d x\;\;\text{with}\;\;\frac{\rho_{\gamma_{N}}}{N} \to \rho
\end{equation}
and weakly-* in $L^\infty(\R^{2})$, strongly in $L^{1}(\R^{2})$ and $L^2 (\R^2)$.
Furthermore, the same properties stay true if we replace $\gamma_{N}$ by $\widetilde{\gamma}_{N}$ the projection onto the $N$ lowest eigenvectors of $\left (-\im\hbar\nabla +\bA\right )^{2}_{C_{r}}-4\pi\rho(x)$.
\end{lemma}

The convergences we claim for $N^{-1}\rho_{\gamma_{N}}$ are stronger than in~the statement of \cite[Lemma~3.2]{FouLewSol-15}. They easily follow from the proof therein.

We are now able to complete the proof of Proposition~\ref{pro:upper}. We take
\begin{equation}
\gamma_{N}=\1\Big (\left (-\im\hbar\nabla +\bA_{e}+\beta\bA[\rho^{\mathrm{TF}}]\right )^{2}_{C_{r}}-4\pi\rho^{\mathrm{TF}}(x)\leqslant 0\Big )
\label{GC}
\end{equation}
and $\widetilde{\gamma}_{N}$ the associated rank-$N$ projector, as in the above lemma. Therefore we know by the above theorem that
$$\frac{\rho_{\widetilde{\gamma}_{N}}}{N} \to \rho^{\mathrm{TF}} $$
weakly-* in $L^\infty(\R^{2})$, strongly in $L^{1}(\R^{2})$ and $L^2 (\R^2)$. We have that (the definition of the weak $L^2$ space is recalled in Appendix~\ref{app:Vlasov} below)
$$\nabla^{\perp }w_{0}=\nabla^{\perp }\log \b x\b\in L^{1}_{\mathrm{loc}}(C_{r})\cap L^{2,w} (\R^2).$$
We deduce that
\begin{equation}
\alpha \bA[\rho_{\widetilde{\gamma}_{N}}]\to \beta\bA[\rho^{\mathrm{TF}}]\nn
\end{equation}
strongly in $L^\infty (\R^2)$, by the weak Young inequality~\cite[Chapter~4]{LieLos-01}. We thus have
\begin{equation}
\E^{\mathrm{af}}\left [\widetilde{\gamma}_{N}\right ]=\frac{\left (1+o_{N}(1)\right )}{N} \Tr\left [\left (p ^{\bA}+\beta\bA[\rho^{\mathrm{TF}}]\right )^{2}\widetilde{\gamma}_{N}\right ]+\frac{1}{N}\int_{\R^{2}}V(x)\rho_{\tilde{\gamma}_{N}}(x)\d x+o_{N}(1).
\end{equation}
We now use~\eqref{semmilim} to take the limit of the trace and~\eqref{eq:CV dens semi} to take the limit of the potential term ($V\in L^{1}_{\mathrm{loc}}(\R^{2})$ by assumption).

Combining with the bounds from Lemmas~\ref{CRE} and~\ref{Ha} we finally obtain
\begin{align*}
 \frac{E^R (N)}{N} &\leq \E^{\mathrm{af}}\left [\widetilde{\gamma}_{N}\right ] + CR + \frac{C}{R^2 N}\\
 &\leq 2\pi\int_{\R^{2}}\rho_{\mathrm{TF}}^{2}(x)\d x+\int_{\R^{2}}V(x)\rho_{\mathrm{TF}}(x)\d x +o_N (1) \\
 &= e_{\mathrm{TF}} + o_N (1)
\end{align*}
provided $R=N^{-\eta}$ with $0< \eta < \frac{1}{2}$, completing the proof of Proposition~\ref{pro:upper}.
%
%

\section{Energy lower bound:  semi-classical limit}~\label{sec:lower semi}

In this section we first derive some useful a priori bounds and properties of the Husimi functions. Next we express our energy in terms of the latter and obtain a classical energy approximation on the phase space, plus some errors terms that we show to be negligible. We will use this expression in Section~\ref{sec:lower MF} to construct our lower bound.

\subsection{A priori estimates.}

\begin{lemma}[\textbf{Kinetic energy bound}]\mbox{}\\ 
Let $\Psi_{N}\in L_{\mathrm{asym}}\left (\R^{2N}\right )$ be a fermionic wave-function such that 
$$\bral\Psi_{N},H_{N}^{R}\Psi_{N}\ketr \leq C N.$$
Then
\begin{equation}
\frac{1}{N^{2}}\bral\Psi_{N},\Bigg (\sum_{j=1}^{N}-\Delta_{j}\Bigg )\Psi_{N}\ketr  \leqslant \frac{C}{R^{2}}
\label{BKE}
\end{equation}
and
\begin{equation}
\int_{\R^{2}}\left (\frac{\rho^{(1)}_{\Psi_{N}}}{N}\right )^{2}\leqslant \frac{C}{R^{2}}.
\label{BORNERHO1}
\end{equation}
Moreover 
\begin{equation}\label{eq:a priori pot}
\int_{\R^{2}}V(x)\rho^{(1)}_{\Psi_{N}}(x)\d x \leq CN.
\end{equation}
\end{lemma}

\begin{proof}
We expand the Hamiltonian and use the Cauchy-Schwarz inequality for operators to get
 \begin{align}
  H_{N}^{R}&=\sum_{j=1}^{N}\left(\left(p_{j}^{\bA}\right)^{2}+\alpha p_{j}^{\bA}\cdot\bA_{j}^{R}+\alpha \bA_{j}^{R}.p_{j}^{\bA}+\alpha^{2}|\bA_{j}^{R}|^{2}+V(x_{j})\right)\nn\\
  &\geqslant\sum_{j=1}^{N}\left((1-2\delta^{-1})\left(p_{j}^{\bA}\right)^{2}+(1-2\delta)\alpha^{2}|\bA_{j}^{R}|^{2}+V(x_{j})\right)\nn\\
  &=\sum_{j=1}^{N}\left(\frac{1}{2}(\left(p_{j}^{\bA}\right)^{2}+V(x_{j}))-7\alpha^{2}|\bA_{j}^{R}|^{2}\right)\nn
 \end{align}
choosing $\delta=4$. Thus we have
\begin{equation}
  \mathrm{Tr}\left[\left(\left(p^{\bA}\right)^{2}+V\right)\gamma_{N}^{(1)}\right]\leqslant CN+7\alpha^{2}\left<\Psi_{N},\sum_{j=1}^{N}|\bA_{j}^{R}|^{2}\Psi_{N}\right>\nn
\end{equation}
where $\gamma_{N}^{(1)}$ is the $1$-particle reduced density matrix $$N\mathrm{Tr}_{2\to N}\left [\ketl\Psi_{N}\ketr\bral\Psi_{N}\brar\right ].$$
We estimate the second term of the above using~\cite[Lemma~2.1]{Girardot-19}:
\begin{equation}
\left<\Psi_{N},\sum_{j=1}^{N}|\bA_{j}^{R}|^{2}\Psi_{N}\right>\leqslant\frac{CN^{3}}{R^{2}}.\nn
\end{equation}
Recalling that $\alpha = N^{-1}$ we get
\begin{equation}
  \mathrm{Tr}\left[\left(\left(p^{\bA}\right)^{2}+V\right)\gamma_{N}^{(1)}\right]\leqslant \frac{CN}{R^{2}}\nn
\end{equation}
Now we use that
$$ \b -\im\hbar\nabla + \bA_{e}\b^{2}\geqslant -\frac{\hbar^{2}}{2}\Delta -\b \bA_{e}\b^{2}$$
and that $\hbar^{2}=N^{-1}$ to obtain
\begin{equation}
\frac{CN}{R^{2}}\geqslant \frac{C}{N}\bral\Psi_{N},\Bigg (\sum_{j=1}^{N}-\Delta_{j}\Bigg )\Psi_{N}\ketr -\mathrm{Tr}\left [\b \bA_{e}\b^{2}\gamma_{N}^{(1)}\right ] \nn
\end{equation}
recalling that $V\geqslant 0$. But $\b\bA_{e}\b^{2}\in L^{\infty}(\R^{2})$, we thus deduce $\eqref{BKE}$. We obtain $\eqref{BORNERHO1}$ by using the Lieb-Thirring inequality \cite{LieThi-75,LieThi-76}
\begin{equation}
\mathrm{Tr}\left [-\Delta \gamma_{N}^{(1)}\right ]\geqslant C\int_{\R^{2}}\left (\rho_{\psi_{N}}^{(1)}(x)\right )^{2}\d x\nn
\end{equation}
combined with $\eqref{BKE}$.

Finally~\eqref{eq:a priori pot} is a straightforward consequence of the positivity of the full kinetic energy (first term in~\eqref{Ham2}).
\end{proof}

\subsection{Properties of the phase space measures}

The following connects the space and momentum densities $\eqref{rhok}$ and $\eqref{tk}$ with the Husimi function and is extracted from \cite[Lemma 2.4]{FouLewSol-15}. The generalization to squeezed states we use does not change the properties of the phase space measures.

\begin{lemma}[\textbf{Densities and fermionic semi-classical measures}]\mbox{}\\
\label{densitiesconv}
Let $F_{\hbar_x}$, $G_{\hbar_p}$ be defined as in~\eqref{F} and~\eqref{G} (recall that $\hbar =\sqrt{\hbar_{x}}\sqrt{\hbar_{p}}$). Let $\Psi_{N}\in L^2_{\rm asym} (\R^{2N})$ be any normalized fermionic wave function. We have
\begin{equation}
\frac{1}{(2\pi)^{2k}}\int_{\R^{2k}}m_{\Psi_{N}}^{(k)}(x_{1},p_{1},...,x_{k},p_{k})\d p_{1}...\d p_{k}=k!\hbar^{2k}\rho^{(k)}_{\Psi_{N}}*\left (\b F_{\hbar_{x}}\b^{2}\right )^{\otimes k}
\label{marg_en_x}
\end{equation}
and
\begin{equation}
\frac{1}{(2\pi)^{2k}}\int_{\R^{2k}}m_{\Psi_{N}}^{(k)}(x_{1},p_{1},...,x_{k},p_{k})\d x_{1}...\d x_{k}=k!\hbar^{2k}t^{(k)}_{\Psi_{N}}*\left (\b G_{\hbar_{p}}\b^{2}\right )^{\otimes k}.
\label{marg_en_t}
\end{equation}
\end{lemma}
The proof is similar to considerations from~\cite{FouLewSol-15}. We recall it in Appendix~\ref{app:misc} for completeness.


\begin{lemma}[\textbf{Properties of the phase space measures}]\mbox{}\\
\label{propH}
Let $\Psi_{N}\in L^{2}_{\mathrm{asym}}(\mathbb{R}^{2N})$ be a normalized fermionic wave function. For every $1\leqslant k\leqslant N$ the function $m_{\Psi_{N}}^{(k)}$ defined in $\eqref{Hus}$ is symmetric and satisfies
\begin{equation}
0\leqslant m_{\Psi_{N}}^{(k)}\leqslant 1 \;\;\mathrm{a.e.}\;\text{on}\;\;\R^{4k}.
\label{pauli}
\end{equation}
In addition
\begin{align}
\frac{1}{(2\pi)^{2k}}\int_{\R^{4k}}m_{\Psi_{N}}^{(k)}(x_{1},p_{1},...,x_{k},p_{k})\d x_{1}...\d p_{k}&=N(N-1)...(N-k+1)\hbar^{2k}\to 1\nn\\
\end{align}
 when $N$ tends to infinity and
\begin{equation}
\frac{1}{(2\pi)^{2}}\int_{\R^{4}}m_{\Psi_{N}}^{(k)}(x_{1},p_{1},...,x_{k},p_{k})\d x_{k}\d p_{k}=\hbar^{2}(N-k+1)m_{\Psi_{N}}^{(k-1)}(x_{1},p_{1},...,x_{k-1},p_{k-1}).\nn
\end{equation}
\end{lemma}
We refer to~\cite[Lemma 2.2]{FouLewSol-15} for the  proof.


\subsection{Semi-classical energy}

We now define a semi-classical analogue of the original energy functional, given in terms of Husimi functions. We denote $\d x_{ij}=\d x_{i}\d x_{j}$ and $x_{ij}=\left (x_{i},x_{j}\right )$:
\begin{align}\label{eq:class ener}
 \E^{R}_{\mathrm{C}}\left [m_{\Psi_{N}}^{(3)} \right] &=\frac{1}{(2\pi)^{2}}\int_{\R^{4}} \left (\b p+\bA_{e}(x)\b^{2}+V(x_{1})\right )m^{(1)}_{\Psi_{N}}(x,p)\d x\d p\nn\\
&+\frac{2\beta}{(2\pi)^{4}}\int_{\R^{8}} (p_{1}+\bA_{e}(x_{1}))\cdot\nabla^{\perp }w_{R}(x_{1}-x_{2}) m^{(2)}_{\Psi_{N}}(x_{12},p_{12})\d x_{12}\d p_{12}\nn \\
&+\frac{\beta^{2}}{(2\pi)^{6}}\int_{\R^{12}} \nabla^{\perp }w_{R}(x_{1}-x_{2})\cdot\nabla^{\perp }w_{R}(x_{1}-x_{3}) m^{(3)}_{\Psi_{N}}(x_{123},p_{123})\d x_{123}\d p_{123}.
\end{align}

Our aim is to show that it correctly captures the leading order of the full quantum energy.

\begin{proposition}[\textbf{Semi-classical energy functional}]\mbox{}\\
\label{semiclassenergy}
Pick some $\eps >0$ and set
\begin{equation}
\sqrt{\hbar_{x}}=N^{-1/2 + \epsilon}, \quad \sqrt{\hbar_{p}}=N^{-\epsilon} \label{hxhp}.
\end{equation}
Pick a sequence $\Psi_{N}$ of fermionic wave-functions such that 
\begin{equation}
\bral \Psi_{N}, H_{N}^{R}\Psi_{N}\ketr \leq CN.
\label{gfgfg}
\end{equation}
Under Assumption~\ref{HYP} we have
\begin{align}
\label{ttt}
\frac{\bral \Psi_{N}, H_{N}^{R} \Psi_{N}\ketr}{N}&\geqslant \left (1-o_{N}(1)\right )\E^{R}_{\mathrm{C}}\left [m_{\Psi_{N}}^{(3)} \right] - CN^{4\eta -1+2\epsilon} - C N ^{2\eta - 1} - C N^{-2\epsilon}.
\end{align}
\end{proposition}

We first derive an exact expression, whose extra unwanted terms will be estimated below. 

\begin{lemma}[\textbf{Semi-classical energy with errors}]\mbox{}\\
\label{Cor}
Let $\Psi_N$ be a fermionic wave-function satisfying the bound $\eqref{gfgfg}$, with $\langle \, . \, \rangle$ the corresponding expectation values:
\begin{align}
&\bral\sum_{j=1}^{N}\left(p^{\bA}_{j}+\alpha \bA^{R}(x_{j})\right)^{2}+V(x_{j})\ketr
 \geqslant N\E^{R}_{\mathrm{C}}\left [m_{\Psi_{N}}^{(3)} \right]\nn\\
&+2\alpha\Re\bral\sum_{j=1}^{N}\sum_{k\neq j}\left (\nabla^{\perp }w_{R}- \nabla^{\perp }w_{R}*\b F_{\hbar_{x}}\b^{2}*\b F_{\hbar_{x}}\b^{2}\right )(x_{j}-x_{k})\cdot(-\im\hbar\nabla_{j})\ketr\label{crossterm}\\
&+2\alpha\bral\sum_{j=1}^{N}\sum_{k\neq j}\bA_{e}(x_{j})\cdot\nabla^{\perp }w_{R}(x_{j}-x_{k})\right.\nn\\
&\;\;\;\;\;\;\;\;\;\;\;\;\;\;\;\;\;\;\;\;\;\;\;\;\;\left.-\int_{\R^{2}}\bA_{e}(x_{j}-u)\cdot\big (\nabla^{\perp}w_{R}*\b F_{\hbar_{x}}\b^{2}\big )(x_{j}-x_{k}-u)\b F(u)\b^{2}\d u\ketr\label{twomagcross}\\
&+\alpha^{2}\sum_{j=1\neq k\neq l}^{N}\bral\nabla^{\perp }w_{R}(x_{j}-x_{k})\cdot\nabla^{\perp }w_{R}(x_{j}-x_{l})-\int_{\R^{2}}\Big ( \nabla^{\perp }w_{R}*\b F_{\hbar_{x}}\b^{2}\Big)(x_{j}-x_{k}-u)\right.\nn \\
&\;\;\;\;\;\;\;\;\;\;\;\;\;\;\;\;\;\;\;\;\;\;\;\;\;\;\;\;\;\;\;\;\;\;\left.   \cdot\Big ( \nabla^{\perp }w_{R}*\b F_{\hbar_{x}}\b^{2}\Big) (x_{j}-x_{l}-u)\b F_{\hbar_{x}}(u)\b^{2}\d u\ketr\label{treebodyterm}\\
&+ \mathrm{Err_1} + \mathrm{Err_2}\label{bigenr}
\end{align}
with
\begin{align}
\mathrm{Err_1} &= 2\Re\bral\sum_{j=1}^{N}(\bA_{e}-\bA_{e}*\b F_{\hbar_{x}}\b^{2})(x_{j})\cdot(-\im\hbar\nabla_{j})\ketr \label{magncross}\\
 &+\bral\sum_{j=1}^{N}(\b\bA_{e}\b^{2}-\b\bA_{e}\b^{2}*\b F_{\hbar_{x}}\b^{2})(x_{j})+V -V * \b F_{\hbar_{x}}\b^{2})(x_{j})\ketr \label{potmagn}
\end{align}
and
\begin{equation}
\mathrm{Err_2}=\alpha^{2}\bral\sum_{j=1}^{N}\sum_{k\neq j}\b\nabla^{\perp }w_{R}(x_{j}-x_{k})\b^{2}\ketr\ -N\hbar_{p}\int_{\R^{2}}\b\nabla f(y)\b^{2}\d y.
\label{lastterm}
\end{equation}
\end{lemma}

\begin{proof}
We expand the whole energy as in $\eqref{HRN}$ and express each term using Husimi's function $\eqref{Hus}$. We begin with calculations involving the $p$ variable. Using~\eqref{tk} and~\eqref{marg_en_t} we can write, for the purely kinetic term,
\begin{align*}
&\bral \Psi_{N},\left (\sum_{j=1}^{N} -\hbar^{2}\Delta_{j}\right )\Psi_{N}\ketr =\int_{\R^{2}}\b p\b^{2}t^{(1)}_{\Psi_{N}}(p)\d p\nn \\
=&\frac{1}{(2\pi\hbar)^{2}}\int_{\R^{2}\times\R^{2}}m_{\Psi_{N}}^{(1)}(x,p)\b p\b^{2}\d x\d p+\int_{\R^{2}\times\R^{2}}t^{(1)}_{\Psi_{N}}(p)\b G_{\hbar_{p}}(q-p)\b^{2}(\b p\b^{2}-\b q\b^{2})\d p\d q\nn\\
=&\frac{1}{(2\pi\hbar)^{2}}\int_{\R^{2}\times\R^{2}}m_{f,\Psi_{N}}^{(1)}(x,p)\b p\b^{2}\d x\d p-\int_{\R^{2}\times\R^{2}}t^{(1)}_{\Psi_{N}}(p)\b G_{\hbar_{p}}(q-p)\b^{2}(\b q-p\b^{2})\d p\d q\nn\\
-&2\left (\int_{\R^{2}}pt^{(1)}_{\Psi_{N}}(p)\d p\right )\cdot \left (\int_{\R^{2}}p\b G_{\hbar_{p}}(p)\b^{2}\d p\right ).
\end{align*}
We now use that $G_{\hbar_{p}}(-p)=\overline{G_{\hbar_{p}}(p)}$ which makes $\b G_{\hbar_{p}}(p)\b^{2}$ an even function (recall that $G_{\hbar_p}$ is the Fourier tranform of $F_{\hbar_{x}}$ $\eqref{F}$) to discard the  last term of the above. On the other hand, using~\eqref{tk} we find
\begin{align*}
&\int_{\R^{2}\times\R^{2}}t^{(1)}_{\Psi_{N}}(p)\b G_{\hbar_{p}}(q-p)\b^{2}(\b q-p\b^{2})\d p\d q=N\int_{\R^{2}}\b G_{\hbar_{p}}\left (p\right )\b^{2}\b p\b^{2}\d p\\
=&\frac{N}{\left (2\pi\hbar\right )^{2}}\int_{\R^{6}}\frac{\b p\b^{2}}{\hbar_{x}}f\left (\frac{x}{\sqrt{\hbar_{x}}}\right )f\left (\frac{x'}{\sqrt{\hbar_{x}}}\right )e^{\frac{\im p\cdot x}{\hbar}}e^{-\frac{\im p\cdot x'}{\hbar}}\d p\d x\d x'\\
=&\frac{N\hbar^{2}}{\left (2\pi\hbar\right )^{2}}\int_{\R^{6}}\nabla f(u)\nabla f(v)e^{\frac{\im p\cdot \left (u-v\right )}{\sqrt{\hbar_{p}}}}\d p\d u\d v
=-N\hbar_{p}\int_{\R^{2}}\b\nabla f(y)\b^{2}\d y,
\end{align*}
thus concluding that
\begin{equation*}
\bral \Psi_{N},\left (\sum_{j=1}^{N} -\hbar^{2}\Delta_{j}\right )\Psi_{N}\ketr=\frac{1}{(2\pi\hbar)^{2}}\int_{\R^{2}\times\R^{2}}m_{f,\Psi_{N}}^{(1)}(x,p)\b p\b^{2}\d x\d p-N\hbar_{p}\int_{\R^{2}}\b\nabla f(y)\b^{2}\d y.
\label{tohusimi}
\end{equation*}
For the  magnetic cross term  $\eqref{magncross}$ we use~\eqref{tool} to obtain, for any $u\in L^{2}(\R^{2})$,
\begin{multline*}
\frac{1}{(2\pi\hbar)^{2}}\int_{\R^{2}\times\R^{2}}p\cdot \bA_{e}(x)\b \bral u,F_{x,p}\ketr\b^{2}\d x\d p=\int_{\R^{2}\times\R^{2}}p\cdot \bA_{e}(x)\b\F_{\hbar }[F_{x,0}u](p)\b^{2}\d p\d x\\
=\int_{\R^{2}} \bA_{e}(x)\cdot\int_{\R^{2}}p\b\F_{\hbar }[F_{x,0}u](p)\b^{2}\d p\d x=
\hbar\int_{\R^{2}}\bA_{e}(x)\cdot\Im\int_{\R^{2}}F_{x,0}(y)\overline{u}(y)\nabla[F_{x,0}u](y)\d y\d x.
\end{multline*}
Then since $F_{x,0}$ and $\b u\b^{2}$ are real we obtain
\begin{align*}
\frac{1}{(2\pi\hbar)^{2}}\int_{\R^{2}\times\R^{2}}p\cdot \bA_{e}(x)\b \bral u,F_{x,p}\ketr\b^{2}\d x\d p &= \hbar\int_{\R^{2}}\bA_{e}(x)\cdot\int_{\R^{2}}\b F_{\hbar_{x}}(y-x)\b^{2}\Im \left [\overline{u}(y)\nabla u(y)\right ]\d y\d x\\
&=\hbar\int_{\R^{2}}(\bA_{e}(x)*\b F_{\hbar_{x}}\b^{2})(y)\cdot\Im \left [\overline{u}(y)\nabla u(y)\right ]\d y.
\end{align*}
Combining the spectral decomposition of the one-particle reduced density matrix with the above we get
\begin{equation}
\hbar\sum_{j=1}^{N}\bral  p_{j}\cdot \bA_{e}(x_{j})\ketr =\frac{1}{\left( 2\pi\hbar\right )^{2}}\int_{\R^{4}}p\cdot\bA_{e}\d m^{(1)}_{\Psi_{N}}+\Re\bral\sum_{j=1}^{N}(\bA_{e}-\bA_{e}*\b F_{\hbar_{x}}\b^{2})(x_{j})\cdot(-\im\hbar\nabla_{j})\ketr .\nn
\end{equation}
Now for the mixed two-body term $\eqref{crossterm}$:
\begin{align*}
M:=&\frac{1}{N(N-1)}\frac{1}{(2\pi\hbar)^{2}}\int_{\R^{4}\times\R^{4}}m^{(2)}_{\Psi_{N}}(x_{1},x_{2},p_{1},p_{2})\left (p_{1}\cdot\nabla^{\perp}w_{R}(x_{1}-x_{2})\right )\d x_{1}\d x_{2}\d p_{1}\d p_{2}\\
=&-\int_{\R^{2(N-2)}}\int_{\R^{4}\times\R^{4}}p_{1}\cdot\nabla^{\perp}w_{R}(x_{1}-x_{2})\left \b \bral F_{x_{1},p_{1}}\otimes F_{x_{2},p_{2}}(\cdot)\Psi_{N}(\cdot,y)\ketr\right \b^{2}_{L^{2}(\R^{4})}\d x_{12}\d p_{12}\d y\\
=&-\int_{\R^{2(N-2)}}\frac{\im\hbar}{2}\int_{\R^{4}}\nabla^{\perp}w_{R}(x_{1}-x_{2})\cdot\int_{\R^{8}}[\nabla_{y_{1}}-\nabla_{z_{1}}]\Big (F_{x_{1},0}(y_{1})F_{x_{1},0}(z_{1})\overline{\Psi}(y_{1},y_{2},y)\\
&\;\;\;\;\;\;\;\;\;\;\;\;\;\;\;\;\Psi_{N}(z_{1},z_{2},y)F_{x_{2},0}(y_{2})F_{x_{2},0}(z_{2})\Big )\delta (y_{2}-z_{2})\delta (y_{1}-z_{1})\d x_{12}\d y_{12}\d z_{12}\d y
\end{align*}
and since $-\frac{\im}{2}(u-\overline{u})=\Im [u]$, we get
\begin{align*}
M=&-\hbar\int_{\R^{2(N-2)}}\int_{\R^{4}\times\R^{4}}\nabla^{\perp }w_{R}(x_{1}-x_{2})\b F_{\hbar_{x}}(y_{1}-x_{1})\b^{2}\b F_{\hbar_{x}}(y_{2}-x_{2})\b^{2}\cdot\Im \Big [\overline{\Psi}_{N}\nabla_{y_{1}}\Psi_{N}\Big ]\d x_{12}\d y_{12}\d y\\
=&\int_{\R^{2(N-2)}}\int_{\R^{4}}\Bigg (\Big (\nabla^{\perp} w_{R}*\b F_{\hbar_{x}}\b^{2}\Big) *\b F_{\hbar_{x}}\b^{2}\Bigg )(x_{1}-x_{2})\Re\Big [ -\im\hbar\nabla_{x_{1}}\Psi_{N}\overline{\Psi }_{N}\Big ]\d x_{1}\d x_{2}\d y
\end{align*}
and conclude
\begin{align*}
&\hbar\sum_{j=1}^{N}\sum_{k\neq j}\bral p_{j}\cdot\nabla^{\perp }w_{R}(x_{j}-x_{k})\ketr=\frac{1}{(2\pi\hbar)^{4}}\int_{\R^{4}\times\R^{4}}\left (p_{1}\cdot\nabla^{\perp}w_{R}(x_{1}-x_{2})\right )\d m^{(2)}_{\Psi_{N}}\nn\\
\;\;\;\;\;\;\;\;&+\Re\bral\sum_{j=1}^{N}\sum_{k\neq j}\left (\nabla^{\perp }w_{R}- \nabla^{\perp }w_{R}*\b F_{\hbar_{x}}\b^{2}*\b F_{\hbar_{x}}\b^{2}\right )(x_{j}-x_{k})\cdot(-\im\hbar\nabla_{j})\ketr.
\end{align*}
For the terms which do not involve the variable $p$ we only discuss the three-body term~\eqref{treebodyterm}, the others being treated in the same way. We apply Lemma~\eqref{densitiesconv}
\begin{align*}
&\frac{1}{(2\pi\hbar)^{3}}\int_{\R^{6}\times\R^{6}}W_{123}(x_{1},x_{2},x_{3})m_{\Psi_{N}}^{(3)}(x_{1},p_{1},..,x_{3},p_{3})\d x_{123}\d p_{123}\\
&=\int_{\R^{6}\times\R^{6}}W_{123}(x_{1},x_{2},x_{3})\rho_{\Psi_{N}}^{(3)}*\left (\left \b F_{\hbar_{x}}\right \b^{2}\right )^{\otimes 3}\d x_{123}\\
&=\int_{\R^{12}}\nabla^{\perp }w_{R}(u-v+y_{1}-y_{2})\cdot \nabla^{\perp }w_{R}(u-w+y_{1}-y_{3})\\
&\;\;\;\;\;\;\;\;\;\;\;\;\;\;\;\;\;\;\;\;\;\;\;\;\;\;\;\;\;\;\;\;\;\;\;\;\;\;\;\;\;\;\;\;\;\;\;\;\;\;\rho_{\Psi_{N}}^{(3)}(u,v,w)\left \b F_{\hbar_{x}}\right \b^{2}(y_{1})\left \b F_{\hbar_{x}}\right \b^{2}(y_{2})\left \b F_{\hbar_{x}}\right \b^{2}(y_{3})\d y_{123}\d u \d v\d w\nn\\
=&\int_{\R^{6}}\Big ( \nabla^{\perp }w_{R}*\left \b F_{\hbar_{x}}\right \b^{2}\Big)(v-u-y_{1}) \cdot\Big ( \nabla^{\perp }w_{R}*\left \b F_{\hbar_{x}}\right \b^{2}\Big) (w-u-y_{1})\left \b F_{\hbar_{x}}(y_{1})\right \b^{2}\rho_{\Psi_{N}}^{(3)}\d u\d v\d w\d y_{1}\nn
\end{align*}
and then
\begin{align*}
&\sum_{j=1}^{N}\sum_{k\neq j}\sum_{l\neq k}\bral \nabla^{\perp}w_{R}(x_{j}-x_{k})\nabla^{\perp}w_{R}(x_{j}-x_{l})\ketr=\frac{1}{(2\pi\hbar)^{6}}\int_{\R^{12}} \left (\nabla^{\perp }w_{R}\cdot\nabla^{\perp }w_{R}\right )\d m^{(3)}_{\Psi_{N}}\nn\\
&+\sum_{j=1\neq k\neq l}^{N}\bral\nabla^{\perp }w_{R}(x_{j}-x_{k})\cdot\nabla^{\perp }w_{R}(x_{j}-x_{l})-\int_{\R^{2}}\Big ( \nabla^{\perp }w_{R}*\b F_{\hbar_{x}}\b^{2}\Big)(x_{j}-x_{k}-u)\right.\nn \\
&\;\;\;\;\;\;\;\;\;\;\;\;\;\;\;\;\;\;\;\;\;\;\;\;\;\;\;\;\;\;\;\;\;\;\left.   \cdot\Big ( \nabla^{\perp }w_{R}*\b F_{\hbar_{x}}\b^{2}\Big) (x_{j}-x_{l}-u)\b F_{\hbar_{x}}(u)\b^{2}\d u\ketr.
\end{align*}
\end{proof}

We now show that all errors in~\eqref{bigenr} are smaller than the classical energy $\E^{\mathrm{C}}$. This will conclude the proof of Proposition~\ref{semiclassenergy}.
The following lemma will deal with all convolutions involved in the estimates.

\begin{lemma}[\textbf{Convolution terms}]\mbox{}\\
\label{plancherel}
For any function $W$ on $L^{\infty}\left (\R^{2}\right )$, consider the $n$-fold convoluted product of $W$ with $\b F_{\hbar_{x}}\b^{2}$.
\begin{equation}
W*_{n}\b F_{\hbar_{x}}\b^{2}=W*\b F_{\hbar_{x}}\b^{2}*\b F_{\hbar_{x}}\b^{2}*...*\b F_{\hbar_{x}}\b^{2}.\nn
\end{equation}
We have the estimate
\begin{equation}
\norm{W*_{n}\b F_{\hbar_{x}}\b^{2}-W}_{L^{\infty}}\leqslant C_{n}\hbar_{x}\norm{\Delta W}_{L^{\infty}}.
\label{convv}
\end{equation} 
\end{lemma}
\begin{proof}
For any $x,y\in \R^2$ there exists some $z\in \R^2$ such that 
$$ W\left(x-\sqrt{\hbar_x}y\right) = W (x) - \sqrt{\hbar_x} \nabla W (x) \cdot y + \frac{\hbar_x ^2}{2} \langle y, \mathrm{Hess} W (z) y\rangle $$
and hence
\begin{align}
\norm{W-W*\b F_{\hbar_{x}}\b^{2}(x)}_{L^{\infty}}&=\left \b W(x)-\frac{1}{\pi\hbar_{x}}\int_{\R^{2}}W(x-y)e^{-\frac{\b y\b^{2}}{\hbar_{x}}}\d y\right \b\nn\\
&=\left \b W(x)-\frac{1}{\pi}\int_{\R^{2}}W\left(x-\sqrt{\hbar_{x}}y\right)e^{-\b y\b^{2}}\d y\right \b\nn\\
&\leqslant C\hbar_{x}\norm{\Delta W}_{L^{\infty}}\nn
\end{align}
using the radial symmetry of $y\mapsto e^{-\b y\b^2} $ to discard the term of order $\sqrt{\hbar_x}$. It follows that for any $n\geq 2$
\begin{equation}
\norm{W*_{n}\b F_{\hbar_{x}}\b^{2}-W}_{L^{\infty}}\leqslant C_{n}\hbar_{x}\norm{\Delta W}_{L^{\infty}}\nn
\end{equation}
because
$$ \norm{\Delta W*_{n-1}\b F_{\hbar_{x}}\b^{2}}_{L^\infty} = \norm{\left(\Delta W \right)*_{n-1}\b F_{\hbar_{x}}\b^{2}}_{L^\infty} \leq C \norm{\Delta W}_{L^\infty}.$$
\end{proof}


\begin{proof}[Proof of Proposition $\eqref{semiclassenergy}$]
Let $\Psi_{N}$ be a sequence of fermionic wave functions such that 
$$\bral\Psi_{N}, H_{N}^{R}\Psi_{N}\ketr = O(N).$$
We will systematically use the bounds $\eqref{BORNERHO1}$ and $\eqref{BKE}$
\begin{equation}
\norm{\rho_{\gamma_{N}}}^{2}_{L^{2}(\R^{2})}\leqslant \frac{CN^{2}}{R^{2}}\;\;\;\;\text{and}\;\;\;\;\;\bral\Psi_{N},\Bigg (\sum_{j=1}^{N}-\Delta_{j}\Bigg )\Psi_{N}\ketr \leqslant \frac{CN^{2}}{R^{2}}\nn
\end{equation} 
to deal with our error terms. Recall that the main term in the energy $\eqref{bigenr}$ is of order $N$.

\medskip

\noindent\textbf{Estimate of~\eqref{magncross}.} We have 
\begin{align*}
\left| \Tr\left(\left(\bA_{e}-\bA_{e}*\b F_{\hbar_{x}}\b^{2})\cdot(-\im\hbar\nabla \right)\gamma_{N}\right)\right|
&\leqslant\norm{(-\im\hbar\nabla )\sqrt{\gamma_{N}}}_{\mathfrak{S}^{2}}\norm{ (\bA_{e}-\bA_{e}*\b F_{\hbar_{x}}\b^{2})\sqrt{\gamma_{N}}}_{\mathfrak{S}^{2}}\nonumber\\
&\leqslant \hbar\sqrt{\Tr(-\Delta)\gamma_{N}}\norm{(\bA_{e}-\bA_{e}*\b F_{\hbar_{x}}\b^{2})\sqrt{\rho_{\gamma_{N}}}}_{L^{2}(\R^{2})}\nonumber\\
&\leqslant C\hbar\frac{N}{R}\norm{\bA_{e}-\bA_{e}*\b F_{\hbar_{x}}\b^{2}}_{L^{\infty}(\R^{2})}\norm{\rho_{\gamma_{N}}}^{1/2}_{L^{1}(\R^{2})}\nonumber\\
&\leqslant \frac{CN}{R}\b\b \bA_{e}-\bA_{e}*\b F_{\hbar_{x}}\b^{2}\b\b_{L^{\infty}(\R^{2})} 
\end{align*}
where $\b\b W\b\b_{\mathfrak{S}^{p}}=\Tr\left [\b W\b^{p}\right ]^{\frac{1}{p}}$ is the Schatten $p$-norm. We used $\eqref{BKE}$ to bound the kinetic term and we now use Lemma~\ref{plancherel} to deduce
\begin{align}
\left| \Tr\left(\left(\bA_{e}-\bA_{e}*\b F_{\hbar_{x}}\b^{2} \cdot(-\im\hbar\nabla \right)\gamma_{N}\right)\right|
&\leqslant \frac{CN\hbar_{x}}{R}\norm{\Delta \bA_{e}}_{L^{\infty}}\leq CN^{\eta + 2\eps}.\label{erAe}
\end{align}

\medskip

\noindent\textbf{Estimate of~\eqref{potmagn}.} We have
\begin{align}
\norm{ (\b\bA_{e}\b^{2}-\b\bA_{e}\b^{2}*\b F_{\hbar_{x}}\b^{2})\rho_{\gamma_{N}}}_{L^{1}(\R^{2})}&\leqslant\norm{\b\bA_{e}\b^{2}-\b\bA_{e}\b^{2}*\b F_{\hbar_{x}}\b^{2}}_{L^{\infty}(\R^{2})}\norm{ \rho_{\gamma_{N}}}_{L^{1}(\R^{2})}\nn\\
&\leqslant CN\norm{\b\bA_{e}\b^{2}-\b\bA_{e}\b^{2}*\b F_{\hbar_{x}}\b^{2}}_{L^{\infty}(\R^{2})}\nn\\
&\leqslant CN\hbar_{x}\norm{\Delta \b\bA_{e} \b^{2}}_{L^{\infty}(\R^{2})} \leq C N^{2\eps}
\label{Aesq}
\end{align} 
where we used Cauchy-Schwarz's inequality and the estimate $\eqref{convv}$.
For the second term of~\eqref{potmagn} we use Assumption~\ref{HYP} and Lemma~\ref{plancherel} to write
\begin{align}
\norm{\left (V-V*\b F_{\hbar_{x}}\b^{2}\right )\rho(x)}_{L^{1}}&\leqslant C\hbar_{x}\int_{\R^{2}}\left \b\Delta V(x)\right \b\rho(x)\d x\nn\\
&\leqslant C \hbar_{x} \int_{\R^{2}}\left( \left| x \right| ^{s-2} + 1 \right) \rho(x)\d x\nn\\
&\leqslant C \hbar_{x} \int_{\R^{2}}\left( \left| x \right| ^{s} + 1 \right) \rho(x) \leq C N^{2\eps} \label{eq:error pot}
\end{align} 
by using Young's inequality 
$$r^{s-2} \leq \frac{s-2}{s} r^s+ \frac{2}{s}$$
the choice~\eqref{hxhp} and the a priori bound~\eqref{eq:a priori pot}.

\medskip

\noindent\textbf{Estimate of~\eqref{lastterm}.} We have
\begin{equation}
\alpha^{2}\bral\sum_{j=1}^{N}\sum_{k\neq j}\b\nabla^{\perp }w_{R}(x_{j}-x_{k})\b^{2}\ketr\leqslant \frac{C}{R^{2}} \leq C N^{2\eta}\nn
\end{equation}
and
\begin{equation}
\hbar_{p}\int_{\R^{2}}\b\nabla f(y)\b^{2}\d y\leqslant C N ^{-2\eps}\nn
\end{equation}
recalling that $f$ is a fixed function and the choice~\eqref{hxhp}. 

\medskip 

\noindent \textbf{Estimate of~\eqref{crossterm}}:
\begin{align}
4\alpha &\left \b\Tr\left [\left (\nabla^{\perp }w_{R}- \nabla^{\perp }w_{R}*\b F_{\hbar_{x}}\b^{2}*\b F_{\hbar_{x}}\b^{2}\right )\cdot (-\im\hbar\nabla_{1})\gamma_{N}^{(2)}\right ]\right \b\nn\\
&\leqslant C\alpha\norm{ (-\im\hbar\nabla )\sqrt{\gamma_{N}}}_{\mathfrak{S}^{2}}\norm{ \left (\nabla^{\perp }w_{R}- \nabla^{\perp }w_{R}*\b F_{\hbar_{x}}\b^{2}*\b F_{\hbar_{x}}\b^{2}\right )\sqrt{\gamma_{N}^{(2)}}}_{\mathfrak{S}^{2}}\nn\\
&\leqslant  C\alpha\hbar\frac{N}{R}\norm{\rho_{\gamma_{N}}^{(2)}}^{1/2}_{L^{1}}\norm{\left (\nabla^{\perp }w_{R}- \nabla^{\perp }w_{R}*\b F_{\hbar_{x}}\b^{2}*\b F_{\hbar_{x}}\b^{2}\right )}_{L^{\infty}}\nn\\
&\leqslant  C\alpha \hbar \frac{N^{2}}{R}\norm{\left (\nabla^{\perp }w_{R}- \nabla^{\perp }w_{R}*\b F_{\hbar_{x}}\b^{2}*\b F_{\hbar_{x}}\b^{2}\right )}_{L^{\infty}}
\label{1}
\end{align}
where we used~\eqref{BORNERHO1}. We then use Lemma~\ref{plancherel} and the estimate~\eqref{normsuplap} to get
\begin{align}
4\alpha &\left \b\Tr\left [\left (\nabla^{\perp }w_{R}- \nabla^{\perp }w_{R}*\b F_{\hbar_{x}}\b^{2}*\b F_{\hbar_{x}}\b^{2}\right )\cdot (-\im\hbar\nabla_{1})\gamma_{N}^{(2)}\right ]\right \b\leqslant CN\frac{\hbar \hbar_{x}}{R^{4}} \leq C N ^{4\eta + 2\eps -1/2}.\nn
\end{align}

\medskip

\noindent\textbf{Estimate of~\eqref{twomagcross}}. We denote
\begin{equation}
\mathrm{I} =\bA_{e}(x_{j})\cdot\nabla^{\perp }w_{R}(x_{j}-x_{k})-\int_{\R^{2}}\bA_{e}(x_{j}-u)\cdot\big (\nabla^{\perp}w_{R}*\b F_{\hbar_{x}}\b^{2}\big )(x_{j}-x_{k}-u)\b F_{\hbar_{x}}(u)\b^{2}\d u
\label{Ae}.
\end{equation}
Then
\begin{align}
4\alpha \Tr\left [\mathrm{K}\gamma^{(2)}_{N}\right ]&\leqslant C\alpha \norm{\rho^{(2)}_{\gamma_{N}}}_{L^{1}}\norm{\bA_{e}}_{L^{\infty}}\norm{\nabla^{\perp }w_{R}-\nabla^{\perp }w_{R}*\b F_{\hbar_{x}}\b^{2}*\b F_{\hbar_{x}}}_{L^{\infty}}\nonumber\\
&\leqslant C\alpha N^{2}\hbar_{x}\norm{\Delta \nabla^{\perp }w_{R}}_{L^{\infty}}\nonumber\\
&\leqslant \frac{CN\hbar_{x}}{R^{3}} \leq C N ^{3\eta + 2 \eps}
\label{2}
\end{align}
where we used~\eqref{BORNERHO1} and~\eqref{plancherel} combined with  the estimate~\eqref{normsuplap}. 

\medskip

\noindent\textbf{Estimate of~\eqref{treebodyterm}}. We denote
\begin{multline}
\mathrm{II}=W_{123}(x_{1},x_{2},x_{3})
\\- \int_{\R^{2}}\Big ( \nabla^{\perp }w_{R}*\b F_{\hbar_{x}}\b^{2}\Big)(x_{1}-x_{2}-u)\cdot\Big ( \nabla^{\perp }w_{R}*\b F_{\hbar_{x}}\b^{2}\Big) (x_{1}-x_{3}-u)\b F_{\hbar_{x}}(u)\b^{2}\d u.
\label{AA}
\end{multline}
Then
\begin{align}
6\alpha^{2}\Tr\left [\mathrm{L}\gamma^{(3)}_{N}\right ]&\leqslant \frac{CN \alpha^{2}}{R}\norm{\left (\nabla^{\perp }w_{R}-\nabla^{\perp }w_{R}*\b F_{\hbar_{x}}\b^{2}*\b F_{\hbar_{x}}\b^{2}\right )\rho_{\gamma_{N}}^{(2)}}_{L^{1}}\nonumber\\
&\leqslant \frac{CN\alpha^{2}}{R}\norm{\rho^{(2)}_{\gamma_{N}}}_{L^{1}}\norm{\nabla^{\perp }w_{R}-\nabla^{\perp }w_{R}*\b F_{\hbar_{x}}\b^{2}*\b F_{\hbar_{x}}\b^{2}}_{L^{\infty}}\nonumber\\
&\leqslant C\alpha^{2}\frac{N^{3}\hbar_{x}}{R}\norm{\Delta \nabla^{\perp }w_{R}}_{L^{\infty}} \nonumber\\
&\leqslant \frac{CN\hbar_{x}}{R^{4}}\leq C N^{4\eta + 2 \eps}
\label{ernaa}
\end{align}
where we used \cite[Lemma 2.1]{Girardot-19} and $\eqref{convv}$ combined with $\eqref{normsuplap}$ to bound the last norm. Collecting all the previous estimates leads to the result.
\end{proof}

\section{Energy lower bound: mean-field limit}\label{sec:lower MF}

This section is dedicated to a lower bound on~\eqref{en}. We have reduced the problem to a classical one in the previous section. We now perform a mean-field approximation of the classical energy~\eqref{eq:class ener}, showing that its infimum is (asymptotically) attained by a factorized Husimi function. The crucial point is to keep track of the semi-classical Pauli principle~\eqref{eq:Pauli semi}.

We first recall the Diaconis-Freedman theorem. We use it to rewrite many-particle probability measures as statistical superpositions of factorized ones. We then estimate quantitatively the probability for the measures in the superposition to violate the Pauli principle. This is the main novelty compared to the approach of \cite{FouLewSol-15}. We finally prove the convergence of the energy and states, Theorems~\ref{th1} and~\ref{th2}. 



\subsection{The Diaconis-Freedman measure}

For $Z_{N}\in \R^{4N}$ a set of points in the phase space, we define the associated empirical measure
\begin{equation}
\Emp_{Z_{N}}=\frac{1}{N}\sum_{i=1}^{N}\delta_{z_{i}}.
\end{equation}
The following originates in~\cite{DiaFre-80}.

\begin{theorem}[\textbf{Diaconis-Freedman}]\mbox{}\\
\label{DF}
Let $\mu_{N} \in \mathcal{P}_{\mathrm{sym}}(\R^{4N})$ the set of symmetric probability measures over $\R^{4N}$. Define the probability measure $P_{\mu_N}$ over $\mathcal{P}(\R^{4})$
\begin{equation}
P_{\mu_{N}}(\sigma ):=\int_{\R^{4N}}\delta_{\sigma =\Emp_{Z_{N}}}\d \mu_{N}(z_{1},\ldots,z_{N})
\label{PN}
\end{equation}
and set 
$$ 
\mut_N := \int_{\cP (\R^4)} \rho ^{\otimes N} \d P_{\mu_N} (\rho)
$$
with assocated marginals
\begin{equation}
\widetilde{\mu}^{(n)}_{N}=\int_{\rho \in \mathcal{P}(\R^4)} \rho^{\otimes n}\d P^{\mathrm{DF}}_{\mu_{N}}(\rho).
\label{DF1}
\end{equation}
Then
\begin{equation}
\norm{\mu_{N}^{(n)}-\widetilde{\mu}_{N}^{(n)}}_{\mathrm{TV}}\leqslant \frac{2n(n-1)}{N}
\label{DF2}
\end{equation}
in total variation norm 
$$
\norm{\mu}_{\mathrm{TV}}= \sup_{\phi\in C_{b}(\R^4), \norm{\phi}_{L^\infty}\leq 1 }\left \b \int_{\R^4}\phi\d \mu\right \b
$$
with $C_{b}\left (\Omega\right )$ the set of continuous bounded functions on $\Omega $. 

Moreover, if the sequence $\left(\mu^{(1)}\right)_{N}$ is tight, there exists a subsequence and a unique probability measure $P\in \mathcal{P}\left ((\mathcal{P}(\R^{4})\right )$ such that (up to extraction)
$$
P_{\mu_{N}}\rightharpoonup P
$$
and, for any fixed $n\in \mathbb{N}^{*}$
\begin{equation}
\widetilde{\mu}^{(n)}_{N}\rightharpoonup\int_{\rho \in \mathcal{P}(\Omega)} \rho^{\otimes n}\d P(\rho).
\label{DF3}
\end{equation}
\end{theorem}

The measure $P$ obtained in the limit is that appearing in the Hewitt-Savage theorem~\cite{HewSav-55}. See~\cite[Theorems 2.1 and 2.2]{Rougerie-spartacus,Rougerie-LMU} for proofs and references. We will use the following consequence of the Diaconis-Freedman construction.

\begin{lemma}[\textbf{First marginals of the Diaconis-Freedman measure}]\mbox{}\\
Let $\mu_{N}\in \cP_{\rm sym} (\R^{4N})$ and the associated Diaconis-Freedman approximation $\widetilde{\mu}_{N}$ be given by the previous theorem. Using the notation $X_{i}=(x_{i},p_{i})$ we have
\label{marginals}
\begin{align}
\widetilde{\mu}_{N}^{(1)}(x_{1},p_{1})&= \mu_{N}^{(1)}(x_{1},p_{1})\label{m1}\\
\widetilde{\mu}_{N}^{(2)}(x_{1},p_{1},x_{2},p_{2}) &=\frac{N-1}{N}\mu^{(2)}_{N}(X_{1},X_{2})+\frac{1}{N}\mu_{N}^{(1)}(X_{1})\delta_{(x_{1},p_{1})=(x_{2},p_{2})}\label{m2}\\
\widetilde{\mu}_{N}^{(3)}(X_{1},X_{2},X_{3})&=\frac{N(N-1)(N-2)}{N^{3}}\mu_{N}^{(3)}(X_{1},X_{2},X_{3})\nn\\
&+\frac{(N-1)}{N^{2}}\Big [\mu_{N}^{(2)}(X_{1},X_{3})\delta_{X_{1}=X_{2}} +\mu_{N}^{(2)}(X_{2},X_{1})\delta_{X_{2}=X_{3}}+\mu_{N}^{(2)}(X_{3},X_{2})\delta_{X_{3}=X_{1}}\Big ]\nn\\
&+\frac{1}{N^{2}}\mu_{N}^{(1)}(X_{1})\delta_{X_{1}=X_{2}}\delta_{X_{2}=X_{3}}
\end{align}
\end{lemma}

\begin{proof}
One can find $\eqref{m1}$ and $\eqref{m2}$ in \cite[Remark 2.3]{Rougerie-LMU}. It follows from~\cite[Equation 2.16]{Rougerie-spartacus,Rougerie-LMU} that
\begin{equation}\label{eq:expand DF}
\widetilde{\mu}_{N}(X_{N})=\int_{\R^{4N}}\mu_{N}(Z)\sum_{\gamma \in \Gamma_{N}}N^{-N}\delta_{X=Z_{\gamma}}\d Z =\int_{\R^{4N}}\mu_{N}(Z)\left (N^{-1}\sum_{j=1}^{N}\delta_{z_{j}=(x_{j},p_{j})}\right )^{\otimes N}\d Z
\end{equation}
where $Z_{\sigma}$ is the $4N$-uple $(x_{\sigma (1)},p_{\sigma (1)},...,x_{\sigma (N)},p_{\sigma (N)})=(z_{\sigma (1)},...z_{\sigma_{N}})$, ($z_{i}=(x_{i},p_{i})$) and $\Gamma_{N}$ the set of maps from $\{1,...,N\}$ to itself (which allows repeated indices). We look at the $n$-th marginal
\begin{align}
\widetilde{\mu}^{(n)}_{N}(X_{N})=\int_{\R^{4N}}\mu_{N}(Z)\left (N^{-1}\sum_{j=1}^{N}\delta_{z_{j}=X_{j}}\right )^{\otimes n}\d Z
\label{mu}.
\end{align}
In particular 
\begin{align*}
\widetilde{\mu}_{N}^{(1)}(x_{1},p_{1})&=N^{-1}\int_{\R^{4N}}\mu_{N}(Z)\left (\sum_{j=1}^{N}\delta_{z_{j}=X_{1}}\right )\d Z \nn\\
\widetilde{\mu}_{N}^{(2)}(x_{1},p_{1},x_{2},p_{2})&=N^{-2}\int_{\R^{4N}}\mu_{N}(Z)\left (\sum_{j=1}^{N}\delta_{z_{j}=X_{1}}\right )\left (\sum_{j=1}^{N}\delta_{z_{j}=X_{2}}\right )\d Z\nn \\
\widetilde{\mu}_{N}^{(3)}(X_{1},X_{2},X_{3})&=N^{-3}\int_{\R^{4N}}\mu_{N}(Z)\left (\sum_{j=1}^{N}\delta_{z_{j}=X_{1}}\right )\left (\sum_{j=1}^{N}\delta_{z_{j}=X_{2}}\right )\left (\sum_{j=1}^{N}\delta_{z_{j}=X_{3}}\right )\d Z .
\end{align*}
Expanding the $n$-th tensor power yields $\widetilde{\mu}_{N}^{(n)}$ in terms of the marginals of $\mu_{N}$ by an explicit calculation.
\end{proof}

We now apply the above constructions to the probability measure
\begin{equation}
m_{N}:=\frac
{m_{\Psi_{N}}^{(N)}(x_{1},p_{1},...,x_{N},p_{N})}{N! (2\pi\hbar)^{2N}}
\label{mN}
\end{equation} 
to obtain an approximation of the marginals (cf Lemma~\ref{propH})
\begin{equation}
m_{N}^{(k)}=\frac
{m_{\Psi_{N}}^{(k)}(x_{1},p_{1},...,x_{k},p_{k})}{N\left (N-1\right )...\left (N-k+1\right )(2\pi\hbar)^{2k}}.
\label{mn}
\end{equation} 
We use the notation $W_{12}$ and $W_{123}$ respectively defined in~\eqref{W12} and~\eqref{W123} to denote the two and the three body operators. In the former case, this is with a slight abuse of notation: we replace $p^{\bA}$ by the classical vector $p_1 + \bA_e(x)$ in the phase space.


\begin{lemma}[\textbf{Energy in terms of the Diaconis-Freedman measure}]\mbox{}\\
\label{Husfac}
Let $m_N \in \cP_{\rm sym} (\R^{4N})$ be as in~\eqref{mN}. Let $P^{\rm DF}_N\in \cP (\cP (\R^4))$ be the associated Diaconis-Freedman measure, as in Theorem~\ref{DF}.  Evaluating each term of~\eqref{ttt} gives
\begin{align}
\frac{1}{(2\pi )^{2}}\int_{\R^{4}} \left (\b p^{\bA}\b^{2}+V(x)\right )m^{(1)}_{f,\Psi_{N}}(x,p)\d x\d p&=\int_{\R^{4}}\int_{\mu\in \mathcal{P}(\R^{4})}\left (\b p^{\bA}\b^{2}+V(x)\right ) \d \mu (x,p) \d P^{\mathrm{DF}}_{N}(\mu )\label{boundofmu}\\
\frac{1}{(2\pi )^{4}}\int_{\R^{8}} W_{12} \;m^{(2)}_{f,\Psi_{N}}(x_{12},p_{12})\d x_{12}\d p_{12}&=\int_{\R^{8}}\int_{\mu\in \mathcal{P}(\R^{4})} W_{12}\;\d \mu^{\otimes 2}\d P^{\mathrm{DF}}_{N}(\mu )-\frac{C}{NR^{2}}\\
\frac{1}{(2\pi )^{6}}\int_{\R^{12}} W_{123}\; \d m^{(3)}_{f,\Psi_{N}}&=\int_{\R^{12}} \int_{\mu\in \mathcal{P}(\R^{4})}W_{123}\; \d \mu^{\otimes 3} \d P^{\mathrm{DF}}_{N}(\mu )-\frac{C}{NR^{2}}.
\end{align}
\end{lemma}
\begin{proof}
To apply the Diaconis-Freedman theorem $\eqref{DF}$ to the measures $m_{f,\Psi_{N}}^{(i)}$, $i=1,2,3$ they have to be normalized as in~\eqref{mn}. We then use the measure $m^{(i)}_{N}$ instead of $m_{f,\Psi_{N}}^{(i)}$ in our expressions, keeping track of the adequate normalization factors. Denoting $Z_{i}=\left (x_{i},p_{i}\right ) $ and recalling
$$
(2\pi\hbar)^{2k}N(N-1)...(N-k+1)m_{N}^{(k)}=m_{\Psi_{N}}^{(k)}
$$
we calculate
\begin{align*}
&\frac{1}{(2\pi )^{2}}\int_{\R^{4}}\b p^{\bA}_{1}\b^{2} m^{(1)}_{f,\Psi_{N}}(x_{1},p_{1})\d x_{1}\d p_{1}=\int_{\R^{4}}\int_{\mu\in \mathcal{P}(\R^{4})}\b p^{\bA}_{1}\b^{2} \mu(x_{1},p_{1})\d x_{1}\d p_{1} \d P^{\mathrm{DF}}_{N}(\mu )
\end{align*}
and
\begin{align}
\frac{1}{(2\pi )^{4}}&\int_{\R^{8}}(p^{\bA}_{1})\cdot\nabla^{\perp }w_{R}(x_{1}-x_{2}) m^{(2)}_{f,\Psi_{N}}(Z_{1},Z_{2})\d Z_{1}\d Z_{2}\nn\\
&=N(N-1)\hbar^{4}\int_{\R^{8}}(p^{\bA}_{1})\cdot\nabla^{\perp }w_{R}(x_{1}-x_{2}) m^{(2)}_{N}(Z_{1},Z_{2})\d Z_{1}\d Z_{2}\nn\\
&=N^{2}\hbar^{4}\int_{\R^{8}}\int_{\mu\in \mathcal{P}(\R^{4})} (p^{\bA}_{1})\cdot\nabla^{\perp }w_{R}(x_{1}-x_{2}) \d \mu^{\otimes 2}\d P^{\mathrm{DF}}_{N}(\mu )\nn\\
&-N\hbar^{4}\int_{\R^{8}}(p^{\bA}_{1})\cdot\nabla^{\perp }w_{R}(x_{1}-x_{2}) m^{(1)}_{N}(X_{1})\delta_{Z_{1}=Z_{2}}\d Z_{1}\d Z_{2}.\label{err1}
\end{align}
Moreover
\begin{align}
&\int_{\R^{12}} \nabla^{\perp } w_{R}(x_{1}-x_{2})\cdot\nabla^{\perp }w_{R}(x_{1}-x_{3}) m^{(3)}_{N}(Z_{1},Z_{2},Z_{3})\d Z_{1}\d Z_{2}\d Z_{3}\nn\\
&=\frac{N^{2}}{(N-1)(N-2)}\int_{\R^{12}} \int_{\mu\in \mathcal{P}(\R^{4})}\nabla^{\perp }w_{R}(x_{1}-x_{2})\cdot\nabla^{\perp }w_{R}(x_{1}-x_{3}) \d \mu^{\otimes 3} \d P^{\mathrm{DF}}_{N}(\mu )\nn\\
&-\frac{1}{(N-2)}\int_{\R^{12}}W_{123}\left (m_{N}^{(2)}(Z_{1},Z_{3})\delta_{Z_{1}=Z_{2}}
 +m_{N}^{(2)}(Z_{2},Z_{1})\delta_{Z_{2}=Z_{3}}+m_{N}^{(2)}(Z_{3},Z_{2})\delta_{Z_{3}=Z_{1}}\right )\label{err2}\\
 &-\frac{1}{(N-1)(N-2)}\int_{\R^{12}}W_{123}m_{N}^{(1)}(Z_{1})\delta_{Z_{1}=Z_{2}}\delta_{Z_{2}=Z_{3}}.\label{err3}
\end{align}
We will discard the error terms in~\eqref{err1}~\eqref{err2} and~\eqref{err3}
using~\cite[Lemma 2.1]{Girardot-19}   
$$\norm{\nabla^{\perp }w_{R}}_{L^{\infty}} \leqslant R^{-1}$$
and the kinetic energy bound~\eqref{BKE}. The latter, combined with equation~\eqref{tohusimi} gives
\begin{equation}
\iint_{\R^4}\b p \b^{2}\d m^{(1)}_{f,\Psi_{N}} (x,p)\leqslant \frac{C}{R^{2}}. 
\label{enerk2}
\end{equation}
Applying the Cauchy-Schwarz inequality to the two particles term we obtain
\begin{align*}
&\left \b\frac{1}{N}\int_{\R^{8}}(p^{\bA}_{1})\cdot\nabla^{\perp }w_{R}(x_{1}-x_{2}) m^{(1)}_{N}(Z_{1})\delta_{Z_{1}=Z_{2}}\d Z_{1}\d Z_{2}\right \b \\
&\leqslant \frac{C\b\b \bA_{e}\b\b_{L^{\infty}}}{NR}+\frac{1}{N}\left (\int_{\R^{8}}\left \b\nabla^{\perp }w_{R}(x_{1}-x_{2})\right \b^{2} m_{N}^{(1)}\delta_{Z_{1}=Z_{2}}\d Z_{1}\d Z_{2}\right )^{1/2}\left (\int_{\R^{4}}\b p \b^{2}\d m^{(1)}_{N}\right )^{1/2}\\
&\leqslant \frac{C}{NR^{2}}.
\end{align*}
As for the three-particles term we have
\begin{align*}
&\frac{1}{N}\left \b \int_{\R^{12}}W_{123}\left (m_{N}^{(2)}(Z_{1},Z_{3})\delta_{Z_{1}=Z_{2}}\right )\right \b\leqslant\frac{C}{NR^{2}}\\
 &\frac{1}{N^{2}}\left\b\int_{\R^{12}}W_{123}\left (m_{N}^{(1)}(Z_{1})\delta_{Z_{1}=Z_{2}}\delta_{Z_{2}=Z_{3}}\right )\right \b\leqslant\frac{C}{N^{2}R^{2}}.
\end{align*}
\end{proof}

\subsection{Quantitative semi-classical Pauli principle}

Lemma $\eqref{Husfac}$ allows to write the energy as an integral of $\E_{V}^{R}$ over $\mathcal{P}(\R^{4})$ plus some negligible error terms. Assuming $R=N^{-\eta}$ for some $\eta < 1/4$, we get from Proposition \ref{semiclassenergy} and the above lemma:
\begin{equation}
\frac{\bral \Psi_{N},H_{N}^{R}\Psi_{N}\ketr}{N}\geqslant \int_{\mathcal{P}(\R^{4})}\E^{R}_{V}[\mu]\d P_{m_{N}}^{\mathrm{DF}}(\mu)- o_N (1)
\label{enr1}
\end{equation}
where
\begin{equation}
\E^{R}_{V}[\mu]=\int_{\R^{4}}\left| p+\bA_{e}(x)+\beta\bA^{R}[\rho](x)\right|^{2}\mu(x,p)\d x\d p +\int_{\R^2}V(x)\rho(x)\d x
\label{evl}
\end{equation}
with
\begin{equation}
\rho(x)=\int_{\R^{2}}\mu (x,p)\d p.
\label{ro}
\end{equation}
Note that $$\E^{R}_{V}\left [\frac{\mu}{\left (2\pi\right )^{2}} \right ]=\E^{R}_{\mathrm{Vla}}[\mu]$$ defined in $\eqref{EVLA}$.
However, $P_{N}^{\mathrm{DF}}$ only charges empirical measures which do not satisfy the Pauli principle~\eqref{pauli}. To circumvent this issue we divide the phase-space 
\begin{equation}
\cup_{m\in\mathbb{N}}\Omega_{m}=\R^{4}
\label{grid}
\end{equation} 
in hyperrectangles labeled $(\Omega_m)_{m\in \mathbb{N}}$. We take them to have side-length $l_x$ in the space coordinates and $l_p$ in the momentum coordinates, ensuring that 
\begin{equation}\label{eq:size tile}
\left \b\Omega_{m}\right \b=l_{x}^{2}l_{p}^{2}=N^{\beta }. 
\end{equation}
The parameter $\beta >0$ and the lengths $l_x,l_p$ will be chosen in the sequel.

If each $\Omega_{m}$ contains less than $(1+\epsilon)(2\pi)^{-2}\b\Omega_{m}\b $ points (for some small $\eps$) our  measure will approximately satisfy the Pauli principle. We thus have to estimate the probability for a box to have the right density of points. This is the purpose of the following lemmas. We denote for a given measure $\mu $:
\begin{equation}
\mathbb{P}_{\mu }\left (\Omega\right )=\int_{\Omega }\d \mu.
\end{equation}

\begin{theorem}[\textbf{Probability of violating the Pauli principle in a phase-space box}]\mbox{}
\label{bigprob}
Let $m_N$ be as in~\eqref{mN} and $P^{\rm DF}_N$ the associated Diaconis-Freedman measure. Recall the tiling~\eqref{grid} and assume~\eqref{eq:size tile} with $\beta < 1$. For any 
$$
0 < \delta < \frac{1-\beta}{2}
$$
we have, for constant $C_\delta >0, c_\delta >0$, that
\begin{equation}
\mathbb{P}_{P^{\mathrm{DF}}_{N}}\left (\left \{ \Emp_{Z_{N}},\int_{\Omega_{m}}\Emp_{Z_{N}}\geqslant \frac{(1+\epsilon )}{\left (2\pi\right )^{2}}\b\Omega_{m}\b\right \}\right )\leqslant C_\delta e^{-c_\delta N^{\delta}\ln\left (1+\epsilon\right )}.
\end{equation}
\end{theorem}

Note that the condition $\beta <1$ means $l_x^2 l_p^2 \geq N^{-1}$. Since the average interparticle distance in phase-space is $N^{-1/4}$,  our hyperrectangles typically contain a large number of particles. 


We start the proof of the above theorem with the 

\begin{lemma}[\textbf{Expansion of the Diaconis-Freedman measure}]\mbox{}\\
\label{lemme_gamma}
Let 
$$\widetilde{m}_{N} = \int_{\cP (\R^4)} \mu^{\otimes N} \d P^{\rm DF}_N (\mu)$$ 
be associated to~\eqref{mN} as in Theorem~\ref{DF}, and $\widetilde{m}^{(n)}_{N}$ the associated marginals. For any $\Omega \subset \R^4$ we  have
\begin{equation}
\int_{\Omega^{n}}\widetilde{m}^{(n)}_{N}\leqslant N^{-n}\sum_{k=1}^{n}\frac{\b \Omega\b^{k}S(n,k)}{\left (2\pi \hbar\right )^{2k}}
\label{intmu}
\end{equation}
with $S(n,k)$ the Stirling number of the second kind, the number of ways to partition a set of $n$ objects into $k$ non-empty subsets.
\end{lemma}

We refer to~\cite{Bona-04} for combinatorial definitions and estimates. 

\begin{proof}[Proof of Lemma~\ref{lemme_gamma}]
The measure $\widetilde{\mu}_{N}$ constructed from a $\mu_N \in \cP_{\rm sym} (\R^{4N})$ in Theorem~\ref{DF} satisfies~\eqref{eq:expand DF}. Hence
\begin{equation}
\widetilde{\mu}^{(n)}_{N} =N^{-n} \sum_{k=1}^{n}\sum_{\gamma\in\Gamma_{n,k}}\int_{\R^{4(N-n)}} \mu_{N}(Z)\delta_{X_{N}=Z_{\gamma}}\d Z
\label{muenforme}
\end{equation}
with $\Gamma_{n,k}$ the set of ordered samples of $n$ indices amongst $N$ (allowing repetitions) with exactly $k$ distinct indices. Integrating on each side of~\eqref{muenforme} and using the symmetry of (recall Lemma~\ref{propH})
$$\mu^{(k)}_{N}=m_{N}^{(k)}\leqslant \frac{1}{\left (2\pi \hbar\right )^{2k}N(N-1)...(N-k+1)}$$ 
we obtain
$$
\int_{\Omega^{n}}\widetilde{m}^{(n)}_{N}\leqslant N^{-n}\sum_{k=1}^{n}\frac{\b\Gamma_{n,k}\b\b \Omega\b^{k}}{\left (2\pi \hbar\right )^{2k}N(N-1)...(N-k+1)}.
$$
Now we calculate the cardinal $\b \Gamma_{n,k}\b$. For a given $k$, choose $k$ indices among $N$ (without order) to, next, be distributed in $n$ boxes. The first step gives us
$$\frac{N!}{(N-k)!k!}$$
choices. For each of these choices we have to distribute $n$ balls in $k$ boxes (boxes can be empty) which corresponds to the number of surjections from $\{1,2,...,n\}$ to $\{1,2,...,k\}$, $M(n,k)$ (with $k\leqslant n$).  Then
$$
\b\Gamma_{n,k}\b =\frac{N!}{(N-k)!k!}M(n,k)
$$
with \cite[Theorem 1.17]{Bona-04}
$$
M(n,k)=k!S(n,k)
$$
and $S(n,k)$ is the Stirling number of the second kind~\cite[Lemma 1.16]{Bona-04}. Thus
$$
\b\Gamma_{n,k}\b =\frac{N!}{(N-k)!}S(n,k).
$$
\end{proof}


We now need a rough control of the sum appearing in~\eqref{intmu} when the set $\Omega$ is not too small.

\begin{lemma}[\textbf{Control of the sum}]\mbox{}\\
\label{contr}
Take $\b\Omega_{m}\b = N^{-\beta}$ with $\beta <1$. Choose some
$$
0 < \delta < \frac{1 - \beta}{2}
$$ 
and define
$$
n = \left\lfloor N^{\delta} \right\rfloor
$$
the integer part of $N^\delta$. We have that
\begin{equation*}
\left (2\pi\right )^{2n}N^{-n}\sum_{k=1}^{n}\frac{\b\Omega_{m}\b^{(k-n)}S(n,k)}{\left (2\pi \hbar\right )^{2k}}\leqslant C N^{\delta'}
\label{poly}
\end{equation*}
for some exponent $\delta' >0$ depending only on $\delta$.
\end{lemma}

\begin{proof}
We recall Stirling's formula
\begin{equation}
C_{1}\sqrt{N}\left (\frac{N}{e}\right )^{N}\leqslant N!\leqslant C_{2}\sqrt{N}\left (\frac{N}{e}\right )^{N}
\label{St}
\end{equation}
and a bound for the Stirling number of the second kind from~\cite{RenDob-69}(see also~\cite{Menon-73,Lieb-68})
\begin{equation}
S(n,k)\leqslant C\begin{pmatrix}
   n \\
   k
\end{pmatrix}k^{(n-k)}.
\label{Ssk}
\end{equation}
Inserting these bounds we have that 
$$
\left (2\pi\right )^{2n}N^{-n} \frac{\b\Omega_{m}\b^{(k-n)}S(n,k)}{\left (2\pi \hbar\right )^{2k}} \leq C \sqrt{\frac{n}{k(n-k)}} \exp(f_\beta (k))
$$
with 
\begin{align*}
f_\beta (k) &= (n-k)\bigg( \log k + (\beta-1) \log N + 2 \log (2\pi) \bigg) + n\log n - k \log k -(n-k)\log (n-k) \\
&= (n-k) \bigg( 2 \log k + \log n + (\beta-1) \log N + 2 \log (2\pi) - \log (n-k)\bigg) + k\log n - n\log k\\
&= (n-k) \bigg( \log k + \log n + (\beta-1) \log N + 2 \log (2\pi) - \log (n-k) \bigg) + k \log\left( 1 + \frac{n-k}{k}\right)\\
&\leq (n-k) \bigg( \log k + \log n + (\beta-1) \log N + 2 \log (2\pi) + 1 - \log (n-k) \bigg),
\end{align*}
using $\log (1+x) \leq x$. Since $k\leq n \sim N^{\delta}$ we find that the leading order of the above is 
$$ f_\beta (k) \lesssim (n-k) \left( 2 \delta + \beta - 1 \right)\log N \leq 0$$
under our assumptions. Hence we find 
$$\left (2\pi\right )^{2n}N^{-n}\sum_{k=1}^{n}\frac{\b\Omega_{m}\b^{(k-n)}S(n,k)}{\left (2\pi \hbar\right )^{2k}} \leq C \sum_{k=1} ^n \sqrt{\frac{n}{k(n-k)}},$$
which yields the result. 
\end{proof}

Now we can complete the

\begin{proof}[Proof of Theorem~\ref{bigprob}]
We call 
\begin{equation}\label{eq:bad set}
\Gamma=\left \{\Emp_{Z_{N}},\int_{\Omega_{m}}\Emp_{Z_{N}}\geqslant  \frac{(1+\epsilon )}{\left (2\pi\right )^{2}}\b\Omega_{m}\b\right \}
\end{equation}
the set of all empirical measure violating the Pauli principle in $\Omega_{m}$ by a finite amount. We have, using $\eqref{DF1}$ that, for any $n\leq N$,
\begin{equation*}
\frac{(1+\epsilon )^{n}}{\left (2\pi\right )^{2n}}\b\Omega_{m}\b^{n}\mathbb{P}_{P_{N}^{\mathrm{DF}}}\left (\Gamma\right ) \leq \int_{\rho\in\Gamma}\left (\int_{\Omega_{m}}\rho^{\otimes n}\right )\d P_{N}^{\mathrm{DF}}(\rho) = \int_{(\Omega_{m})^{n}}\widetilde{\mu}^{(n)}_{N}. 
\end{equation*}
Next, using Lemma~\ref{lemme_gamma} we obtain
\begin{equation*}
\mathbb{P}_{P_{N}^{\mathrm{DF}}}\left (\Gamma\right )\leq (1+\epsilon )^{-n}  \left (2\pi\right )^{2n}N^{-n}\sum_{k=1}^{n}\frac{\b\Omega_{m}\b^{(k-n)}S(n,k)}{\left (2\pi \hbar\right )^{2k}}.
\end{equation*}
Choosing $n$ as in Lemma~\ref{contr} and inserting the latter result we deduce
\begin{equation*}
\mathbb{P}_{P_{N}^{\mathrm{DF}}}\left (\Gamma\right )\leqslant C\frac{N^{\delta'}}{\left (1+\epsilon\right )^{n}}\leq CN^{\delta'}e^{-N^{\delta}\ln\left (1+\epsilon\right )}
\end{equation*}
and the result follows.
\end{proof}

\subsection{Averaging the Diaconis-Freedman measure}

We divide $\R^{4}$ into hyperrectangles 
$$\Omega_m = \Omega_{m_x}\times \Omega_{m_p}$$
as in~\eqref{grid}. We set 
\begin{equation}
\b\Omega_{m}\b=\b\Omega_{m_{x}}\b\b\Omega_{m_{p}}\b=l_{x}^{2}l_{p}^{2}=N^{-\beta}
\label{division}
\end{equation}
for some $0<\beta < 1$ to be fixed later. We know from Theorem~\ref{bigprob} that the probability for one such hyperrectangle to violate the Pauli principle is exponentially small. This means that the contribution to~\eqref{enr1} of empirical measures 
$$
\mu = \frac{1}{N} \sum_{j=1} ^N \delta_{z_j}
$$
with $Z_N = (z_1,\ldots,z_j)\in \Gamma$ (see~\eqref{eq:bad set}) will be negligible (by the union bound), provided the number of hyperrectangles is not too large. The next steps of our proof are thus
\begin{itemize}
 \item to reduce estimates to the contribution of a finite set of phase-space. 
 \item in the latter set, to average empirical measures to replace $\mu$ in~\eqref{enr1} by true Pauli-principle abiding measures (in the sense of~\eqref{eq:Pauli semi}).
\end{itemize}
As for the first step, let $L>0$ be a number of the form
\begin{equation}
L=nN^{-\frac{\beta}{4}}
\label{L}
\end{equation}
for some $n\in\mathbb{N}$ and define the hypercube
$$
S_{L}=\left [-L,L\right ]^{4}
$$ 
of side $2L$ centered at the origin. This way the number of $\Omega_{m}$-hyperrectangles contained within $S_{L}$ is $2^{4}n^{4}$. We first discard the energetic contribution of $S_{L}^{c}$.

\begin{lemma}[\textbf{Reduction to the energy in the phase-space hypercube $S_{L}$}]\mbox{}\\
\label{energyinbox}
Let $\sigma >0$ and $\mu$ a probability measure satisfying
\begin{align}
\int_{\R^{4}}\left (\b p \b^{2}+V(x) \right )\d \mu(x,p)\leqslant D.
\label{finitude}
\end{align}
Its energy can be bounded from below as
\begin{equation}
\E^{R}_{V}\left [\mu \right ]\geqslant\left (1-\sigma\right )\int_{\R^{4}}\left (\left |p ^{\bA} +\beta\bA^{R}\left [\1_{S_{L}}\mu\right ]\right| ^{2}+V\right )\1_{S_{L}}\d \mu(x,p) -C \frac{D^2}{\sigma R^{2} \inf(L^{4},L^{2s})}
\end{equation}
where 
$$
\bA^{R}\left [\1_{S_{L}}\mu\right ] = \int_{\R^{4}}\nabla^{\perp}w_{R}(x-y)\1_{S_{L}}(y,p)\d \mu(y,p) =\bA^R \left[ \rho_{\1_{S_{L}}\mu} \right]
$$
and 
$$
p^{\bA} = p + \bA_e(x).
$$
\end{lemma}

The parameter $\sigma $ will tend to $0$ at the very end of the proof (after we take $N\to \infty,L\to \infty,R\to 0$ while ensuring $R^6 \inf(L^{4},L^{2s}) \to \infty$).

\begin{proof}
First, the positivity of the integrand gives
\begin{equation*}
\E^{R}_{V}\left [\mu \right ]\geqslant\int_{S_{L}}\left (\left| p^{\bA }+\bA^{R}\left [\mu\right ]\right| ^{2}+V\right )\d \mu(x,p) 
\end{equation*}
Then we split $\bA^{R}\left [\mu\right ]$ into a contribution from $S_{L}$ and one from its complement,
\begin{equation*}
\bA^{R}\left [\mu\right ]=\bA^{R}\left [\1_{S_{L}}\mu\right ]+\bA^{R}\left [\1_{S^{c}_{L}}\mu\right ],
\end{equation*}
obtaining
\begin{align*}
\left (p^{\bA}+\bA^{R}\left [\1_{S_{L}}\mu\right ]+\bA^{R}\left [\1_{S^{c}_{L}}\mu\right ]\right )^{2}&=\left (p^{\bA}+\bA^{R}\left [\1_{S_{L}}\mu\right ]\right )^{2}+\left \b\bA^{R}\left [\1_{S^{c}_{L}}\mu\right ]\right \b^{2}\\
&+2\left (p^{\bA}+\bA^{R}\left [\1_{S_{L}}\mu\right ]\right )\cdot \bA^{R}\left [\1_{S^{c}_{L}}\mu\right ].
\end{align*}
Next we use that $ab\leqslant \frac{a^{2}}{2\sigma }+\frac{\sigma b^{2}}{2}$
to deduce
\begin{equation*}
2\left (p^{\bA}+\bA^{R}\left [\1_{S_{L}}\mu\right ]\right )\cdot \bA^{R}\left [\1_{S^{c}_{L}}\mu\right ]\geqslant-\sigma\left (p^{\bA}+\bA^{R}\left [\1_{S_{L}}\mu\right ]\right )^{2}-\frac{1}{\sigma }\left \b \bA^{R}\left [\1_{S^{c}_{L}}\mu\right ]\right \b^{2}.
\end{equation*}
But
\begin{align*}
\int_{S_{L}}\left \b\bA^{R}\left [\1_{S^{c}_{L}}\mu\right ]\right \b^{2}\left (x\right )\d \mu(x,p)&\leqslant \norm{ \bA^{R}\left [\1_{S^{c}_{L}}\mu\right ]}^{2}_{L^{\infty}}\\
&\leqslant \norm{ \nabla^{\perp }w_{R}*\1_{S_{L}^{c}}\mu}^{2}_{L^{\infty}}\\
&\leqslant \norm{\1_{S_{L}^{c}}\mu }_{L^{1}}^{2} \norm{ \nabla^{\perp } w_{R}}^{2}_{L^{\infty}}\\
&\leqslant \frac{C}{R^{2}}\norm{ \1_{S_{L}^{c}}\mu }_{L^{1}}^{2}
\end{align*}
where we used Young's inequality combined with~\cite[Lemma 2.1]{Girardot-19}. The measure $\mu $ satisfies $\eqref{finitude}$ so, using Assumption~\ref{HYP},
\begin{equation}
\int_{\R^{4}}  \1_{S_{L}^{c}}\mu (x,p)  \d x\d p\leqslant \frac{D}{\inf(L^2,L^s)}\nn
\end{equation}
which concludes the proof.
\end{proof}

We now turn to the averaging procedure turning empirical measures based on configurations from $\Gamma^c$ (recall~\eqref{eq:bad set}) into bounded measures satisfying the Pauli principle~\eqref{eq:Pauli semi}:

\begin{definition}[\textbf{Averaging map}]\mbox{}\label{def:Ave}\\
Let $\mu$ be a positive measure on the phase-space hypercube $S_L = [-L,L]^4$. We associate to it its local average with respect to the grid~\eqref{grid}:
\begin{equation}
\Ave\left [\mu\right ]=\sum_{m=1}^{16n^{4}}\1_{\Omega_{m}}\int_{\Omega_{m}}\frac{\mu}{\b\Omega_{m}\b}
\label{mutilde}
\end{equation}
with the associated position-space density (recall that $\Omega_m = \Omega_{m_x} \times \Omega_{m_p}$)
\begin{equation}
\rho_{\Ave[\mu]} := \int\Ave\left [\mu \right ]\d p=\sum_{m=1}^{16n^{4}}\1_{\Omega_{m_{x}}}\int_{\Omega_{m}}\frac{\mu }{\b\Omega_{m_{x}}\b}.
\label{rhotilde}
\end{equation}
\end{definition}

Note that $\Ave\left [\mu \right ]$ lives on $S_L$ by definition. We next quantify the energetic cost of the above procedure:


\begin{theorem}[\textbf{Averaging the empirical measure}]\mbox{}\\
\label{boundedmeasure}
Assume $\beta < 1$ in~\eqref{eq:size tile} as well as 
\begin{equation}\label{eq:lxlp}
l_x \ll R , \quad l_p \ll 1 \mbox{ in the limit } N\to \infty. 
\end{equation}
Let $Z_{N}=(x_{1},p_{1},...,x_{N},p_{N})$ with $\Emp_{Z_N} \in \Gamma_\epsilon^c$,   
\begin{equation}\label{eq:bad set 2}
\Gamma_\epsilon =\left \{\Emp_{Z_{N}},\exists \Omega_m \subset S_L \mbox{ such that }\int_{\Omega_{m}}\Emp_{Z_{N}}\geqslant  \frac{(1+\epsilon )}{\left (2\pi\right )^{2}}\b\Omega_{m}\b\right \}.
\end{equation}
Assume that~\eqref{finitude} holds for $\mu = \Emp_{Z_N}$. Then, for any $\sigma >0$,
\begin{multline}
\E^{R}_{V}[\Emp_{Z_N}] \geqslant\left (1-\sigma\right )\left (1-o_{N}(1)\right )\E^{R}_{V}\left [ \Ave\left [\Emp_{Z_{N}}\right ]\right ]
-o_N (1) -\frac{CD^2}{\sigma R^{2} \inf(L^{4},L^{2s})}.
\label{TheTheorem}
\end{multline} 
\end{theorem}

\begin{proof}
In this proof we denote 
$$
\mu = \Emp_{Z_N}
$$
for brevity. We first use Lemma~\ref{energyinbox} to write
\begin{equation*}
\E^R_{V}\left [\mu \right ]\geqslant (1-\sigma )\E^R_{V}\left [\1_{S_{L}}\mu \right ]-C\frac{D^2}{\sigma R^{2} \inf(L^{4},L^{2s})}
\end{equation*}
and now estimate $\left \b\E^R_{V}\left [\1_{S_{L}}\mu \right ]-\E^R_{V}\left [\Ave\left[\mu\right] \right ]\right \b$.
We expand $\E^R_{V}$ $\eqref{evl}$ and proceed term by term. Recall from~\eqref{L} that 
$$
S_{L}=\cup_{m=1}^{2^{4}n^{4}}\Omega_{m}.
$$
We first calculate
\begin{align*}
\int_{S_{L}}\left ((p^{\bA})^{2}+V(x)\right )\d (\mu-\Ave\left [\mu\right ])&=\int_{S_{L}}\left (\left (p^{\bA}\right )^{2}+V(x)\right )\Emp_{Z_{N}}\d x\d p\\
& -\frac{1}{N}\sum_{m=1}^{16n^{4}}\int_{\Omega_{m}}\frac{\left (p^{\bA}\right )^{2}+V(x)}{\b \Omega_{m}\b}\d x\d p\int_{\Omega_{m}}\sum_{i=1}^{N}\delta_{z_{i}=(x,p)}\d x\d p .
\end{align*} 
We denote
\begin{align*}
G_{N}=\left \{i, z_{i}\in S_{L}\right \}\\
G^{c}_{N}=\left \{i, z_{i}\in S^{c}_{L}\right \}
\end{align*} 
and let $m(i)$ be the label of the box $\Omega_m$ particle $z_i$ is in. This way
$$
\sum_{i=1}^{N}\int_{\Omega_{m}}\delta_{z_{i}=(x,p)}\d x\d p=\sum_{i\in G_{N}}\delta_{m=m(i)}
$$
and
\begin{align*}
\int_{S_{L}}\left ((p^{\bA})^{2}+V(x)\right )\d (\mu-\Ave\left [\Emp_{Z_{N}}\right ])&=\frac{1}{N}\sum_{i\in G_{N}}\left (\left (p_{i}^{\bA}\right )^{2}+V(x_{i})\right )\\
 -\int_{\Omega_{m(i)}}\frac{\left (p^{\bA}\right )^{2}+V(x)}{\b \Omega_{m(i)}\b}\d x\d p .
\end{align*} 
By the mean value theorem applied in each $\Omega_{m} $ there exists a point $\widetilde{z}_{i} = (\widetilde{x}_{i},\widetilde{p}_{i})\in \Omega_{m(i)}$ such that
\begin{equation*}
\int_{\Omega_{m(i)}}\frac{\left (p^{\bA}\right )^{2}+V(x)}{\b \Omega_{m(i)}\b}\d x\d p=\left (\widetilde{p}_{i}^{\bA}\right )^{2}+V(\widetilde{x}_{i})
\end{equation*}
and we have, using Assumption~\ref{HYP} to control the variations of $\bA_e$,
\begin{align*}
&\left \b\int_{S_{L}}\left ((p^{\bA})^{2}+V(x)\right )\d (\Ave\left [\Emp_{Z_{N}}\right ]-\mu)\right \b\leqslant\frac{1}{N}\sum_{i\in G_{N}}\left\b\left ((p_{i}^{\bA})^{2}- (\tilde{p}_{i}^{\bA})^{2}+V(x_{i})-V(\tilde{x}_{i})\right )\right \b\nn\\
&\leqslant \frac{C}{N}\sum_{i\in G_{N}} \left(\sup_{z=(x,p)\in \Omega_{m(i)}} l_p \b p \b  + l_{x} \sup_{z=(x,p)\in \Omega_{m(i)}}\left \b \nabla V(x) \right \b \right) + C l_{x}\nn\\
&=\frac{C}{N}\sum_{\Omega_m \subset S_L} \left(N_m l_p \sup_{z=(x,p)\in \Omega_{m}} \b p \b  + N_m l_{x} \sup_{z=(x,p)\in \Omega_{m}}\left \b \nabla V(x) \right \b \right) + C l_{x}
\end{align*}
where we denote 
\begin{equation}\label{eq:Riemman 1}
 N_m := \sharp\left\{ z_i \in Z_N \cap \Omega_m\right\} \leq C|\Omega_m|.
\end{equation}
Since our assumptions imply that for any $\Omega_m \subset S_L$
\begin{equation}\label{eq:Riemman 2}
 N_m \leq C|\Omega_m|
\end{equation}
we may recognize Riemann sums and use Assumption~\ref{HYP} again to deduce
\begin{align}
&\left \b\int_{S_{L}}\left ((p^{\bA})^{2}+V(x)\right )\d (\Ave\left [\Emp_{Z_{N}}\right ]-\mu)\right \b \nn\\ 
&\;\;\;\;\;\;\;\;\;\;\;\;\;\;\leq Cl_{p}\int_{\R^{4}}\b p\b \d \Ave\left[\mu\right]+ Cl_{x}\int_{\R^{4}}\b x\b^{s-1} \d \Ave\left[\mu\right]+ Cl_{p}\nn\\
&\;\;\;\;\;\;\;\;\;\;\;\;\;\;\leq Cl_{p}\int_{\R^{4}}\b p\b^{2} \d \Ave\left[\mu\right]+ Cl_{x}\int_{\R^{4}}V \d \Ave\left[\mu\right] +Cl_{x}+Cl_{p}
\label{renvoirienman}
\end{align}
using Young's inequality 
$$r^{s-1} \leq \frac{s-1}{s} r^s+ \frac{1}{s}.$$
We next treat the two body operator $W_{12}$ defined in~\eqref{W12} (understanding now $p$ as a momentum variable rather than an operator) in the same way. We set
$$\d x=\d x_{1}\d x_{2}, \quad \d p=\d p_{1}\d p_{2}.$$
By the mean-value theorem, for each $(x_i,p_i)\in \Omega_{m(i)}$ there exists some $\widetilde{z}_{i}= (\widetilde{x}_{i},\widetilde{p}_{i}) \in \Omega_{m(i)}$ such that
\begin{align}
\int_{\R^{8}}W_{12}&\;\d \left (\Ave\left [\Emp_{Z_{N}}\right ]^{\otimes 2}-\left (\1_{S_{L}}\mu\right )^{\otimes 2}\right )=\int_{S_{L}\times S_{L}}W_{12}\left [\mu^{\otimes 2}-\Ave\left [\Emp_{Z_{N}}\right ]^{\otimes 2}\right ]\d x\d p\nn \\
&=\frac{1}{N^{2}}\sum_{i\in G_{N}}\sum_{j\in G_{N}}\Big [\left (p_{i} +\bA_{e}(x_{i})\right ).\nabla^{\perp}w_{R}(x_{i},x_{j})-\left (\widetilde{p}_{i}+\bA_{e}(\widetilde{x}_{i})\right ).\nabla^{\perp}w_{R}(\widetilde{x}_{i},\widetilde{x}_{j})\Big ].
\label{sepsl}
\end{align}
Expanding the functions around the points  $\widetilde{z}_{i}$ of the sum there exists a  $c_{i}\in \Omega_{m(i)}$ such that
\begin{align*}
\left \b \nabla^{\perp}w_{R}(x_{i})-\nabla^{\perp}w_{R}(\widetilde{x}_{i})\right \b &\leqslant J_{\nabla w_{R}}(c_{i})\cdot\left (x_{i}-\widetilde{x}_{i}\right )\\
\left \b p_{i}-\widetilde{p}_{i}\right \b &\leqslant l_{p}
\end{align*}
with $J_{\nabla w_{R}}$ the Jacobian matrix of $\nabla w_R$. We denote
\begin{equation}\label{eq:def W}
W = \partial^{2}_{u} w_{R}
\end{equation}
where $\partial_{u}$ is the radial derivative explicitly calculated in~\eqref{derseconde}. We then use $\eqref{jacobian}$ and estimate
\begin{multline}
\left \b\int_{\R^{8}}W_{12}\;\d \left (\Ave\left [\Emp_{Z_{N}}\right ]^{\otimes 2}-\left (\1_{S_{L}}\mu\right )^{\otimes 2}\right )\right \b \\
\leq  \frac{C}{N^{2}} \sum_{\Omega_m \subset S_L } \sum_{\Omega_q\subset S_L} l_{x} N_m N_q \sup_{z_m \in \Omega_{m},z_q \in \Omega_{q}} \b p_m \b \left| W(x_m-x_q)\right|  \\
+ \frac{C}{N^{2}} \sum_{\Omega_m \subset S_L } \sum_{\Omega_q\subset S_L} N_m N_q l_p \sup_{z_m \in \Omega_{m},z_q \in \Omega_{q}} \left| \nabla^{\perp}w_{R} \left(x_m-x_{q}\right)\right|
\label{inew12}
\end{multline}
where we used that $\bA_{e}$ and its derivatives are in $L^{\infty}$ by assumption and use the convention $z_i=(x_i,p_i)$ systematically. We denote  
\begin{align*}
\overline{W_1} (x_1-x_2) &= \sup\left\{ |W(y_1-y_2)|, y_1 \in \Omega_{m_x} (x_1), y_2 \in \Omega_{m_x} (x_2)  \right\}\\
\overline{W_2} (x_1-x_2) &= \sup\left\{ |\nabla^{\perp} w_R (y_1-y_2)|, y_1 \in \Omega_{m_x} (x_1), y_2 \in \Omega_{m_x} (x_2)  \right\}\\
P (p_1) &= \sup\left\{ |p|, p\in \Omega_{m_p} (p)\right\}
\end{align*}
with $\Omega_{m_x} (x)\subset \R^2$ and $\Omega_{m_p} (p)\subset \R^2$ the square of our grid~\eqref{division} that $x$ belongs to (respectively $\Omega_{m_p} (p)$ the square $p$ belongs to). Hence
\begin{multline*}
\left \b\int_{\R^{8}}W_{12}\;\d \left (\Ave\left [\Emp_{Z_{N}}\right ]^{\otimes 2}-\left (\1_{S_{L}}\mu\right )^{\otimes 2}\right )\right \b \\ 
\leq C l_x \iint_{\R^8} P (p_1) \overline{W_1} (x_1-x_2) \Ave[\mu] (x_1,p_1) \Ave[\mu] (x_2,p_2) \d x_1 \d x_2 \d p_1 \d p_2   \\
 + C l_p \iint_{\R^8} \overline{W_2} (x_1-x_2) \Ave[\mu] (x_1,p_1) \Ave[\mu] (x_2,p_2) \d x_1 \d x_2 \d p_1 \d p_2 =: \mathrm{I} + \mathrm{II}.
\end{multline*}
Recognizing a Riemann sum similarly as in \eqref{renvoirienman} yields, using the properties of $W$ (see~\eqref{derseconde}), 
\begin{align*}
\mathrm{I} &\leq C l_x \int_{\R^4} |p_1|  \Ave[\mu] (x_1,p_1) \int_{\R^2} \left( \frac{1}{R} |\nabla w_0 (x_1-x_2)| + \frac{1}{R^2} \1_{|x_1-x_2|\leq R} \right) \rho_{\Ave[\mu]} (x_2) \d x_2 \d x_1 \d p_1\\
&\leq C \frac{l_x}{R} \norm{\rho_{\Ave[\mu]}}_{L^2} \int |p| \Ave[\mu] (x,p)\d x \d p
\\
&\leq C \frac{l_x}{R} \norm{\rho_{\Ave[\mu]}}_{L^2} \left( \int |p|^2 \Ave[\mu] (x,p)\d x \d p\right)^{1/2}
\end{align*}
where we also used the weak Young and the Cauchy-Schwarz inequality. Similarly, since $\nabla w_R \in L^{2,w}$, 
$$
\mathrm{II} \leq C l_p \iint_{\R^4} \left| \nabla w_R (x_1-x_2) \right| \rho_{\Ave[\mu] (x_1)} \rho_{ \Ave[\mu] (x_2)} \d x_1 \d x_2 \leq C l_p \norm{\rho_{\Ave [\mu]}}_{L^2}.
$$ 
We next turn to the three-body term~\eqref{W123}

$$
\mathrm{III}:= \left \b\int_{\R^{12}}W_{123} \d\left (\left (\1_{S_{L}}\mu\right ) ^{\otimes 3}-\Ave\left [\Emp_{Z_{N}}\right ]^{\otimes 3}\right )\right \b 
$$
with
$$
W_{123}=\nabla^{\perp }w_{R}\left (x_{1}-x_{2}\right )\nabla^{\perp }w_{R}\left (x_{1}-x_{3}\right).
$$
We denote (with $W$ as in~\eqref{eq:def W})
\begin{align*}
G (x_1,x_2,x_3) &= \left \b W( x_1- x_2) \right \b \left \b \nabla^{\perp }w_{R}(x_1 - x_3)\right \b + \left \b \nabla^{\perp }w_{R}(x_1-x_2) \right \b \left \b W(x_1 - x_3)\right \b\\
\overline{G} (x_1,x_2,x_3)&= \sup\left\{ G(y_1,y_2,y_3), y_1\in \Omega_{m_x} (x_1), y_2\in \Omega_{m_x} (x_2), y_3\in \Omega_{m_x} (x_3)\right\}
\end{align*}
and observe that $G$ gives a pointwise upper bound to gradients of $W_{123}$ in any variable. Hence, arguing as above 
\begin{align*}
\mathrm{III} &\leq C l_x \int_{R^6} \overline{G} (x_1,x_2,x_3) \rho_{\Ave[\mu]} (x_1) \rho_{\Ave[\mu]} (x_2) \rho_{\Ave[\mu]} (x_3) \d x_1 \d x_2 \d x_3\\
&\leq C \frac{l_x}{R} \norm{\rho_{\Ave[\mu]}}_{L^2} ^2 
\end{align*}
where we used the weak Young inequality twice. Collecting the previous inequalities we conclude 
\begin{multline*}
\E^R_{V}\left [\mu \right ]\geqslant (1-\sigma )\E^R_{V}\left[\Ave[\mu]\right]-C(l_x+l_p) \\
- C \frac{l_x}{R}\norm{\rho_{\Ave [\mu]}}_{L^2}\left( \norm{\rho_{\Ave [\mu]}}_{L^2} + \left( \int|p|^2 \d \Ave[\mu]\right)^{1/2}\right)  \\ -C l_p \norm{\rho_{\Ave [\mu]}}_{L^2} - C\frac{D^2}{\sigma R^{2} \inf(L^{4},L^{2s})}\\
\\ \geq (1-\sigma )\E^R_{V}\left[\Ave[\mu]\right]-o_N(1) -o_N (1)  \int|p|^2 \d \Ave[\mu]  - C\frac{D^2}{\sigma R^{2} \inf(L^{4},L^{2s})},
\end{multline*}
inserting our assumptions on the various parameters involved $l_{x}^{2}=N^{-\beta +\epsilon}\ll R^{2}$ and $l_{p}^{2}=N^{-\epsilon}$. We have also used the fact that, since 
\begin{equation}\label{eq:bound Ave} 
\Ave[\mu] \leq \frac{1+\epsilon}{2 \pi^2},
\end{equation}
the bathtub principle~\cite[Theorem~1.14]{LieLos-01} implies that
$$ \int|p|^2 \d \Ave[\mu] \geq C_\eps \norm{\rho_{\Ave[\mu]}}^2.$$
Finally, arguing as in Appendix~\ref{app:Vlasov} below (using~\eqref{eq:bound Ave} again),
$$
\E^R_{V}\left[\Ave[\mu]\right] \geq \frac{1}{2}\int|p|^2 \d \Ave[\mu]  - C  \norm{\rho_{\Ave[\mu]}} ^2
$$
and we can absorb an error term in the main term to obtain 
$$
\E^R_{V}\left [\mu \right ]\geq (1-\sigma )(1-o_N(1))\E^R_{V}\left[\Ave[\mu]\right]-o_N(1) -o_N (1)  \norm{\rho_{\Ave[\mu]}} ^2  - C\frac{D^2}{\sigma R^{2} \inf(L^{4},L^{2s})}.
$$
A last use of the bathtub principle gives 
$$
\E^R_{V}\left[\Ave[\mu]\right] \geq C \norm{\rho_{\Ave[\mu]}} ^2.
$$
Thus the proof is concluded by absorbing one last error term in the main expression.
\end{proof}


\subsection{Convergence of the energy}

Going back to our lower bound~\eqref{enr1}, assuming $R=N^{-\eta}, 0< \eta < 1/4$, we have
\begin{align}
\frac{E^{R}\left (N\right )}{N}&=\frac{\bral \Psi_{N},H_{N}^{R}\Psi_{N}\ketr}{N}\nn\\
&\geqslant\int_{\mathcal{P}(\R^{4})}\E^{R}_{V}[\mu]\d P_{m_N}^{\mathrm{DF}}(\mu)- o_N (1) \nn \\
&\geqslant \int_{\Gamma _{\epsilon}^c \cap \Xi_\tau}\E^{R}_{V}[\mu]\d P_{m_{N}}^{\mathrm{DF}}(\mu)-o_N (1)
\label{enr2}
\end{align}
where
\begin{equation}
\E^{R}_{V}[\mu]=\int_{\R^{4}}\left (\left |p^{\bA}+\beta\bA^{R}\left [\mu\right ](x)\right |^{2}+V(x)\right ) \d \mu,
\end{equation}
$\Gamma^\epsilon$ is as in~\eqref{eq:bad set},
$$ \Xi_\tau = \left\{ \mu \in \cP (\R^4)  \int_{\R^{4}}\left (\b p \b^{2}+V(x)\right )\d \mu (x,p)\leqslant \frac{\tau}{R^{2}}\right\},$$
and we have used the positivity of the integrand to go to the second line of~\eqref{enr2}. We next pick $l_x,l_p$ to enforce~
$$
l_{x} \ll R, \quad l_p \ll 1, \quad l_x ^2 l_p^2 = N^{-\beta} 
$$
for some $\beta <1$. This is certainly compatible with our other requirements, since $R\gg N^{-1/4}$ by assumption. We can then use Theorem~\ref{boundedmeasure} to deduce 
\begin{equation}\label{departcv} 
\frac{E^{R}\left (N\right )}{N} \geq (1-\sigma)(1-o_N(1)) \int_{\Gamma _{\epsilon}^c \cap \Xi_\tau}\E^{R}_{V}\left[\Ave[\mu]\right]\d P_{m_{N}}^{\mathrm{DF}}(\mu)-o_N (1) - C\frac{\tau^2}{\sigma R^{6} \inf(L^{4},L^{2s})}.
\end{equation}
Next, denote, for some $\gamma>0$ 
\begin{equation}
\Theta_{\gamma} =\left \{\mu\in \cP (\R^4) \mbox{ such that } \int_{S_{L}^{c}}\mu < \gamma\right \}
\end{equation}
the set of measures whose mass lie mostly in $S_{L}$ and note that, as per Lemmas~\ref{convR} and~\ref{infsansepsilon},
$$ \E^{R}_{V}\left[\Ave[\mu]\right] \geq (1-CR) (1-2\epsilon)(1-\gamma)e_{\rm TF}$$
for any $\mu \in \Gamma_\epsilon ^c \cap \Theta_{\gamma}.$ Clearly, if we set  
\begin{equation}\label{eq:gamma} 
\gamma = c \frac{\tau}{R^2 \inf (L^2,L^{2s})}
\end{equation}
with a suitable $c>0$ we have that 
$$ \Theta_{\gamma} \subset \Xi_\tau.$$
Hence, with this choice,
\begin{align}\label{eq:ener presque}
\frac{E^{R}\left (N\right )}{N} &\geq (1-CR) (1-2\epsilon)(1-\gamma)(1-\sigma)(1-o_N(1)) e_{\rm TF} \int_{\Gamma _{\epsilon}^c\cap \Xi_\tau} \d P_{m_N}^{\mathrm{DF}}(\mu)-o_N (1) \nn 
\\&- \frac{C\tau^2}{\sigma R^6\inf(L^2,L^{2s})}\nn
\\&\geq (1-2\epsilon)(1-\gamma)(1-\sigma)(1-o_N(1)) e_{\rm TF} \left( 1 - \mathbb{P}_{m_N}^{\mathrm{DF}} (\Gamma_\epsilon) - \mathbb{P}_{m_N}^{\mathrm{DF}} (\Xi_\tau ^c)  \right)\nn
\\&- \frac{C\tau^2}{\sigma R^6\inf(L^2,L^{2s})},
\end{align}
with a self-explanatory notation for $\mathbb{P}_{m_N}^{\mathrm{DF}}$. There remains to estimate the probability of the ``bad events'' in the last inequality. First, there is a $\delta >0$ such that
\begin{align}
\mathbb{P}_{m_{N}}^{\mathrm{DF}}\left (\Gamma_{\epsilon} \right )&=\mathbb{P}_{m_{N}}^{\mathrm{DF}}\left (\cup_{S_{K}\in S_{L}} \left \{ \Emp_{Z_{N}},\int_{\Omega_{k}}\Emp_{Z_{N}}\geq \frac{(1+\epsilon )}{\left (2\pi\right )^{2}}\b\Omega_{k}\b\right \}\right )\nn \\
&\leqslant \sum_{k=1}^{16n^{4}}\mathbb{P}_{m_{N}}^{\mathrm{DF}}\left (\left \{ \Emp_{Z_{N}},\int_{\Omega_{k}}\Emp_{Z_{N}}\geqslant \frac{(1+\epsilon )}{\left (2\pi\right )^{2}}\b\Omega_{k}\b\right \}\right )\nn\\
&\leqslant C_{\delta}L^{4}N^{\beta}e^{-c_{\delta}N^{\delta}\ln\left (1+\epsilon\right )}.
\label{errorpauli}
\end{align}
We used the union bound, the bound on the probability of violating the Pauli principle in a single box from Theorem~\ref{bigprob} and the relation~\eqref{L}. 

On the other hand, using Markov's inequality,~\eqref{BKE}, Lemma~\ref{densitiesconv} and~\eqref{m1} we find 
\begin{align}\label{eq:kin ener trop}
\mathbb{P}_{m_{N}}^{\mathrm{DF}} (\Xi_\tau) &\leq R^2 \tau ^{-1} \int_{\cP (\R^4)} \left( \int_{\R^4} \left( |p|^2 + V (x) \right) \d \mu (x,p) \right)\d P_{m_N}^{\rm{DF}} (\mu)\nn\\
&= C R^2 \tau ^{-1} \int_{\R^4} \left( |p|^2 + V (x) \right) \d m_N ^{(1)} (x,p) \nn\\
&\leq C \tau ^{-1}.
\end{align}
There remains to insert~\eqref{errorpauli} and~\eqref{eq:kin ener trop} in~\eqref{eq:ener presque}. We can now set $L=N^\kappa$ with $\kappa >0$ a suitably large, fixed number and $\tau = N^{\kappa'}$ with $\kappa'$ a suitably small, fixed number. This makes all errors depending on $N$ negligible because, crucially,~\eqref{errorpauli} contains an exponentially small term in $N$. In particular, recalling~\eqref{eq:gamma}, $\gamma \to 0$ with such a choice of a parameters. Hence we can pass to the limit in~\eqref{eq:ener presque}, first letting $N,R^{-1},L \to \infty$ with the previously specified scalings, and finally let $\epsilon,\sigma \to 0$, obtaining the energy lower bound claimed in~\ref{th1}, matching the upper bound from Section~\ref{sec:upper}. 

\subsection{Convergence of states}
We turn to the proof of Theorem~\ref{th2}. Coming back to~\eqref{departcv} and inserting the previous choice of parameters we have that 
\begin{equation}\label{eq:states 1}
e^{\mathrm{TF}}\geq (1-\sigma) (1-o_{N}(1)) \int_{\mathcal{P}\left (\R^{4}\right )}\1_{\Xi_{\tau}} (\mu) \1_{\Gamma_{\epsilon } ^c} (\mu) \E_{V}\left [\Ave\left[\mu\right]\right ] \d P_{m_N}^{\mathrm{DF}}(\mu )-o_{N}(1) (1+\sigma^{-1}).
\end{equation}
It follows from~\eqref{eq:states 1} and results of Appendix~\ref{app:Vlasov} that $m_{N}^{(1)}$ is tight. Applying the last part of Theorem~\ref{DF} we thus obtain that the Diaconis-Freedman measure $P_{m_N}^{\rm DF}$ weakly converges to a probability measure $P\in \cP (\cP (\R^4))$ (this is the Hewitt-Savage measure). It follows from Theorem~\ref{bigprob} (or~\cite[Theorem~2.6]{FouLewSol-15}) that $P$ is concentrated on probability measures satisfying $\mu \leq (2\pi) ^{-2}$. Combining with Lemma~\ref{propH} and~\eqref{mn} we have that, for any fixed $k\geq 0$ 
$$
\frac{1}{(2\pi) ^{2k}} m_{\Psi_N} ^{(k)}\underset{N\to \infty}{\to} \int_{\cP (\R^4)} \mu^{\otimes k} \d P(\mu)
$$
weakly as measures. We now identify the limit measure $P$ to conclude the proof of Theorem~\ref{th2}.

To this aim it is convenient to write the integral above as 
\begin{equation}\label{eq:apply lower 1}
\int_{\mathcal{P}\left (\R^{4}\right )}\1_{\Xi_{\tau}} (\mu) \1_{\Gamma_{\epsilon } ^c} (\mu) \E_{V}\left [\Ave\left[\mu\right]\right ] \d P_{N}^{\mathrm{DF}}(\mu ) = \int_{\mathcal{P}\left (\R^{4}\right )}\E_{V} \left [\Ave\left[\mu\right]\right ] \d Q_{N} (\mu )
\end{equation}
with 
\begin{equation}\label{eq:Q_N}
\d Q_N [\mu] :=  \1_{\Xi_{\tau}} (\mu) \1_{\Gamma_{\epsilon } ^c} (\mu) \d P_{N}^{\mathrm{DF}}(\mu ).
\end{equation}
We will see $\E_{V}\left [\Ave\left[\mu\right]\right ]$ as a lower semi-continuous function on the space $\left (\mathcal{P}\left (\R^{4}\right ),d_{1}^{W} \right )$ where
\begin{equation*}
d_{1}^{W}\left (\mu ,\nu\right )=\sup_{\norm{\phi}_{\rm Lip}\leqslant 1}\left \b\int_{\R^{4}}\phi\d \mu -\int_{\R^{4}}\phi\d \nu\right \b
\end{equation*}
is the Monge-Kantorovitch-Wasserstein distance, known to metrize the weak convergence of measures~\cite[Theorem~6.9]{Villani-08} (and we use Kantorovitch-Rubinstein duality~\cite[Remark~6.5]{Villani-08}). We have

\begin{lemma}[\textbf{Limit of $\Ave$}]\mbox{}\\
\label{contav}
Let $\mu\in \cP (\R^4).$ In the $d_{1}^{W}$ metric, 
\begin{equation}\label{eq:ave conv}
\Ave\left [\mu \right ]=\sum_{m}^{16n^{4}}\1_{\Omega_{m}}\int_{\Omega_{m}}\frac{\d \mu }{\b\Omega_{m}\b} \to \mu,
\end{equation}
when $N\to \infty$ and $L\to \infty$.
\end{lemma}

\begin{proof}
We pick a function $\phi$, $\norm{\nabla \phi}_{L^{\infty}}\leqslant 1
$ and denote 
\begin{equation*}
\mathrm{a}\left [\phi\right ]_{m}=\int_{\Omega_{m}}\frac{\phi(x,p)}{\b\Omega_{m}\b}\d x\d p
\end{equation*}
its average over a box $\Omega_{m}$. We have that, for any $x\in \Omega_m$,
\begin{equation}
\left \b \phi(x) -\mathrm{a}\left [\phi\right ]_m\right \b\leqslant Cl\norm{\nabla \phi}_{L^{\infty}}
\label{lip}
\end{equation}
where $l= \max (l_x,l_p)$ is the largest dimension of the box $\Omega_{m}$.
We thus have
\begin{equation*}
\int_{\R^{4}}\phi(x)\d \Ave\left [\mu\right ]=\sum_{m}^{16n^{4}}\int_{\Omega_{m}}\mathrm{a}\left [\phi\right ]_{m}\d \mu =\sum_{m}\int_{\Omega_{m}}\phi\d \mu +o_{N}(1)
\end{equation*}
where the $o_N (1)$ is uniform in $\norm{\nabla \phi}_{L^{\infty}}$ by~\eqref{lip}. 
\end{proof}

It follows that, if we take a sequence $\mu_n \to\mu $ when $N\to \infty$,
\begin{equation*}
d_{1}^{W}\left (\Ave\left [\mu_n\right ],\mu\right )\leqslant d_{1}^{W}\left (\Ave\left [\mu_n\right ],\mu '\right )+d_{1}^{W}\left (\mu_n,\mu\right )\leqslant o_{N}(1).
\end{equation*}
Hence, combining with Lemma~\eqref{lsc}, we deduce
\begin{equation}\label{eq:lower semi}
\liminf_{N\to\infty,\;\mu '\to \mu }\E_{V} \left [\Ave\left[\mu'\right]\right ] \geqslant \E_{V}\left [\mu \right ].
\end{equation}
We turn to the convergence of $Q_N$ defined in~\eqref{eq:Q_N}. For any continous bounded function $\Phi\in C_b (\cP (\R^4))$ over probability measures,
$$
\left|\int_{\mathcal{P}\left (\R^{4}\right )}\Phi(\mu ) \left(\d Q_{N} \left (\mu\right ) - \d P_{m_N} ^{\rm DF}\left (\mu\right )\right)\right|\underset{N \to \infty}{\to} 0 
$$
by~\eqref{errorpauli} and~\eqref{eq:kin ener trop}. Hence, applying Theorem~\ref{DF} we have that 
$$ Q_N \to P$$
as measures over $\cP (\R^4)$, where $P$ is the limit of the Diaconis-Freedman measure. Combining with~\eqref{eq:lower semi} we may apply the improved Fatou Lemma of~\cite[Theorem 1.1]{FeiKasZad-14} to~\eqref{eq:apply lower 1}, obtaining 
$$
\liminf_{N\to \infty} \int_{\mathcal{P}\left (\R^{4}\right )}\E_{V} \left [\Ave\left[\mu\right]\right ] \d Q_{N} (\mu ) \geq \int_{\mathcal{P}\left (\R^{4}\right )}\E_{V} \left [\mu\right] \d P (\mu ). 
$$
Combining with~\eqref{eq:states 1} and letting $\sigma \to 0$ after $N\to \infty$ we conclude
$$
e^{\mathrm{TF}} \geq \int_{\mathcal{P}\left (\R^{4}\right )}\E_{V} \left [\mu\right] \d P (\mu ).
$$
Recall that $P$ is concentrated on probability measures satisfying $\mu \leq (2\pi) ^{-2}$ and that $e^{\mathrm{TF}}$ is the infimum of $\E_{V} \left [\mu\right]$ over those. Hence $P$ must be concentrated on the unique minimizer of $\E_{V} \left [\mu\right]$, and this concludes the proof.
%
%
%
%
%
%
%
%
%
%
%
%
%
%
%
%
%
%
%
%
%
%
%
%

\appendix 

\section{Properties of the Vlasov functional}\label{app:Vlasov}

In this appendix we etablish some of the fundamental properties of the functional $\E_{V}$ and some useful bounds on the vector potential $\bA^{R}[\rho]$ associated to a measure $\rho$.



\begin{lemma}[\textbf{Lower semicontinuity of $\E_{V}$}]\mbox{}\\
\label{lsc}
Let $(\mu_n)_{n\geq 0}$ be a sequence of positive measures on $\R^4$ satisfying 
$$0\leq \mu_{n}\leqslant \left (2\pi\right )^{-2}, \quad \int_{\R^4} \mu_n \leq C$$
with $C$ independent of $N$. If $\mu_n$ converges to $\mu$ as measures we have
\begin{equation}
\liminf_{n\to\infty}\E_{V}\left [\mu_{n}\right ]\geqslant \E_{V}\left [\mu\right ].
\end{equation}

\end{lemma}

\begin{proof}
We recall that the marginal in $p$ of $\mu_{n}$ is
\begin{equation}
\rho_{n} (x)=\int_{\R^{2}}\mu_{n} (x,p)\d p.
\end{equation}
It follows from applying the bathtub principle~\cite[Theorem 1.14]{LieLos-01} in the $p$ variable that
\begin{equation}
\E_{V}\left [\mu_n \right ]\geqslant \E_{V}^{\mathrm{TF}}\left [\rho_n \right ]:=2\pi \int_{\R^{2}}\rho_n ^{2}(x)\d x+\int_{\R^{2}}V(x)\rho_n(x)\d x.\nn
\end{equation}
Hence we may assume that $\norm{\rho_n}_{L^{2}}$ is uniformly bounded and that $\rho_{n} \rightharpoonup \rho$ weakly in $L^{2}\left (\R^{2}\right)$. We then deduce from the weak Young inequality~\cite[Chapter~4]{LieLos-01} that
\begin{equation}
\norm{ \bA^{2}[\rho_{n}]\rho_{n}}_{L^{1}}\leqslant C \norm{ \nabla^{\perp}w*\rho_{n}}^{2}_{L^{\infty}}\leqslant C \norm{\rho_{n}}^{2}_{L^{2}}\norm{\nabla^{\perp}w}^{2}_{L^{2,w}}\leqslant C.\nn
\end{equation}
We recall that the weak $L^{p}(\R^{d})$ space is the set of functions
 \begin{equation}
 L^{w,2}=\left \{f\;\b\; \mathrm{Leb}\left (\left \{x\in \R^{d}:\b f(x)\b >\lambda\right \}\right )\leqslant\left (\frac{C}{\lambda}\right )^{p} \right \}\nn
 \end{equation}
with the associated norm
 \begin{equation}
 \norm{f}_{L^{w,2}}=\sup_{\lambda >0}\lambda \;\mathrm{Leb}\left ( x\in \R^{d},\;\left \b f(x)\right \b >\lambda \right )^{1/2}\nn
\end{equation} 
where $\mathrm{Leb}$ is the Lebesgue measure. We expand the energy
\begin{equation}
\E_{V}\left [\mu_{n} \right ]=\int_{\R^{4}}\left (\left (p^{\bA}\right )^{2}+2\beta p^{\bA}\cdot \bA[\rho_{n}]+\beta^{2}\bA^{2}[\rho_{n}]+V\right )\mu_{n}(x,p)\d x\d p
\label{gh}
\end{equation}
and use that $ab\leqslant \frac{a^{2}}{2\sigma }+\frac{\sigma b^{2}}{2}\nn$
to get
\begin{equation}
2\left (p^{\bA}\right )\cdot \bA\left [\rho_{n}\right ]\geqslant-\sigma\b p^{\bA}\b^{2}-\frac{1}{\sigma }\left \b \bA\left [\rho_{n}\right ]\right \b^{2}\nn
\end{equation}
and hence
\begin{equation}
\E_{V}\left [\mu_{n}\right ]+C\norm{\bA^{2}\left [\rho_{n}\right ]\rho_{n}}_{L^{1}}\geqslant C\int_{\R^{4}}\left (\b p^{\bA} \b^{2}+V\right )\d \mu_{n}.\nn
\end{equation}
Thus $(\mu_{n})$ is tight and, up to extraction, converges strongly in $L^{1}\left (\R^{4}\right )$. In order to obtain the convergence of $\bA^{2}[\rho_{n}]\mu_{n}$ we write
\begin{align}
\rho_{n}(x_{1})\rho_{n}(x_{2})\mu_{n}(x_{3},p_{3})- \rho(x_{1})\rho(x_{2})\mu(x_{3},p_{3})&=\left ( \rho_{n}(x_{1})- \rho(x_{1})\right )\left ( \rho_{n}(x_{2})- \rho(x_{2})\right )\mu_{n}(x_{3},p_{3})\nn\\
&+\left ( \rho_{n}(x_{1})- \rho(x_{1})\right ) \rho(x_{2})\mu_{n}(x_{3},p_{3})\nn\\
&+ \rho(x_{1})\left ( \rho_{n}(x_{2})- \rho(x_{2})\right )\mu_{n}(x_{3},p_{3})\nn\\
&-\rho(x_{1})\rho(x_{2})\left (\mu_{n}(x_{3},p_{3})-\mu(x_{3},p_{3})\right )
\label{ghgh}
 \end{align}
and treat each resulting term separately. The first term yields
 \begin{align}
 \norm{ \bA^{2}[\rho_{n}-\rho]\mu_{n}}_{L^{1}} \leqslant \norm{\mu_{n}}_{L^1}\norm{ \nabla^{\perp}w*\left (\rho_{n}-\rho\right )}^{2}_{L^{\infty}}\leqslant \norm{\rho_{n}-\rho}^{2}_{L^{2}}\norm{\nabla^{\perp}w}^{2}_{L^{2,w}}\nn
 \end{align}
 using that $\b\b\mu_{n}\b\b_{L^1}\leq C$ and the weak Young inequality.  For the second term of $\eqref{ghgh}$
  \begin{align}
\norm{ \bA[\rho_{n}-\rho]\bA[\rho]\mu_{n}}_{L^{1}}\leqslant \norm{\bA[\rho]}_{L^{\infty}}\norm{ \nabla^{\perp}w*(\rho_{n}-\rho )}_{L^{\infty}}\leqslant \norm{\rho }_{L^{2}}\norm{\rho_{n}-\rho}_{L^{2}}\norm{\nabla^{\perp}w}^{2}_{L^{2,w}}\nn
 \end{align}
 and for the last one 
 \begin{align}
 \norm{ \bA^{2}[\rho](\mu_{n}-\mu)}_{L^{1}}\leqslant\norm{\mu_{n}-\mu}_{L^{1}}\norm{\bA^{2}[\rho]}_{L^{\infty}}\leqslant\norm{\mu_{n}-\mu}_{L^1}\norm{\nabla^{\perp}w}^{2}_{L^{2,w}}\norm{\rho}_{L^2}\nn
 \end{align}
 For the cross term of $\eqref{gh}$, $p^{\bA}.\bA\left [\rho_{n}\right ]\mu_{n}$ we observe that
 \begin{equation}
 \norm{\bA[\rho -\rho_{n}]}_{L^2}\leqslant \norm{ \rho -\rho_{n}}_{L^1}\norm{\nabla^{\perp}w}_{L^{2,w}}\nn
\end{equation} so $\bA[\rho_{n}]$ converges strongly in $L^{2}$.  We have $\b\b \left (p^{\bA}\right )^{2}\mu_{n}\b\b_{L^{1}}\leqslant C$  so $p^{\bA}\mu_{n}$ converges weakly in $L^{2}$ and by weak-strong convergence we deduce that the cross term converges. We conclude using Fatou's lemma for $V$ and the kinetic term
\begin{align*}
\liminf_{n\to \infty}\int_{\R^{4}}\b p+\bA_{e}(x)\b^{2}\mu_{n}(x,p)\d x\d p &\geqslant \int_{\R^{4}}\b p+\bA_{e}(x)\b^{2}\mu (x,p)\d x\d p\\
\liminf_{n\to \infty}\int_{\R^{2}}V(x)\rho_{n}(x)\d x &\geqslant \int_{\R^{2}}V(x)\rho (x)\d x .
\end{align*} 
 \end{proof}


 \begin{lemma}[\textbf{Existence of minimizers}]\mbox{}\\
 There exists a minimizer for the problem
 \begin{equation}
 e^{\mathrm{TF}}=\inf \left \{\E_{V}[\mu ]\;\b\; 0\leqslant \mu \leqslant \left (2\pi\right )^{-2}\;\b\; \int_{\R^{4}}\mu = 1 \right \}.\nn
 \end{equation}
 
 \end{lemma}
 \begin{proof}
 We consider a minimizing sequence $\left(\mu_{n}\right)_{n}$ converging to a candidate minimizer $\mu_{\infty}$. By Lemma $\eqref{lsc}$ we have
\begin{equation}
e^{\rm TF} = \inf_{n\to\infty}\E_{V}\left [\mu_{n}\right ]\geqslant \E_{V}\left [\mu_{\infty}\right ].\nn
\end{equation}
We also found during the previous proof that $\left(\mu_{n}\right)_{n}$ must be tight, hence 
$$ \int_{\R^{4}}\mu_\infty = 1.$$
 \end{proof}


\begin{lemma}[\textbf{Convergence to $\E_{V}$ when $R\to 0$}]\mbox{}\\
\label{convR}
For any measure $\mu \leqslant \left (2\pi\right )^{-2}$ such that $\int_{\R^{4}}\mu \leqslant C$ we have that
\begin{equation}
\left \b \E^{R}_{V}[\mu]-\E^{0}_{V}[\mu]\right \b\leqslant CR \left (\E^{0}_{V}[\mu]+\E^{R}_{V}[\mu]\right )\nn
\end{equation}
where 
\begin{equation}
\E^{R}_{V}[\mu]=\int_{\R^{4}}\left (p^{\bA}+\beta\bA^{R}\left [\mu\right ](x)\right )^{2}+V(x) \d \mu .\nn
\end{equation}
\end{lemma}

\begin{proof}
\begin{align*}
\left \b \E^{R}_{V}[\mu]-\E^{0}_{V}[\mu]\right \b &= \left \b\norm{\left (p^{\bA}+\beta\bA^{R}\left [\mu\right ](x)\right )\sqrt{\mu }}_{L^2}^{2}-\norm{\left (p^{\bA}+\beta\bA^{0}\left [\mu\right ](x)\right )\sqrt{\mu }}_{L^2}^{2}\right \b\\
&\leqslant \left \b\norm{\left (p^{\bA}+\beta\bA^{R}\left [\mu\right ](x)\right )\sqrt{\mu }}_{L^2}-\norm{\left (p^{\bA}+\beta\bA^{0}\left [\mu\right ](x)\right )\sqrt{\mu }}_{L^2}\right \b\\
&\;\;\;\;\;\;\;\;\;\cdot\left (\E_{V}^{R}\left [\mu \right ]^{1/2}+\E_{V}\left [\mu \right ]^{1/2}\right )\\
&\leqslant \norm{\left (\bA^{R}\left [\mu \right ]-\bA\left [\mu \right ]\right )\sqrt{\mu }}_{L^2}
\left (\E_{V}^{R}\left [\mu \right ]^{1/2}+\E_{V}\left [\mu \right ]^{1/2}\right )
\end{align*}
where we have used  the triangle inequality. Moreover we have that
\begin{align*}
\norm{\left (\bA^{R}\left [\mu \right ]-\bA\left [\mu \right ]\right )^{2}\rho }^{1/2}_{L^1}&\leqslant \norm{\rho }_{L^1}^{1/2}\norm{\left (\nabla^{\perp }w_{R}-\nabla^{\perp }w_{0}\right )*\rho }_{L^\infty}\\
&\leqslant \norm{\rho }_{L^2}\norm{\nabla^{\perp }w_{R}-\nabla^{\perp }w_{0}}_{L^{2,w}}
\end{align*}
by the weak Young inequality and because
\begin{equation*}
\int_{\R^{2}}\left (\int_{\R^{2}}\frac{1}{\b x\b }\chi (u) \d u-\int_{B(0,R)}\frac{1}{\b x\b}\chi (u)\d u\right )\d x\leqslant \int_{B(0,2R)}\frac{1}{\b x\b }\d x =CR.
\end{equation*}
Using that the minimizer in $p$
\begin{equation*}
\mu^{\mathrm{TF}}=\left (2\pi\right )^{-2}1\!\!1\left (\b p^{\bA}+\bA[\rho ]\b\leqslant \sqrt{4\pi \rho}\right )
\end{equation*} 
is explicit we have
$$
\E_{V}[\mu]\geqslant\E_{\mathrm{TF}}(\rho )=2\pi\int_{\R^{2}}\rho^{2}(x)\d x+\int_{\R^{2}}V(x)\rho(x)\d x
$$
and the minimization problem is now formulated in terms of
\begin{equation*}
\rho (x)=\int_{\R^{2}}\mu (x,p)\d p.
\end{equation*}
This gives
\begin{equation}
\E_{V}[\mu]\geqslant C\norm{\rho}_{L^2}^{2}\nn
\end{equation}
and concludes the proof.
\end{proof}
 \begin{lemma}[\textbf{Dependence on the upper perturbed constraint}]\mbox{}\\
The infimum of $\E_{V}$ does not depend on $\epsilon $ nor $\gamma $ at first order
\label{infsansepsilon}
\begin{equation}
 \inf \left \{\E_{V}[\mu ]\;\b\;0\leqslant \mu \leqslant \frac{1+\epsilon}{\left (2\pi\right )^{2}},\;\int \mu = 1-\gamma \right \}\geqslant \left (1-2\epsilon \right )\left (1-\gamma\right )e_{\mathrm{TF}}.\nn
\end{equation}

\end{lemma}
\begin{proof}
We calculate the infimum with the Bathtub principle \cite[Theorem 1.14]{LieLos-01}.
This infimum is achieved for
\begin{equation}
\mu =\mu_{\epsilon}^{\mathrm{TF}}=\left (2\pi\right )^{-2}(1+\epsilon )1\!\!1\left (\b p^{\bA}+\beta\bA[\rho ]\b^{2}\leqslant s(x)\right )\nn
\end{equation}
where $s(x)=\frac{4\pi\rho}{1+\epsilon}$ because
\begin{equation}
\int_{\R^{2}}\mu^{\mathrm{TF}} (x,p)\d p=\rho (x)=\int_{\R^{2}}(1+\epsilon )1\!\!1\left (\b p^{\bA}+\beta\bA[\rho ]\b\leqslant s(x)\right )\d p = \pi (1+\epsilon )s(x).\nn
\end{equation}
So, evaluating the energy
\begin{align}
\E_{V}[\mu_{\epsilon}^{\mathrm{TF}} ]=\E^{\epsilon}_{\mathrm{TF}}[\rho ]=2\pi\int_{\R^{2}}\frac{\rho^{2}\left (x\right )}{\left (1+\epsilon \right )^{2}}\d x+\int_{\R^{2}}\frac{V(x)\rho(x)}{\left( 1+\epsilon \right )}\d x\nn\\
\end{align}
but we can see that
\begin{equation}
\E^{\epsilon}_{\mathrm{TF}}[\rho ]\geqslant \E_{\mathrm{TF}}[\rho ]\left (1-2\epsilon \right )
\label{lasten}
\end{equation}
where $ \E_{\mathrm{TF}}[\rho ]$ is given in $\eqref{ETF}$.
\end{proof}

\section{Bounds for $\bA^{R}$}
\begin{lemma}[\textbf{Bounds linked to $\bA^{R}$}]\mbox{}\\
All second-order directional derivatives of the function
\begin{equation}
w_{R}(u_{1},u_{2}): \R^{2}\to\R^{2}\nn
\end{equation}
are bounded in absolute value by the radial derivative:
\begin{equation}
\left \b \partial_{u_{i}}\partial_{u_{j}} w_{R}(\vec{u})\right \b\leqslant C\left \b\partial_{u}^{2}w_{R}(u)\right \b
\label{jacobian}
\end{equation}
for any $(i,j)\in \left \{1,2\right \}^{2}$. We also have the estimates
\begin{equation}
\norm{\Delta \nabla^{\perp }w_{R}}_{L^{\infty}}\leqslant \frac{C}{R^{3}},
\label{normsuplap}
\end{equation}
\begin{equation}
\norm{\nabla^{\perp }w_{R}}_{L^\infty}\leqslant \frac{C}{R},
\label{borne_W1}
\end{equation}
\begin{equation}
\norm{\Delta \nabla^{\perp }w_{R}}_{2}\leqslant \frac{C}{R^{2}},
\label{borne_lap}
\end{equation}
\begin{equation}
\norm{\nabla \left (\nabla^{\perp}w_{R}(u)\right )^{i}}_{L^\infty}\leqslant \frac{C}{R^{2}},
\label{borne_W2}
\end{equation}
\begin{equation}
\norm{\nabla \left (\nabla^{\perp}w_{R}(u).\nabla^{\perp}w_{R}(v)\right )}_{L^\infty}\leqslant \frac{C}{R^{3}}.
\label{borne_W3}
\end{equation}

\end{lemma}
\begin{proof}
Recall that 
\begin{align*}
 w_{R}(x)=\left(\log\b \;.\;\b *\chi_{R}\right)(x)
 \end{align*}
 with $\chi_{R}$ defined as in $\eqref{khiR}$.
 We call $u_{1}$ and $u_{2}$ the two components of the vector $\vec{u}$ and $u$ its norm. Using Newton's theorem \cite[Theorem 9.7]{LieLos-01} we write
\begin{align*}
w_{R}(u)=\int_{\R^{2}}\log\b \vec{u}-\vec{v}\b \chi_{R}(v)\mathrm{d}v&=2\pi\log(u)\int_{0}^{u}\chi_{R}(r)r\d r + 2\pi\int_{u}^{+\infty}r\log(r)\chi_{R}(r)\d r
\end{align*}
hence
\begin{align}
&\partial_{u}w_{R}(u)=\frac{1}{u}\int_{B(0,u)}\chi_{R}(v) \d v\nn\\
&\partial^{2}_{u}w_{R}(u)=-\frac{1}{u^{2}}\int_{B(0,u)}\chi_{R}(v) \d v +2\pi\chi_{R}(u)\label{ddeux}\\
&\partial^{3}_{u}w_{R}(u)=\frac{1}{u^{3}}\int_{B(0,u)}\chi_{R}(v) \d v +C\frac{\chi_{R}(u)}{u}+C\phi_{R}(u)
\label{derseconde}
\end{align}
where $\phi_{R}(u)=\partial_{u}\chi_{R}(u)$ is bounded with compact support.
We observe that regardless of whether $u$ is smaller or greater than $2R$ we have
\begin{align}
\left \b \partial_{u}w_{R}(u)\right \b\leqslant \frac{C}{R}\;\;\text{and}\;\;
\left \b \partial^{2}_{u}w_{R}(u)\right \b\leqslant \frac{C}{R^{2}}
\label{ine_der_wr}
\end{align}
and get $\eqref{borne_W1}$. Moreover
\begin{equation}
\nabla^{\perp }w_{R}(u_{1},u_{2})=\nabla^{\perp }w_{R}(u)=\frac{\partial_{u}w_{R}(u)}{u}\vec{u}^{\perp }.
\label{pot_vec}
\end{equation}
We compute the two derivatives of the second component of $\eqref{pot_vec}$
\begin{align}
\left \b\partial_{u_{1}}(\frac{u_{1}}{u}\partial_{u}w_{R})\right \b &=\left \b \frac{\partial_{u}w_{R}}{u}+\frac{u_{1}^{2}}{u^{2}}\partial^{2}_{u}w_{R}-\frac{u_{1}^{2}}{u^{3}}\partial_{u}w_{R}\right \b\leqslant \frac{C}{u^{2}}\int_{B(0,u)}\chi_{R}(v) \d v +C\chi_{R}(u)\\
\left \b\partial_{u_{2}}(\frac{u_{1}}{u}\partial_{u}w_{R})\right \b &=\left \b -\frac{u_{1}u_{2}}{u^{3}}\partial_{u}w_{R}(u)+\frac{u_{1}u_{2}}{u^{2}}\partial^{2}_{u}w_{R}(u)\right \b\leqslant \frac{C}{u^{2}}\int_{B(0,u)}\chi_{R}(v) \d v +C\chi_{R}(u).
\label{calcull}
\end{align}
We do the same with the first component of $\nabla^{\perp }w_{R}$ and get $\eqref{jacobian}$. If we differentiate once again we get
\begin{align}
\left \b\partial^{2}_{u_{1}}(\frac{u_{1}}{u}\partial_{u}w_{R})\right \b &\b\leqslant \frac{C}{u^{3}}\int_{B(0,u)}\chi_{R}(v) \d v +C\frac{\chi_{R}(u)}{u}+C\phi_{R}(u)\nn
\end{align}
which is also the case for the other component and derivative. We deduce
\begin{equation}
\norm{\Delta \nabla^{\perp }w_{R}}_{L^{\infty}}\leqslant \frac{C}{R^{3}}.\nn
\end{equation}
We can also compute
\begin{align}
\norm{\Delta \nabla^{\perp }w_{R}}_{L^2}^{2}\leqslant \frac{C}{R^{4}}\nn
\end{align}
which gives $\eqref{borne_lap}$. To get $\eqref{borne_W2}$ we combine $\eqref{ine_der_wr}$ and $\eqref{calcull}$.
For the third inequality $\eqref{borne_W3}$ we expand a little bit our first expression
\begin{align*}
\nabla^{\perp}w_{R}(u).\nabla^{\perp}w_{R}(v)&=(u_{1}v_{1}+u_{2}v_{2})\frac{\partial_{u}w_{R}(u)\partial_{v}w_{R}(v)}{uv}
\end{align*}
so the first component of the gradient is
\begin{align*}
\partial_{u_{1}}\nabla^{\perp}w_{R}(u).\nabla^{\perp}w_{R}(v)&=v_{1}\frac{\partial_{u}w_{R}(u)\partial_{v}w_{R}(v)}{uv}+\frac{u_{1}}{u}\partial_{u}\left [\frac{\partial_{u}w_{R}(u)\partial_{v}w_{R}(v)}{uv}\right ](u_{1}v_{1}+u_{2}v_{2}).
\end{align*}
Now
\begin{align*}
\norm{\nabla \left (\nabla^{\perp}w_{R}(u).\nabla^{\perp}w_{R}(v)\right )}^{2}_{L^\infty}&=\sum_{x=u_{1},u_{2},v_{1},v_{2}}\left (\partial_{x}\nabla^{\perp}w_{R}(u).\nabla^{\perp}w_{R}(v)\right )^{2}
\end{align*}
The rest of the proof consists of the computation of the above term using basic inequalities.
\end{proof}

\section{Computations for squeezed coherent states}\label{app:misc}

We give for completeness three proofs that we skipped in the main text.

\begin{proof}[\textbf{Proof of Lemma ~\ref{Resolution_identity}}]\mbox{}\\
\label{proof_id}
For any $u\in L^{2}(\R^{2})$,
\begin{align}
\bral F_{x,p},u \ketr&=\int_{\R^{2}}F_{\hbar_{x}}(y-x)u(y)e^{-\im\frac{p.y}{\hbar}}\d y\nn\\
&=2\pi\hbar\F_{\hbar}\left [F_{\hbar_{x}}(\cdot -x)u(\cdot)\right ](p)=2\pi\hbar\left (\F_{\hbar}\left [F_{\hbar_{x}}(\cdot -x)\right ]*\F_{\hbar}\left [u(\cdot )\right ]\right )(p)\\
&=\int_{\R^{2}}G_{\hbar_{p}}(k-p)\F_{\hbar} [u](k)e^{-\im\frac{k.x}{\hbar}}\d k=2\pi\hbar\F_{\hbar }\left [G_{\hbar_{p}}(\cdot -p)\F_{\hbar } [u]\right ](x).
\label{tool}
\end{align}
It follows that
\begin{align*}
\frac{1}{(2\pi\hbar)^{2}}\int_{\R^{2}}\int_{\R^{2}}\bral \psi, P_{x,p}\psi\ketr \;\d x\d p&=\frac{1}{(2\pi\hbar)^{2}}\int_{\R^{2}}\int_{\R^{2}}\b \bral F_{x,p},u \ketr\b^{2}\d x\d p\\
&=\int_{\R^{2}}\left (\int_{\R^{2}}\left \b\F_{\hbar }\left [F_{\hbar_{x}}(\cdot -x)u(\cdot)\right ](p)\right \b^{2}\d p\right )\d x\\
&=\int_{\R^{2}}\left (\int_{\R^{2}}\left \b F_{\hbar_{x}}(y -x)u(y)\right \b^{2}\d y\right )\d x\\
&=\b\b F_{\hbar_{x}}\b\b^{2}_{L^{2}(\R^{2})}\b\b u\b\b^{2}_{L^{2}(\R^{2})}=\b\b u\b\b^{2}_{L^{2}(\R^{2})}.
\end{align*}

\end{proof}


\begin{proof}[\textbf{Fourier transform of $F$}]\mbox{}\\
\label{ftcal}
We need to calculate the Fourier Transform of the Gaussian
\begin{equation}
G_{\hbar_{p}}(p)=\frac{1}{2\pi\hbar}\frac{\sqrt{\hbar_{x}}^{2}}{\sqrt{\hbar_{x}}}\frac{1}{\sqrt{\pi}}\left (\int_{\R}e^{-\frac{u^{2}}{2}-\frac{\im pu}{\sqrt{\hbar_{p}}}}\d u\right )^{2}=\frac{1}{2\pi\sqrt{\hbar_{p}}}\frac{1}{\sqrt{\pi}}G_{1}^{2}(p)\nn
\end{equation}
with
\begin{equation}
G_{1}(p)=\int_{\R}e^{-\frac{u^{2}}{2}-\frac{\im pu}{\sqrt{\hbar_{p}}}}\d u\nn
\end{equation}
but we also have
\begin{equation}
\frac{\d G_{1}(p)}{\d p}=-\frac{p}{\hbar_{p}}G_{1}(p)\nn
\end{equation}
which gives the result.
\end{proof}

\begin{proof}[Proof of Lemma~\ref{densitiesconv}]
\label{densf}
We use $\eqref{tool}$ and write for every fixed $y\in \R^{2(N-k)}$ 
\begin{align*}
&\Big < F_{x_{1},p_{1}}(\cdot )\otimes...\otimes F_{x_{k},p_{k}}(\cdot ),\Psi_{N}(\cdot ,z)
\Big >_{L^{2}(\R^{2k})}\\
=&(2\pi\hbar)^{2k}\F_{\hbar}\left [F_{x_{1},0}(\cdot )\otimes...\otimes F_{x_{k},0}(\cdot )\Psi_{N}(\cdot , z)\right ](p_{1},...,p_{k})\\
=&(2\pi\hbar)^{2k}\F_{\hbar}\left [G_{0,p_{1}}(\cdot )\otimes...\otimes G_{0,p_{k}}(\cdot )\F_{\hbar} \left [\Psi_{N}\right ](\cdot , z)\right ](x_{1},...,x_{k}).
\end{align*}
Next we sum over the $p_{j}$'s using $\eqref{Hus}$:
\begin{align*}
&\frac{1}{(2\pi)^{2k}}\int_{\R^{2k}}m_{\Psi_{N}}^{(k)}(x_{1},p_{1},...,x_{k},p_{k})\d p_{1}...\d p_{k}\\
=&\frac{k!}{(2\pi)^{2k}}\begin{pmatrix}
   N \\
   k
\end{pmatrix}\int_{\R^{2k}}\d p_{1}...\d p_{k}\int_{\R^{2(N-k}}\left \b \Big < F_{x_{1},p_{1}}(\cdot )\otimes...\otimes F_{x_{k},p_{k}}(\cdot ),\Psi_{N}(\cdot ,z)
\Big >\right \b^{2}\d z\\
=&k!\begin{pmatrix}
   N \\
   k
\end{pmatrix}\int_{\R^{2k}}\d p_{1}...\d p_{k}\int_{\R^{2(N-k)}}\left \b \F_{\hbar}\left [F_{x_{1},0}(\cdot )\otimes...\otimes F_{x_{k},0}(\cdot )\Psi_{N}(\cdot , z)\right ](p_{1},...,p_{k})\right \b^{2}\d z\\
=&k!\begin{pmatrix}
   N \\
   k
\end{pmatrix}\int_{\R^{2k}}\left \b F_{\hbar_{x}}(y_{1}-x_{1})...F_{\hbar_{x}}(y_{k}-x_{k})\Psi_{N}(y)\right \b^{2}\d y_{1}...\d y_{N}\\
=&k!\hbar^{2k}\rho^{(k)}_{\Psi_{N}}*\left (\b F_{\hbar_{x}}\b^{2}\right )^{\otimes k}(x_{1},...,x_{k})
\end{align*}
which gives $\eqref{marg_en_x}$. The proof of $\eqref{marg_en_t}$ is similar.
\end{proof}


\bibliographystyle{siam}

\begin{thebibliography}{10}

\bibitem{AdaTet-98}
{\sc R.~Adami and A.~Teta}, {\em On the {Aharonov-Bohm} effect.}, Lett.
  Math.Phys, 43 (1998), pp.~43--53.

\bibitem{AroSchWil-84}
{\sc S.~Arovas, J.~Schrieffer, and F.~Wilczek}, {\em Fractional statistics and
  the quantum {H}all effect}, Phys. Rev. Lett., 53 (1984), pp.~722--723.

\bibitem{AvrHerSim-78}
{\sc J.~Avron, I.~Herbst, and B.~Simon}, {\em Schr\"odinger operators with
  magnetic fields. i. general interactions}, Duke Math. J., 45 (1978),
  pp.~847--883.

\bibitem{BarGolGotMau-03}
{\sc C.~Bardos, F.~Golse, A.~D. Gottlieb, and N.~J. Mauser}, {\em Mean-field
  dynamics of fermions and the time-dependent {H}artree-{F}ock equation}, J.
  Math. Pures Appl. (9), 82 (2003), pp.~665--683.
  

\bibitem{BarEtalFev-20}
{\sc Bartolomei, H., Kumar, M., Bisognin, R., Marguerite, A., Berroir, J.-M.,
  Bocquillon, E., Pla{\c c}ais, B., Cavanna, A., Dong, Q., Gennser, U., Jin,
  Y., and F{\`e}ve, G.}
\newblock Fractional statistics in anyon collisions.
\newblock {\em Science 368}, 6487 (2020), 173--177.  
  

\bibitem{BenPorSch-14}
{\sc N.~Benedikter, M.~Porta, and B.~Schlein}, {\em Mean-field evolution of
  fermionic systems}, Communications in Mathematical Physics, 331 (2014),
  pp.~1--45.

\bibitem{BenPorSch-15}
{\sc N.~{Benedikter}, M.~{Porta}, and B.~{Schlein}}, {\em {Effective Evolution
  Equations from Quantum Dynamics}}, Springer Briefs in Mathematical Physics,
  Springer, 2016.

\bibitem{Bona-04}
{\sc M.~B\'ona}, {\em Combinatorics of permutations.}, Discrete mathematics ans
  its applications, Chapman \& Hall/CRC, 2004.

\bibitem{BouSor-92}
{\sc M.~Bourdeau and R.~Sorkin}, {\em When can identical particles collide ?},
  Phys. Rev. D, 45 (1992), pp.~687--696.

\bibitem{ChoLeeLee-92}
{\sc M.~Y. Choi, C.~Lee, and J.~Lee}, {\em Soluble many-body systems with
  flux-tube interactions in an arbitrary external magnetic field}, Phys. Rev.
  B, 46 (1992), pp.~1489--1497.

\bibitem{ClaEtalChi-18}
{\sc L.~W. Clark, B.~M. Anderson, L.~Feng, A.~Gaj, K.~Levin, and C.~Chin}, {\em
  Observation of density-dependent gauge fields in a bose-einstein condensate
  based on micromotion control in a shaken two-dimensional lattice}, Phys. Rev.
  Lett., 121 (2018), p.~030402.

\bibitem{ComRob-12}
{\sc M.~Combescure and D.~Robert}, {\em Coherent states and applications in
  mathematical physics}, Theoretical and Mathematical Physics, Springer,
  Dordrecht, 2012.

\bibitem{CorDubLunRou-19}
{\sc M.~Correggi, R.~Duboscq, N.~Rougerie, and D.~Lundholm}, {\em Vortex
  patterns in the almost-bosonic anyon gas}, Europhysics Letters, 126 (2019),
  p.~20005.

\bibitem{CorLunRou-16}
{\sc M.~Correggi, D.~Lundholm, and N.~Rougerie}, {\em Local density
  approximation for the almost-bosonic anyon gas}, Analysis and PDEs, 10
  (2017), pp.~1169--1200.

\bibitem{CorOdd-18}
{\sc M.~Correggi and L.~Oddis}, {\em Hamiltonians for two-anyon systems}, Rend.
  Mat. Appl. 39, 39 (2018), pp.~277--292.

\bibitem{DabSto-98}
{\sc L.~Dabrowski and P.~Stovicek}, {\em {Aharonov-Bohm} effect with
  $\delta$-type interaction}, J. Math. Phys, 39 (1998), pp.~47--62.

\bibitem{DiaFre-80}
{\sc P.~Diaconis and D.~Freedman}, {\em Finite exchangeable sequences}, Ann.
  Probab., 8 (1980), pp.~745--764.

\bibitem{EdmEtalOhb-13}
{\sc M.~J. Edmonds, M.~Valiente, G.~Juzeli\ifmmode~\bar{u}\else \={u}\fi{}nas,
  L.~Santos, and P.~\"Ohberg}, {\em Simulating an interacting gauge theory with
  ultracold bose gases}, Phys. Rev. Lett., 110 (2013), p.~085301.

\bibitem{ElgErdSchYau-04}
{\sc A.~Elgart, L.~Erd\H{o}s, B.~Schlein, and H.-T. Yau}, {\em Nonlinear
  {H}artree equation as the mean field limit of weakly coupled fermions}, J.
  Math. Pures Appl., 83 (2004), pp.~1241--1273.

\bibitem{FeiKasZad-14}
{\sc E.~Feinberg, P.~Kasyanov, and N.~Zadoianchuk}, {\em Fatou's lemma for
  weakly converging probabilities}, Theory of Probability \& Its Applications,
  58 (2014), pp.~683--689.

\bibitem{FouLewSol-15}
{\sc S.~Fournais, M.~Lewin, and J.-P. Solovej}, {\em The semi-classical limit
  of large fermionic systems}, Calculus of Variations and Partial Differential
  Equations, 57 (2018), p.~105.

\bibitem{FouMad-19}
{\sc S.~Fournais and P.~Madsen}, {\em Semi-classical limit of confined
  fermionic systems in homogeneous magnetic fields}, 2019.

\bibitem{FroKno-11}
{\sc J.~Fr{\"o}hlich and A.~Knowles}, {\em A microscopic derivation of the
  time-dependent {H}artree-{F}ock equation with {C}oulomb two-body
  interaction}, J. Stat. Phys., 145 (2011), pp.~23--50.

\bibitem{Girardot-19}
{\sc T.~Girardot}, {\em Average field approximation for almost bosonic anyons
  in a magnetic field}, Journal of Mathematical Physics, 61 (2020), p.~071901.

\bibitem{Goerbig-09}
{\sc M.~O. Goerbig}, {\em Quantum {H}all effects}.
\newblock arXiv:0909.1998, 2009.

\bibitem{GusSig-06}
{\sc S.~J. Gustafson and I.~M. Sigal}, {\em Mathematical Concepts of Quantum
  Mechanics}, Universitext, Springer, 2nd~ed., 2006.

\bibitem{HewSav-55}
{\sc E.~Hewitt and L.~J. Savage}, {\em Symmetric measures on {C}artesian
  products}, Trans. Amer. Math. Soc., 80 (1955), pp.~470--501.

\bibitem{Husimi-40}
{\sc K.~Husimi}, {\em Some formal properties of the density matrix}, Proc.
  Phys. Math. Soc. Japan, 22 (1940), p.~264.

\bibitem{Jain-07}
{\sc J.~K. Jain}, {\em {Composite fermions}}, Cambridge University Press, 2007.

\bibitem{LarLun-16}
{\sc S.~Larson and D.~Lundholm}, {\em {Exclusion bounds for extended anyons}},
  Archive for Rational Mechanics and Analysis, 227 (2018), pp.~309--365.

\bibitem{Laughlin-99}
{\sc R.~B. Laughlin}, {\em Nobel lecture: Fractional quantization}, Rev. Mod.
  Phys., 71 (1999), pp.~863--874.

\bibitem{LewMadTri-19}
{\sc M.~Lewin, P.~S. Madsen, and A.~Triay}, {\em Semi-classical limit of large
  fermionic systems at positive temperature}, J. Math. Phys., 60 (2019),
  p.~091901.

\bibitem{Lieb-68}
{\sc E.~Lieb}, {\em {Concavity properties and a generating function for
  Stirling numbers}}, J. Comb. Theory, 5 (1968), pp.~203--206.

\bibitem{Lieb-81b}
{\sc E.~H. Lieb}, {\em {Thomas-{F}ermi and related theories of atoms and
  molecules}}, Rev. Mod. Phys., 53 (1981), pp.~603--641.

\bibitem{LieLos-01}
{\sc E.~H. Lieb and M.~Loss}, {\em Analysis}, vol.~14 of Graduate Studies in
  Mathematics, American Mathematical Society, Providence, RI, 2nd~ed., 2001.

\bibitem{LieSei-09}
{\sc E.~H. Lieb and R.~Seiringer}, {\em The {S}tability of {M}atter in
  {Q}uantum {M}echanics}, Cambridge Univ. Press, 2010.

\bibitem{LieSim-77b}
{\sc E.~H. Lieb and B.~Simon}, {\em The {T}homas-{F}ermi theory of atoms,
  molecules and solids}, Adv. Math., 23 (1977), pp.~22--116.

\bibitem{LieSolYng-94b}
{\sc E.~H. Lieb, J.-P. Solovej, and J.~Yngvason}, {\em Asymptotics of heavy
  atoms in high magnetic fields: {II}. {S}emi-classical regions}, Comm. Math.
  Phys, 161 (1994), pp.~77--124.

\bibitem{LieSolYng-95}
\leavevmode\vrule height 2pt depth -1.6pt width 23pt, {\em Ground states of
  large quantum dots in magnetic fields}, Phys. Rev. B, 51 (1995),
  pp.~10646--10665.

\bibitem{LieThi-75}
{\sc E.~H. Lieb and W.~E. Thirring}, {\em Bound on kinetic energy of fermions
  which proves stability of matter}, Phys. Rev. Lett., 35 (1975), pp.~687--689.

\bibitem{LieThi-76}
\leavevmode\vrule height 2pt depth -1.6pt width 23pt, {\em Inequalities for the
  moments of the eigenvalues of the {S}chr{\"o}dinger {H}amiltonian and their
  relation to {S}obolev inequalities}, Studies in Mathematical Physics,
  Princeton University Press, 1976, pp.~269--303.

\bibitem{LunRou-15}
{\sc D.~Lundholm and N.~Rougerie}, {\em {The average field approximation for
  almost bosonic extended anyons}}, J. Stat. Phys., 161 (2015), pp.~1236--1267.

\bibitem{LunRou-16}
\leavevmode\vrule height 2pt depth -1.6pt width 23pt, {\em {Emergence of
  fractional statistics for tracer particles in a Laughlin liquid}}, Phys. Rev.
  Lett., 116 (2016), p.~170401.

\bibitem{LunSol-13a}
{\sc D.~Lundholm and J.~P. Solovej}, {\em {Hardy and Lieb-Thirring inequalities
  for anyons}}, Comm. Math. Phys., 322 (2013), pp.~883--908.

\bibitem{LunSol-13b}
\leavevmode\vrule height 2pt depth -1.6pt width 23pt, {\em Local exclusion
  principle for identical particles obeying intermediate and fractional
  statistics}, Phys. Rev. A, 88 (2013), p.~062106.

\bibitem{LunSol-14}
\leavevmode\vrule height 2pt depth -1.6pt width 23pt, {\em {Local exclusion and
  Lieb-Thirring inequalities for intermediate and fractional statistics}}, Ann.
  Henri Poincar\'e, 15 (2014), pp.~1061--1107.

\bibitem{LunSei-18}
\leavevmode\vrule height 2pt depth -1.6pt width 23pt, {\em {Fermionic behavior
  of ideal anyons}}, Lett. Math. Phys., 108 (2018), pp.~2523--2541.

\bibitem{Mashkevich-96}
{\sc S.~Mashkevich}, {\em {Finite-size anyons and perturbation theory}}, Phys.
  Rev. D, 54 (1996), pp.~6537--6543.

\bibitem{Menon-73}
{\sc V.~Menon}, {\em {On the maximum of Stirling numbers of the second kind}},
  J. Comb. Theory (A), 15 (1973), pp.~11--24.

\bibitem{Molinari-17}
{\sc L.~Molinari}, {\em Notes on wick's theorem in many-body theory}.
\newblock arXiv:1710.09248, 2017.


\bibitem{NakEtalMan-20}
{\sc Nakamura, J., Liang, S., Gardner, G.~C., and Manfra, M.~J.}
\newblock Direct observation of anyonic braiding statistics.
\newblock {\em Nature Physics 16\/} (2020), 931--936.


\bibitem{ReeSim2}
{\sc M.~Reed and B.~Simon}, {\em Methods of {M}odern {M}athematical {P}hysics.
  {II}. {F}ourier analysis, self-adjointness}, Academic Press, New York, 1975.

\bibitem{RenDob-69}
{\sc B.~Rennie and A.~Dobson}, {\em {On Stirling numbers of the second kind}},
  J. Comb. Theory, 7 (1969), pp.~116--121.

\bibitem{Rougerie-LMU}
{\sc N.~Rougerie}, {\em {De Finetti theorems, mean-field limits and
  Bose-Einstein condensation}}.
\newblock arXiv:1506.05263, 2014.
\newblock LMU lecture notes.

\bibitem{Rougerie-hdr}
\leavevmode\vrule height 2pt depth -1.6pt width 23pt, {\em Some contributions
  to many-body quantum mathematics}.
\newblock arXiv:1607.03833, 2016.
\newblock habilitation thesis.

\bibitem{Rougerie-spartacus}
\leavevmode\vrule height 2pt depth -1.6pt width 23pt, {\em Th{\'e}or{\`e}mes de
  De Finetti, limites de champ moyen et condensation de Bose-Einstein}, Les
  cours Peccot, Spartacus IDH, Paris, 2016.
\newblock Cours Peccot, Coll{\`e}ge de France : f{\'e}vrier-mars 2014.

\bibitem{Rougerie-EMS}
\leavevmode\vrule height 2pt depth -1.6pt width 23pt, {\em {Scaling limits of
  bosonic ground states, from many-body to nonlinear Schr\"odinger}}.
\newblock arXiv:2002.02678, 2020.

\bibitem{Schatten-60}
{\sc R.~Schatten}, {\em Norm Ideals of Completely Continuous Operators}, vol.~2
  of Ergebnisse der Mathematik und ihrer Grenzgebiete, Folge, 1960.

\bibitem{Solovej-notes}
{\sc J.~P. Solovej}, {\em Many body quantum mechanics}.
\newblock Web: http://web.math.ku.dk/~solovej/MANYBODY, Course Homepage for Many-body Quantum
 Physics, ESI 2014.


\bibitem{Takahashi-86}
{\sc K.~Takahashi}, {\em Wigner and husimi functions in quantum mechanics}, J.
  Phys. Soc. Japan, 55 (1986), pp.~762--779.

\bibitem{Thirring-81}
{\sc W.~Thirring}, {\em A lower bound with the best possible constant for
  {C}oulomb hamiltonians}, Comm. Math. Phys., 79 (1981), pp.~1--7.

\bibitem{Trugenberger-92b}
{\sc C.~Trugenberger}, {\em {Ground state and collective excitations of
  extended anyons}}, Phys. Lett. B, 288 (1992), pp.~121--128.

\bibitem{ValWesOhb-20}
{\sc G.~Valent\'{\i}-Rojas, N.~Westerberg, and P.~\"Ohberg}, {\em Synthetic
  flux attachment}, Phys. Rev. Research, 2 (2020), p.~033453.

\bibitem{Villani-08}
{\sc C.~Villani}, {\em Optimal transport, old and new.}, vol.~108 of
  Grundlehren der mathematischen Wissenschaften, Springer, 2008.

\bibitem{YakEtal-19}
{\sc E.~Yakaboylu, A.~Ghazaryan, D.~Lundholm, N.~Rougerie, M.~Lemeshko, and
  R.~Seiringer}, {\em Quantum impurity model for anyons}, Phys. Rev. B, 102
  (2020), p.~144109.

\bibitem{YakLem-18}
{\sc E.~Yakaboylu and M.~Lemeshko}, {\em Anyonic statistics of quantum
  impurities in two dimensions}, Phys. Rev. B, 98 (2018), p.~045402.
  
  \bibitem{Bar-20}
{\sc H.~Bartolomei, M.~Kumar, R.~Bisognin,  A.~Marguerite  J.-M.~Berroir,  E.~Bocquillon,  B.~Pla\c{c}ais,   A.~Cavanna,    Q.~Dong,  U.~Gennser, et al.}, {\em Fractional statistics in anyon collisions}, Science , (2020), p.~173177
  
  \bibitem{nak-20}
{\sc J.~ Nakamura  S.~Liang  G.-C.~Gardner M.-J.~Manfra}, {\em Direct observation of anyonic braiding statistics at the $\nu$=1/3 fractional quantum Hall state}, Nature Physics Journal Volume 16, (2020)
  
  

    



\end{thebibliography}

%

%

\end{document}